\documentclass[12pt]{article}

\newlength{\abstractwidth}
\usepackage{epsf}
\usepackage{color}
\usepackage{graphicx}
\usepackage{hyperref}
\usepackage{dsfont}

\usepackage{amsmath, nccmath}
\usepackage{amsmath, amsthm}
\usepackage{amssymb}
\usepackage{amsfonts}
\usepackage{latexsym}
\usepackage{mathtools}
\usepackage{shuffle}

\usepackage{tikz, pgf}
\usepackage{tkz-fct}
\usepackage{pgfplots}
\usetikzlibrary{shapes.misc}
\usetikzlibrary{shapes,snakes}
\usetikzlibrary{decorations.pathmorphing}	
\usetikzlibrary{decorations.markings}

\usetikzlibrary{arrows,shapes,positioning}
\usetikzlibrary{decorations.markings}
\usepackage[rightcaption]{sidecap}
\tikzstyle arrowstyle=[scale=1]
\tikzstyle directed=[postaction={decorate,decoration={markings,
    mark=at position .65 with {\arrow[arrowstyle]{stealth}}}}]
\tikzstyle reverse directed=[postaction={decorate,decoration={markings,
    mark=at position .65 with {\arrowreversed[arrowstyle]{stealth};}}}]
\usetikzlibrary{positioning}

\renewcommand{\thefootnote}{\fnsymbol{footnote}}
\renewcommand{\thanks}[1]{\footnote{#1}}
\newcommand{\starttext}{
\setcounter{footnote}{0}
\renewcommand{\thefootnote}{\arabic{footnote}}}

\numberwithin{equation}{section}
\newcommand{\bea}{\begin{eqnarray}}
\newcommand{\eea}{\end{eqnarray}}
\newcommand{\be}{\begin{eqnarray}}
\newcommand{\ee}{\end{eqnarray}}
\newcommand{\<}{\langle}
\renewcommand{\>}{\rangle}
\newcommand{\bma}{\begin{matrix}}
\newcommand{\ema}{\end{matrix}}

\def\cG{{\cal G}}
\def\cH{{\cal H}}

\def\cJ{{\cal J}}
\def\cK{{\cal K}}
\def\cL{{\cal L}}
\def\cM{{\cal M}}

\def\cO{{\cal O}}
\def\cP{{\cal P}}
\def\cQ{{\cal Q}}
\def\cR{{\cal R}}
\def\cS{{\cal S}}
\def\cT{{\cal T}}
\def\cU{{\cal U}}
\def\cV{{\cal V}}
\def\cW{{\cal W}}

\def\mA{\mathfrak{A}}
\def\mB{\mathfrak{B}}
\def\mC{\mathfrak{C}}

\def\mH{\mathfrak{H}}

\def\mJ{\mathfrak{J}}

\def\mg{\mathfrak{g}}

\def\muu{\mathfrak{u}}

\def\mw{\mathfrak{w}}

\def\CC{{\mathbb C}}

\def\ZZ{{\mathbb Z}}

\def\Re{{\rm Re \,}}
\def\Im{{\rm Im \,}}

\def\half{{1\over 2}}
\def\thalf{\tfrac{1}{2}}
\def\p{\partial}

\def\a{\alpha}

\def\ep{\varepsilon}
\def\om{\omega}

\def\ta{{\hat a}}
\def\tb{{\hat b}}
\def\ti{{t}}

\def\bom{{ \, \overline \omega}}
\def\Sp{{\text{Sp}}}
\def\GL{{\text{GL}}}
\def\bGam{{\boldsymbol{\Gamma}}}
\def\MM{\cM}
\def\LL{\cL}
\def\TT{\cT}

\def\xih{{\hat \xi}}
\def\etah{{\hat \eta}}
\def\PPer{{\cal Z}}
\def\notc{{\cal Z}}

\def\gDHS{\bGam_{\rm DHS}}
\def\gminus{\bGam_{-}}
\def\gE{\bGam_{\rm E}}
\def\ach{\check{a}}
\def\bch{\check{b}}
\def\xich{\check{\xi}}
\def\etach{\check{\eta}}

\newcommand{\tGamma}[4]{\tilde \Gamma\Big(   \smallmatrix#1 \\  #2 \\ #3
 \endsmallmatrix ;#4\Big)}

\def\no{\nonumber}
\def\sm{\smallskip}

\definecolor{Cyan}{cmyk}{1.,0,0,0}
\definecolor{Magenta}{cmyk}{0,1.,0,0}
\definecolor{Yellow}{cmyk}{0,0,1.,0}
\definecolor{White}{cmyk}{0,0,0,0}
\definecolor{Orange}{cmyk}{0,0.61,0.87,0}
\definecolor{RedOrange}{cmyk}{0,0.77,0.87,0}
\definecolor{Red}{cmyk}{0,1.,1.,0}
\definecolor{Purple}{cmyk}{0.45,0.86,0,0}
\definecolor{Violet}{cmyk}{0.79,0.88,0,0}
\definecolor{Blue}{cmyk}{1,0.5,0,0}
\definecolor{ProcessBlue}{cmyk}{0.96,0,0,0}
\definecolor{GreenYellow}{cmyk}{0.6,0,1.,0}
\definecolor{Black}{cmyk}{0,0,0,1}
\definecolor{dgreen}{rgb}{0,0.70,0.30}
\newtheorem{thm}{Theorem}[section]
\newtheorem{lem}[thm]{Lemma}
\newtheorem{prop}[thm]{Proposition}
\newtheorem{cor}[thm]{Corollary}

\newtheorem{deff}[thm]{Definition}
\newtheorem{rmk}[thm]{Remark}

\def\qq{y}
\def\XX{{\cal X}}
\def\YY{{\cal Y}}
\def\tp{t'}

\newif\ifnote 
\notetrue

\begin{document}
\starttext
\setcounter{footnote}{0}

\begin{flushright}
2025 December 10
 \\
 UUITP-01/25
\end{flushright}

\vskip 0.3in

\begin{center}

{\Large \bf Relating flat connections and polylogarithms}

\vskip 0.1in

{\Large \bf  on higher genus Riemann surfaces}

\vskip 0.2in

{\large Eric D'Hoker${}^{a}$, Benjamin Enriquez${}^{b}$, Oliver Schlotterer${}^{c}$, Federico Zerbini${}^{d}$} 

\vskip 0.15in

{ \sl ${}^{a}$Mani L. Bhaumik Institute for Theoretical Physics}\\
{\sl  Department of Physics and Astronomy}\\
{\sl University of California, Los Angeles, CA 90095, USA}

\vskip 0.1in 

{\sl ${}^b$IRMA and Universit\'e de Strasbourg}\\
{\sl  7, Rue Descartes, 67084 Strasbourg, France}

\vskip 0.1in

{\sl ${}^c$Department of Physics and Astronomy,} \\
  { \sl Department of Mathematics,} \\
  { \sl Centre for Geometry and Physics,} \\ 
  {\sl Uppsala University, 75120 Uppsala, Sweden}
 
 \vskip 0.1in

{\sl ${}^{d}$Departamento de Matem\'aticas Fundamentales, UNED}\\
{\sl Calle de Juan del Rosal, 28040 Madrid, Spain}

\vskip 0.15in 

{\tt \small dhoker@physics.ucla.edu, enriquez@math.unistra.fr, oliver.schlotterer@physics.uu.se, f.zerbini@mat.uned.es}

\vskip 0.2in

\begin{abstract}
\vskip 0.1in
In this work, we relate two recent constructions that generalize classical (genus-zero) polylogarithms to higher-genus Riemann surfaces. A flat connection valued in a freely generated Lie algebra on a punctured Riemann surface of arbitrary genus produces an infinite family of homotopy-invariant iterated integrals associated to all possible words in the alphabet of the Lie algebra generators. Each iterated integral associated to a word is a higher-genus polylogarithm.  Different flat connections taking values in the same Lie algebra on a given Riemann surface may be related to one another by the composition of a gauge transformation and an automorphism of the Lie algebra, thus producing closely related families of polylogarithms. In this paper we provide two methods, which are inverses of one another,  to explicitly relate in this way the meromorphic multiple-valued connection introduced by Enriquez in e-Print 1112.0864 and the non-meromorphic single-valued and modular-invariant connection introduced by D'Hoker, Hidding and Schlotterer,  in e-Print 2306.08644.

\end{abstract}
\end{center}

\newpage

\setcounter{tocdepth}{2} 
\tableofcontents

\baselineskip=15pt
\setcounter{equation}{0}
\setcounter{footnote}{0}

\newpage

\section{Introduction}
\setcounter{equation}{0}
\label{sec:Intro}

Perturbative computations in quantum field theory and string theory involve complicated multidimensional integrals. It has become increasingly clear over the past few decades that these integrals may profitably be organized in terms of \textit{polylogarithms}, and various generalizations thereof. Broadly defined, the term polylogarithm is being used here for a \emph{multiple-valued function} on a manifold~$M$ defined as a \emph{homotopy-invariant iterated integral} over an integration path $\gamma:[0,1]\to M$ starting at some fixed integration base-point $\gamma(0)=y$ and depending on the integration end-point $\gamma(1)=x$. By homotopy-invariant iterated integral we mean a suitable linear combination~\cite{Chen} of iterated integrals along the path~$\gamma$ over differential 1-forms $\phi_1,\cdots,\phi_k$,
\bea
\int_{0}^1\phi_1\big (\gamma(t_1) \big )\int_{0}^{t_1}\phi_2 \big ( \gamma(t_2) \big )\cdots \int_{0}^{t_{k-1}}\phi_k \big ( \gamma(t_k) \big ),
\eea
whose value depends on the base-point~$y$, on the end-point~$x$ and on the 
homotopy class~$[\gamma]$ of the integration path~$\gamma$, but not upon the specific 
path~$\gamma$ in a given homotopy class~$[\gamma]$.

\sm

Polylogarithms on the punctured sphere, also known as \emph{hyperlogarithms} \cite{Lappo:1953, Goncharov:1995, Brown:2009qja, Panzer:2015ida}, complete the space of rational functions to a space of multiple-valued functions that is closed under addition, multiplication, differentiation and the taking of primitives. Polylogarithms were generalized to the elliptic case \cite{Levin:1997, Levin:2007, BrownLevin, Broedel:2014vla}, where they complete the space of elliptic functions to a function space on the punctured torus with the same properties as its genus-zero analogue \cite{Broedel:2017kkb, Enriquez:2023}.
When evaluated at special points, genus-zero polylogarithms produce multiple zeta values while their genus-one counterparts produce elliptic multiple zeta values \cite{Enriquez:Emzv, Broedel:2015hia, Matthes:thesis, Zerbini:thesis} (see also section \ref{sec:further} of this introduction).  Recently, different further generalizations of polylogarithms to Riemann surfaces of arbitrary genus have been proposed  with the goal of providing spaces of functions that are closed under the taking of primitives \cite{Enriquez:2022, DHS:2023, Baune:2024}. There is good evidence that these higher genus polylogarithms may provide a useful organizational set-up for perturbative string theory calculations \cite{DHoker:2017pvk, DHoker:2018mys, DHoker:2020tcq, DHoker:2023khh} while also promising relevance to higher loop Feynman integrals in quantum field theory \cite{Huang:2013kh, Georgoudis:2015hca, Doran:2023yzu, Marzucca:2023gto, Jockers:2024tpc, Duhr:2024uid}.

\sm

A general and efficient geometrical construction of polylogarithms is by taking the path-ordered exponential solution of the differential equation induced by a flat connection valued in a (completed) free Lie algebra. Indeed, the coefficient of the path-ordered exponential corresponding to a given word in the Lie algebra generators is an iterated integral which is homotopy invariant, due to the flatness of the connection. 
At genus zero classical polylogarithms arise from the Knizhnik--Zamolodchikov connection, whose holonomy is closely related to the Drinfeld associator \cite{Drinf1, Drinf2};
 the significance of the latter to number theory has been discussed in \cite{Le:Murakami, Terasoma:Selberg, Enriquez:Emzv}, whereas its relevance to string amplitudes was proposed in \cite{Drummond:2013vz, Broedel:2013aza, Kaderli:2019dny, Baune:2024uwj}.

\sm

The approaches to generalize polylogarithms to arbitrary punctured Riemann surfaces of genus $h\geq 1$ in terms of flat connections can be divided into three different categories. Chronologically, the first is via a holomorphic multiple-valued\footnote{The expression ``multiple-valued connection'' (resp.\ ``single-valued connection'') is an abuse of terminology, and expresses here the fact that, if a connection $\nabla$ is written as $d-\mathcal J$, then $\mathcal J$ is a multiple-valued (resp.\ single-valued) differential 1-form. } connection with a regular singularity at the puncture: the connection $d-\cK_\text{E}$ introduced by Enriquez~\cite{Enriquez:2011}. The second is via a holomorphic single-valued connection with an irregular singularity at the puncture, which is the case of the families of connections introduced by Enriquez and Zerbini in \cite{Enriquez:2021, Enriquez:2022}. The third is via a non-holomorphic single-valued and modular-invariant connection with a regular singularity at the puncture: the connection $d-\cJ_\text{DHS}$  introduced by D'Hoker, Hidding, and Schlotterer (DHS) in \cite{DHS:2023}. All these connections take values in the (completed) freely generated Lie algebra $\mg$ on $2h$ generators and give rise to different, but closely related, families of polylogarithms. The relation between the first two approaches was partly discussed in \cite{Enriquez:2021} and will be the subject of the forthcoming article \cite{Enriquez:next}.

\sm

In the present paper, we shall relate the first and the third approaches by showing that $d-\cK_\text{E}$ can be obtained from $d-\cJ_\text{DHS}$ by combining a gauge transformation and an automorphism of the Lie algebra $\mg$. More precisely, using different methods, we will construct two such relations between these connections, and show that the two corresponding pairs of an automorphism and a gauge transformation are inverses of one another. Combining the two constructions we will obtain a relation between the function spaces generated by the two corresponding families of polylogarithms. We shall provide explicit formulas for the relations between connections and generating functions for polylogarithms, and evaluate the general formulas to low orders.

\subsection{Flat connections, iterated integrals, and polylogarithms}\label{sec:1.1}

An efficient geometric construction of polylogarithms starts from a flat connection $d_x-\cJ(x;c)$, with $\cJ(x;c)$ a  differential 1-form in $x$ on a Riemann surface~$\Sigma$, possibly multiple-valued, taking values in the completion\footnote{If one assigns degree~1 to each variable $c_i$, the Lie algebra freely generated by $c$ is a graded Lie algebra (of Lie polynomials); its completion with respect to this grading is the Lie algebra $\mg$ of Lie series.} ~$\mg$ of the Lie algebra that is freely generated by a set $c=\{c_1, \cdots, c_n\}$. Flatness of the connection, which is expressed in terms of the Maurer--Cartan equation
\bea
\label{1.MC}
d_x \cJ(x;c) - \cJ(x;c) \wedge \cJ(x;c) =0,
\eea
guarantees the integrability of  the differential equation
\bea
d_x \bGam(x,y;c) = \cJ(x;c)  \, \bGam (x,y;c),
\label{diffgamma}
\eea 
subject to the initial condition $\bGam(y,y;c)=1$, for a function $\bGam (x,y;c)$ that is a scalar in  $x$ and $y$ and which takes values in the Lie group\footnote{\label{label3}This is defined to be the group of group-like elements in the (Hopf) algebra $\mathbb C\langle\!\langle c\rangle\!\rangle$ of formal series in non-commutative generators $c_1, \cdots, c_n$. The latter is the free associative algebra generated by $c$ equipped with the Hopf algebra structure associated with word concatenation, and should be interpreted here as the degree completion of the universal enveloping algebra of the free Lie algebra generated by $c$.} $\exp(\mg)$ of $\mg$. The solution of (\ref{diffgamma}) can be written in terms of the path-ordered exponential of $\cJ$,\footnote{It will often be convenient to indicate the variable over which a given integration is being carried out, especially so when several integrations are involved; we shall reserve the letter $\ti$ for this purpose.}
\bea
\label{1.PE}
\bGam (x,y;c) = \text{P} \exp \int _y ^x \cJ(\ti;c).
\eea
Taylor expanding the path-ordered exponential in powers of $\cJ$ gives an explicit expression for  $\bGam(x,y;c)$ in terms of a 
series of iterated integrals, each of which takes values  in $\mathbb C\langle\!\langle c\rangle\!\rangle$ (see footnote \ref{label3}),
\bea
\label{1.exp}
\bGam (x,y;c) = 1 + \sum_{k=1}^\infty \int _y ^x \cJ(\ti_1;c) \int _y ^{\ti_1} \cJ(t_2;c) \cdots \int _y^{\ti_{k-1}} \cJ(\ti_k;c).
\eea
Flatness of $\cJ(x;c)$ further guarantees that $\bGam (x,y;c)$ is homotopy invariant, namely,  it depends\footnote{The dependence on the path is suppressed from the notation, because we will rather consider $x$ and $y$ as variables on the universal cover~$\tilde\Sigma$ of~$\Sigma$, which is equivalent to specifying the class of the path between points $x$ and $y$ on $\Sigma$.} only on the homotopy class of the integration path from $y$ to $x$ and is independent of the specific representative in the class. Also, $\bGam$ satisfies the path-concatenation formula
\bea
\label{1.comp}
\bGam (x,z;c) = \bGam (x,y;c)\, \bGam (y,z;c),
\eea
where the product on the right is taken in the group $\exp(\mg)$ and concatenates the words in $c$ from the two factors. This property can be used to deduce useful formulas for the monodromy of the solution, namely for the value of $\bGam (\gamma \cdot x,y;c)$ in terms of $\bGam (x,y;c)$, with $\gamma$ an element of the fundamental group\footnote{\label{ftbasept} Here and elsewhere we are denoting by $\qq$ both a chosen integration base-point in the universal cover~$\tilde\Sigma$ and its image in~$\Sigma$. The notation $\gamma \cdot \qq$, inspired by the induced isomorphism $\pi_1(\Sigma,\qq)\simeq \mathrm{Aut}(\tilde\Sigma/\Sigma)$, stands then for the endpoint of the (unique) lift to~$\tilde\Sigma$ of the path $\gamma$ which starts at~$\qq$. There is then a unique element $\gamma$ of $\mathrm{Aut}(\tilde\Sigma/\Sigma)$ which takes $\qq$ to $\gamma \cdot \qq$, and $\gamma \cdot x$ denotes the image of $x$ under this automorphism.} $\pi_1(\Sigma,\qq)$. Indeed, using \eqref{1.comp} to write $\bGam (\gamma \cdot x,y;c) = \bGam (\gamma \cdot x, \gamma \cdot y;c)\, \bGam (\gamma \cdot y,y;c)$, setting $\mu(\gamma,y;c)=\bGam( \gamma \cdot y , y;c)$, and assuming that $\cJ(x;c)$ is single-valued, we obtain the following formula 
\bea
\label{1.eqmon}
\bGam (\gamma \cdot x,y;c) =\bGam (x,y;c)\,\mu(\gamma,y;c).
\eea 
For fixed~$y$ and~$c$, and assuming that $\cJ(x;c)$ is single-valued, one verifies using \eqref{1.comp} that~$\mu$ is a homomorphism from $\pi_1(\Sigma,\qq)$ to the group $\exp(\mg)$, namely that\footnote{Here $\gamma_1  \!\star\! \gamma _2$ stands for the composition of two paths $\gamma _1, \gamma _2 \in  \pi_1 (\Sigma,\qq)$ with the convention that the composed path traverses first $\gamma_1$ followed by $\gamma_2$.
This implies that $(\gamma_1  \!\star\! \gamma _2) \! \cdot  \! y = \gamma_2 \! \cdot \! (\gamma_1\! \cdot \! y)$
and therefore $\mu (\gamma_1 \!\star\! \gamma _2,y;c)
= \bGam ( \gamma_2\!\cdot \! (\gamma_1 \!  \cdot  \! y),y;c)
= \bGam ( \gamma_2\!\cdot \! (\gamma_1 \!  \cdot  \! y), \gamma_2\!\cdot \! y;c) \bGam ( \gamma_2\!\cdot \!y,y;c)$ which
then leads to the concatenation order on the left side of (\ref{1.eqcomplawmu})
since $ \bGam ( \gamma_2\!\cdot \! (\gamma_1 \!  \cdot  \! y), \gamma_2\!\cdot \! y;c)
=  \bGam ( \gamma_1 \!  \cdot  \! y,  y;c)$
and $\bGam ( \gamma_j \!  \cdot  \! y,  y;c) =  \mu (\gamma _j,y;c)$.}
\bea\label{1.eqcomplawmu}
 \mu(\gamma_1,y;c) \mu (\gamma _2,y;c) = \mu (\gamma_1 \!\star\! \gamma _2,y;c),
 \eea
 thus yielding a monodromy representation.

\sm

While $\bGam$ is homotopy invariant, the contribution to the series in (\ref{1.exp}) from a single value of $k \geq 1$ is not homotopy invariant.  Thus, it is understood that the iterated integrals for all values of $k$ are taken along the same path in a given homotopy class. The expansion in powers of $\cJ$ may be expressed as a series over \textit{words} $\mw \in \mathcal W(c)$, 
\bea
\label{1.words}
\bGam (x,y;c) =  \sum_{\mw \in \mathcal W(c)} \mw \, \Gamma  (\mw; x,y),
\eea
where $\mathcal W(c)$ is the set of all words in the alphabet of letters $\{ c_1, \cdots, c_n\}$, which is the \textit{monoid} freely generated by $c$, where the sum in (\ref{1.words}) includes the empty word $\emptyset$ for which $\Gamma(\emptyset; x,y)=1$. For each non-empty word $\mw$, the coefficient $\Gamma (\mw;x,y)$ is a homotopy-invariant iterated integral, to which we shall generally refer as a \emph{polylogarithm associated with the connection~$\cJ$}. It follows from the fact that $\bGam (x,y;c)$ takes values in $\exp(\mg)$ that the product of polylogarithms for words $\mw_1$ and $\mw_2$ and identical endpoints $x,y$ may be expressed as a sum of polylogarithms associated with words $\mw$ belonging to the shuffle product\footnote{\label{ft:shuf}We recall that the shuffle product $\mw_1 \shuffle \, \mw_2$ is the subset of $\mathcal W(c)$ containing the words $\mw$ obtained from all possible ways of interlacing the letters of $\mw_1$ and $\mw_2$ such that the order of the letters in each word is preserved, see appendix \ref{sec:shprf}. The shuffle product extends to a commutative and associative operation which turns the $\mathbb C$-vector space of linear combination of words $\mw\in \mathcal W(c)$ into a ring, with neutral element given by the empty word $\emptyset$. A useful reference on free Lie algebras and the shuffle product is  \cite{Reutenauer}.} $\mw_1 \shuffle \, \mw_2$ (see \cite{Brown:UNP} and \cite{Reutenauer} for proofs and further discussions). 
\bea
\label{1.shuffle}
\Gamma(\mw_1;x,y) \Gamma(\mw_2;x,y) = \sum _{\mw \in \mw_1 \shuffle \, \mw_2} \Gamma (\mw; x,y).
\eea

\subsection{Relating flat connections and polylogarithms}\label{sec:1.2}

We shall from now on consider connections on a once-punctured surface $\Sigma _p = \Sigma \smallsetminus \{ p \}$, where~$\Sigma $ is a compact Riemann surface of arbitrary genus $h \geq 1$ (not necessarily hyperelliptic as those higher-genus surfaces encountered in the current particle-physics literature \cite{Huang:2013kh, Georgoudis:2015hca, Doran:2023yzu, Marzucca:2023gto, Jockers:2024tpc, Duhr:2024uid}). The fundamental group $\pi_1(\Sigma_p,\qq)$, with base-point $\qq \in \Sigma_p$, is freely generated by closed loops $\mA^I, \mB_I$ for $I=1,\cdots, h$ which can be chosen such that the intersection pairing $\mJ$ of their images (denoted with the same symbol) in the homology group $H_1(\Sigma_p,\mathbb Z)$ satisfies $\mJ(\mA^I, \mA^J)=\mJ(\mB_I, \mB_J)=0$ and $\mJ(\mA^I, \mB_J)=\delta ^I_J$. The completed free Lie algebra $\mg$ in which the connections take their values will be chosen to be generated by $2h$ independent elements. We shall denote a choice of such generators of $\mg$ by $a \cup b$, where $a = \{  a^1 , \cdots , a^h\}$ and $b=\{ b_1, \cdots, b_h \}$; these are elements of $\mg^h$ and can be interpreted as a basis of the dual of $H^1_{\mathrm{dR}}(\Sigma_p)$, corresponding to the cycles $\mA^1, \cdots, \mA^h, \mB_1 , \cdots , \mB_h$ via the isomorphism induced by the period pairing (see \cite{Enriquez:2022}). 

\sm

In section \ref{sec:new5} (see Theorem \ref{sec04}(d)), we will introduce the notion of flat connection over a principal $\exp(\mathfrak{g})$-bundle over $\Sigma_p$. We define the monodromy of such a flat connection and we show that any two such connections ${\cal J}$ and ${\cal K}$ can be related by the combination of a gauge transformation and an automorphism of $\mathfrak{g}$, provided that their monodromy representation $\mu_{\cal J},\mu_{\cal K}$ are such that the families $([\log (\mu_\bullet(\mA^1))]_1,\cdots ,[\log (\mu_\bullet(\mA^h))]_1, [\log (\mu_\bullet(\mB_1))]_1,$ $\cdots ,[\log (\mu_\bullet(\mB_h))]_1)$ (where $[\cdot]_1$ is the projection to the degree-one part of~$\mg$) for $\bullet = \cJ,\cK$ are bases of the vector space generated by $a\cup b$.

\sm

In the remainder of this paper we shall consider two specific connections on~$\Sigma_p$, both taking values in~$\mg$. The first is a holomorphic multiple-valued connection $d_x-\cK_\text{E}(x,p;a,b)$ given by specializing a more general construction from \cite{Enriquez:2011}, while the second is the non-holomorphic, single-valued and modular-invariant connection $d_x-\cJ_\text{DHS}(x,p;a,b)$ introduced  in \cite{DHS:2023}.  
Their definition and properties will be reviewed in section \ref{sec:2} below. More precisely, we will uniquely characterize the $\mg$-valued differentials $\cK_\text{E}$ and $\cJ_\text{DHS}$ through their functional properties in Theorems \ref{2.thm:1} and \ref{2.thm:2}, respectively. Generalizations of both connections to acquire simple poles in $x$ at an arbitrary number of punctures may be found in \cite{Enriquez:2011, DHS:2023}.

\sm

Since both connections $d_x-\cK_\text{E}(x,p;a,b)$ and  $d_x-\cJ_\text{DHS}(x,p;a,b)$ satisfy the 
assumptions of Theorem \ref{sec04}(d), it follows that~$d-\cJ_\text{DHS}$ and~$d-\cK_{\rm E}$ must be related by combining a gauge transformation with a Lie algebra automorphism of~$\mg$. The key result of this paper is to provide two different methods for the explicit construction of the gauge transformation and of the automorphism which relate the two connections. In both constructions, the gauge transformation will be a smooth multiple-valued function on~$\Sigma \times \Sigma$ which takes values in $\exp(\mg_b)$, where~$\mg_b$ is the Lie sub-algebra of $\mg$ given by the (completed) free Lie algebra generated by~$b$.
 
 \sm
 
In section \ref{sec:new5}, we will show that these combinations of a gauge transformation and a Lie algebra automorphism are part of a group acting on the set of flat connections. We then prove that said combinations are mutually inverse elements of this group action. By separately constructing both pairs of gauge transformation and Lie algebra automorphism, we obtain complementary information on the translation between the polylogarithm functions constructed from $d-\cJ_\text{DHS}$ and~$d-\cK_{\rm E}$: first, structural relations between the algebras of functions as detailed in section \ref{sec:5}, and second, convenient starting points for
practical calculations where each formulation of polylogarithms and their integration kernels is directly expressed in terms of ingredients of the other.

\subsubsection{First construction: $\cK_\text{E}$ from $ \cJ _\text{DHS}$}
\label{subsub:1.2.1}

More specifically, section \ref{sec:3} is devoted to the construction (see Theorem \ref{3.thm:1})  of a gauge transformation $ \cU_{\rm DHS}(x,p)$, which will be a smooth function in $x$ on the universal cover of $\Sigma_p$,  and an automorphism mapping $a \cup b$ to an alternative set of generators $ \ta \cup \tb$ of $\mg$ such that 
\bea
\label{1.KJ}
\cK_\text{E} (x,p; a,b) & = & \cU_{\rm DHS}(x,p)^{-1}  \cJ _\text{DHS} (x,p;\ta, \tb) \,\cU_{\rm DHS}(x,p)
\no \\ &&
-\cU_{\rm DHS}(x,p)^{-1} d_x \, \cU_{\rm DHS}(x,p),
\eea
which is equivalent to the following relation  between the two connections
\bea
\label{gaugetransf1}
d_x-\cK_\text{E} (x,p; a, b) = \cU_{\rm DHS}(x,p)^{-1}  \big ( d_x-\cJ _\text{DHS} (x,p;\ta, \tb) \big ) \, \cU_{\rm DHS}(x,p).
\eea
Here, the gauge transformation $\cU_{\rm DHS}(x,p)$ is obtained from the path-ordered exponential 
\bea
\gDHS (x,y,p;\xi,\eta) = \text{P} \exp \int _y ^x \cJ_{\rm DHS}(\ti,p;\xi,\eta)
\eea 
by specializing the free non-commutative variables $\xi=\{ \xi^1, \cdots , \xi^h\}$ and $\eta=\{\eta_1, \cdots, \eta_h\}$ to appropriate values $\hat\xi,\hat\eta\in \mg^h_b$ which satisfy\footnote{Throughout, unless otherwise indicated, we will follow the Einstein convention in which a pair of identical upper and lower indices are summed over $1,2,\cdots,h$, without writing the summation sign.} $[\hat \eta_I, \hat \xi^I]=0$.
This in turn implies that $\cJ_{\rm DHS}(t,p;\hat \xi, \hat \eta)$ is regular at $t=p$, and therefore that $\gDHS (x,y,p; \xih,\etah)$ is regular when $y$ approaches $p$; we are therefore allowed to set 
\bea
\cU_{\rm DHS}(x,p)=\gDHS (x,p,p;\hat\xi,\hat\eta).
\eea
Throughout we shall suppress the variables $\hat \xi$ and $\hat \eta$ in writing $\cU_\text{DHS}(x,p)$. The notation $\ta, \tb$ stands for an alternative set of generators $\ta = \{ \ta^1, \ldots, \ta^h \}$ and $\tb=\{ \tb_1, \ldots, \tb_h \}$ of~$\mg$, so that each hatted element is a Lie series in the original unhatted elements from the set $a \cup b$, and the map $a \cup b \to \ta \cup \tb$ from unhatted to hatted elements can be viewed as an automorphism of the Lie algebra~$\mg$.

\sm

We outline a procedure to determine the gauge transformation $\cU_{\rm DHS}(x,p)$ and the automorphism $a \cup b \to \ta \cup \tb$, constructed above, in a series expansions in powers of the generators $b$.
This procedure leads to explicit formulas relating the expansion coefficients $g^{I_1\cdots  I_r}{}_J(x,p)$ and $ f^{I_1 \cdots I_r}{}_J(x,p)$  of $\cK_\text{E} (x,p; a,b)$ and $\cJ _\text{DHS} (x,p;a, b) $, respectively, which furnish the integration kernels for the associated polylogarithms \cite{Enriquez:2011, DHS:2023}. For example, to low degree, the formulas of Proposition \ref{3.prop:5} imply the relations
\begin{align}
g^I{}_J(x,p) &= f^I{}_J(x,p)+ \TT^{I}(x,p) \omega_J(x) + \omega_K(x) \MM^{KI}{}_J(p),
\label{intro.41} \\
g^{I_1 I_2}{}_J(x,p) &= f^{I_1 I_2}{}_J(x,p)
+ \TT^{I_1}(x,p) f^{I_2}{}_J(x,p)
\no \\ &\quad
+f^{I_1}{}_K(x,p) \MM^{K I_2}{}_J(p)
- \MM^{I_1 I_2}{}_K(p) f^K{}_J(x,p) \notag \\
&\quad +\TT^{I_1 I_2}(x,p) \omega_J(x) + \TT^{I_1}(x,p) \omega_K(x) \MM^{K I_2}{}_J(p)\notag \\
&\quad -\MM^{I_1 I_2}{}_K(p) \TT^K(x,p) \omega_J(x) + \omega_K(x) \MM^{KI_1 I_2}{}_J(p).
\notag 
\end{align}
The smooth functions $\TT^{I}(x,p),\TT^{I_1 I_2}(x,p)$ and  $\MM^{KI}{}_J(p), \MM^{KI_1 I_2}{}_J(p)$ 
arise as expansion coefficients of the gauge transformation $\cU_{\rm DHS}(x,p)$ and the automorphism 
$a \cup \, b \to \ta \cup \, \tb$ in powers of $b$, respectively, and may be algorithmically computed to arbitrary rank.

\subsubsection{Second construction: $ \cJ _\text{DHS}$ from $\cK_\text{E}$}

Section \ref{sec:4} is devoted to the construction (see Theorem \ref{4.thm:1}) of another gauge transformation $ \cU_{\rm E}(x,p)$ and automorphism mapping $a \cup b$ to alternative generators $\check{a} \cup \check{b}$ of $\mg$ such that \begin{equation}
\label{1.JK}
\cJ_\text{DHS} (x,p;a,b) = \cU_{\rm E}(x,p)^{-1}  \cK _\text{E} (x,p;\check{a}, \check{b}) \,\cU_{\rm E}(x,p) - \cU_{\rm E}(x,p)^{-1} d_x \, \cU_{\rm E}(x,p),
\end{equation}
which is equivalent to the following relation between the two connections
\bea
\label{gaugetransf2}
d_x-\cJ_\text{DHS} (x,p;a,b) =  \cU_{\rm E}(x,p)^{-1}   \big(d_x-\cK _\text{E} (x,p;\check{a}, \check{b}) \big) \, \cU_{\rm E}(x,p).
\eea 
Here, the full gauge transformation 
\bea
\cU_{\rm E}(x,p)=\gE(x,p,p;\check{\xi},\check{\eta})\,\gminus(x,p;b)^{-1} 
\label{Uefactors}
\eea
 takes the form of a product. The second factor is the inverse of the (anti-holomorphic) path-ordered exponential
\bea
\label{1.Gamb}
\gminus(x,p;b)=\text{P} \exp\int_p^x\big({-}\pi\bar\omega^I(t)b_I \big),
\eea 
with suitably normalized anti-holomorphic Abelian differentials $\bar\omega^I$ defined in (\ref{permat}) and \eqref{2.eqdefombar} below (also, see footnote \ref{footnote:Y} for the raising of indices of $\bar\omega^I$ in (\ref{1.Gamb}) through the inverse of the imaginary part of the period matrix).  The first factor in (\ref{Uefactors}) is obtained by specializing the arguments $\xi,\eta$ of the path-ordered exponential 
\bea
\gE(x,y,p;\xi,\eta) = \text{P} \exp \int _y ^x \cK _\text{E}(\ti,p;\xi,\eta)
\label{keisreg}
\eea
to  appropriate values $\check{\xi},\check{\eta}\in \mg^h_b$ which satisfy $[ \check{\eta}_I , \check{\xi}^I ] = 0$. Since the residue of the pole of $\cK _\text{E}(\ti,p;\xich,\etach)$ in $t$ at $p$ is proportional to $[ \check{\eta}_I , \check{\xi}^I ]$, this in turn implies that $\cK_{\rm E}$ in the integrand of (\ref{keisreg}) is regular at $t=p$, and therefore that we are allowed to take $y=p$ as integration base-point. As in the first construction from (\ref{1.KJ}), the dependence of $\cU_{\rm E}(x,p)$ on the Lie algebra elements is omitted, and the notation $\check{a}, \check{b}$ stands for an alternative set of generators $\check{a} = \{ \check{a}^1, \ldots, \check{a}^h \}$ and $\check b=\{ \check b_1, \ldots, \check b_h \}$ of~$\mg$ which are Lie series in the original generators $a,b$, so that the map $a \, \cup \, b \to \check a \, \cup \, \check b$ induces an automorphism of the Lie algebra~$\mg$.

\sm

The expansion of this second construction in the generators $b$ leads to an
equivalent set of explicit formulae between the integration kernels $g^{I_1\cdots  I_r}{}_J(x,p)$ and $ f^{I_1 \cdots I_r}{}_J(x,p)$. The resulting analogues of the low-order relations (\ref{intro.41}) feature an
alternative formulation of its coefficients $\TT$ and $\MM$ which can also be algorithmically computed
to any desired order.

\subsubsection{Relating the polylogarithms obtained from $ \cJ _\text{DHS}$ and $\cK_\text{E}$}

The path-ordered exponentials $\bGam _\text{DHS}$ and $\bGam_\text{E}$ exploited in the construction of the gauge transformations are the generating series of polylogarithms associated to the connections $d- \cJ _\text{DHS}$ and $d-\cK_\text{E}$, respectively. These two series are related by 
\bea
\bGam_\text{E} (x,y,p;a,b) = \cU_{\rm DHS}(x,p)^{-1} \, \bGam _\text{DHS}(x,y,p;\ta, \tb) \, \cU_{\rm DHS}(y,p)
\label{gaevsgadhs}
\eea
as well as by
\bea
\bGam_\text{DHS} (x,y,p;a,b) = \cU_{\rm E}(x,p)^{-1} \, \bGam _\text{E}(x,y,p;\check{a}, \check{b}) \, \cU_{\rm E}(y,p),
\label{gaevsgadhs2}
\eea
because it follows from \eqref{1.KJ} (resp.\ \eqref{1.JK}) that both sides of \eqref{gaevsgadhs} (resp.\ \eqref{gaevsgadhs2}) satisfy the same differential equation with the same initial condition.
The expansion of $\cU_{\rm DHS}$ and $ \cU_{\rm E}$ in words, analogous to the expansion  of~$\bGam$ in  (\ref{1.words}), together with the expansion of the Lie series $\hat\xi,\hat\eta,\hat a, \hat b$ and $\check{\xi},\check{\eta},\check{a},\check{b}$, leads to two different families of explicit formulas which can be used to relate the polylogarithms for the two connections. By comparing these formulas one can obtain non-trivial identities among functions of $y$, $p$ and of the moduli of the surface.

\sm

Moreover, one can use \eqref{gaevsgadhs} and \eqref{gaevsgadhs2} to deduce information about the spaces of functions generated by the respective polylogarithms. Suppose we denote by $\mathcal H(\mathcal J)$ the algebra of polylogarithms associated with a flat connection $d-\mathcal J$, namely the subalgebra of the algebra of smooth functions on the universal cover of $\Sigma_p$ which is generated by the coefficients $\Gamma(\mw;x,y)$ in (\ref{1.words}) (for a fixed punctured Riemann surface~$\Sigma_p$ and any\footnote{Notice that, by the path-concatenation formula (\ref{1.comp}), changing the base point $y$ does not change the space of polylogarithms $\mathcal H(\mathcal J)$.} fixed $y\in\tilde\Sigma_p$) of the path-ordered exponential $\bGam(x,y;c)=\text{P} \exp \int _y ^x \cJ(\ti;c)$.  Then, combining the two gauge transformations\footnote{The transformation \eqref{gaevsgadhs2} is used to prove the inclusion $\subseteq$ in \eqref{hspaces}, whereas the transformation \eqref{gaevsgadhs}  is used to prove the opposite inclusion $\supseteq$.}, we will prove (see Theorem \ref{5.thmpolspaces} below) the relation
\bea
\mathcal H(\cJ_{\rm DHS})\,=\,\mathcal H(\cK_{\rm E})\cdot \mathcal H(\cJ^{(0,1)}_{\rm DHS}),
\label{hspaces}
\eea
where $\cJ^{(0,1)}_{\rm DHS}(x;b)=-\pi\bar\omega^I(x)b_I$ is the (purely anti-holomorphic) $(0,1)$-part of the differential form $\cJ_{\rm DHS}$. In other words, the polylogarithms generated from the single-valued but non-meromorphic connection $d-\cJ_{\rm DHS}$ in (\ref{dhspoly}) are polynomials in the polylogarithms constructed from the meromorphic connection $d-\cK_{\rm E}$ in (\ref{epoly}) and in the iterated integrals of $\bar \omega^I$, whose coefficients will in general depend on~$y$ and on the moduli of $\Sigma_p$.  Moreover, we deduce from (\ref{hspaces}) that $\mathcal H(\cK_{\rm E})$ is given by the intersection of $\mathcal H(\cJ_{\rm DHS})$ with the algebra of holomorphic multiple-valued functions (see Corollary \ref{5.corol}).


\subsection{Further directions}
\label{sec:further}

The results of this work suggest a variety of follow-up questions and future lines of investigation that should be  relevant to both mathematicians and physicists.

\sm

First, various iterated integrals over $\mA^I$ and $\mB_I$ cycles that arise as expansion coefficients of the gauge transformations and automorphisms constructed in this work may be viewed as higher-genus analogues of elliptic multiple zeta values \cite{Enriquez:Emzv, Broedel:2015hia, Matthes:thesis, Zerbini:thesis}. Their detailed structure and  interrelations remain to be explored and their evaluation remains to be further simplified beyond the genus-one case.  More generally, an improved understanding of the relations among higher-genus multiple zeta values may shed light on the generalization of Tsunogai's derivation algebra \cite{Tsunogai} 
 and analogues beyond genus one of the elliptic associator \cite{CEE, Enriquez:ass, Hain:KZB} (see e.g.\ \cite{Gonzalez:2020} for associators at higher genus). 

\sm

Second, our investigations into the monodromies of higher-genus polylogarithms have implications for the  construction of their single-valued counterparts, which is relegated to  future work. At genus one already, special values of single-valued analogues of elliptic polylogarithms, known as modular graph functions \cite{DHoker:2015gmr, DHoker:2015wxz} and modular graph forms \cite{DHoker:2016mwo}, are playing a key role in organizing the low energy expansion of string theory (recent overviews may be found in \cite{Berkovits:2022ivl, DHoker:2022dxx}; see also \cite{DHoker:2019blr, Gerken:2020xte}). Number theoretic properties of modular graph forms were further studied in \cite{Zerbini:2015rss, Brown:2017qwo, Brown:2017mgf, DHoker:2019xef, Zagier:2019eus}. At higher genus, single-valued versions of the polylogarithms encountered in this work generalize the modular  tensors introduced in \cite{vdG3, Kawazumi, Kawazumi:lecture,DHoker:2020uid} and the higher genus modular graph forms introduced in \cite{DHoker:2017pvk, DHoker:2018mys} to depend on various marked points on the surface. They provide powerful tools for string-amplitude computations, offer a novel perspective on the Fay identities for products of Szeg\"o kernels \cite{DHoker:2023khh,DHoker:2024ozn,Baune:2024ber} and provide promising new angles on the theory of single-valued periods \cite{Brown:2013gia, Brown:2018omk}.


\subsection*{Organization}

The remainder of this paper is organized as follows. Section \ref{sec:2}  is dedicated to a review of flat connections on Riemann surfaces of arbitrary genus. We discuss the meromorphic multiple-valued connection $d-\cK_\text{E}$  introduced in \cite{Enriquez:2011}, and the non-meromorphic but single-valued and modular-invariant connection $d-\cJ_\text{DHS}$ of \cite{DHS:2023}, and their respective restriction to genus-one surfaces. In section \ref{sec:3} and section \ref{sec:4} we present the two announced constructions of a relation between  $d-\cK_\text{E}$ and $d-\cJ_\text{DHS}$ given by the composition of a gauge transformation and an automorphism of the Lie algebra $\mg$. The two pairs of a gauge transformation and an automorphism will be shown to be inverses of one another in section \ref{sec:new5}. In section \ref{sec:5} we relate the two associated spaces of polylogarithms. For some of the results in the main text, the proofs are relegated to appendices  \ref{sec:C} and \ref{sec:cor}. Explicit expressions for the relation between the connections are worked out to low orders in appendix \ref{sec:A}, and the structure of the general  construction is provided in appendix \ref{sec:B}. 


\subsection*{Acknowledgments}

We are grateful to Martijn Hidding for helpful discussions in early stages of this project
and collaboration on related topics. The research of ED is supported in part by NSF grant PHY-22-09700. The research of BE has been partially funded by ANR grant ``Project HighAGT ANR20-CE40-0016''. The research of OS is supported by the European Research Council under ERC-STG-804286 UNISCAMP as well as the strength area ``Universe and mathematical physics'' which is funded by the Faculty of Science and Technology at Uppsala University. This research was supported in part by grant NSF PHY-2309135 to the Kavli Institute for Theoretical Physics (KITP), and ED and OS thank the organizers of the program ``What is string theory'' for creating a stimulating atmosphere in which part of this work was done. OS and FZ are grateful to the Hausdorff Research Institute for Mathematics in Bonn 
 for their hospitality and the organizers and participants of the Follow-Up Workshop ``Periods in Physics, Number Theory and Algebraic Geometry'' for valuable discussions. 
 Moreover, OS and FZ cordially thank the
Galileo Galilei Institute (GGI) for Theoretical Physics in Florence for the hospitality and the INFN for partial support during the program ``Resurgence and Modularity in QFT and String Theory''. OS is grateful to the Simons foundation for financial support during the GGI programme. Finally, the research of OS and FZ was supported by the Munich Institute for Astro-, Particle and BioPhysics (MIAPbP) which is funded by the Deutsche Forschungsgemeinschaft (DFG, German Research Foundation) under Germany's Excellence Strategy -- EXC-2094 -- 390783311. The research of FZ was partly supported by the Spanish Ministry of Science, Innovation and Universities under the 2023 grant ``Proyecto de generación de conocimiento'' PID2023-152822NB-I00, as well as by the Royal Society, under the grant URF\textbackslash R1\textbackslash 201473.

The authors declare that the data supporting the findings of this study are available within the paper and its bibliography.

\newpage

\section{Review of the flat connections}
\setcounter{equation}{0}
\label{sec:2}

Let $\Sigma$ be a compact Riemann surface of genus $h$. We denote by $\tilde \Sigma$ the universal (simply-connected) cover of $\Sigma$, and the associated canonical projection by $\pi : \tilde \Sigma \to \Sigma$. 
For any $p\in\Sigma$, let us consider also the once-punctured surface $\Sigma _p = \Sigma \smallsetminus \{ p \}$. Its fundamental group $\pi_1(\Sigma_p,\qq)$ of $\Sigma _p$ is freely generated by $\mA^I, \mB_I$, $I=1,\cdots, h$, which are $2h$ closed loops in $\Sigma$ based at $\qq$ that do not contain the point~$p$. We assume that, when viewed as generators of $\pi_1(\Sigma,\qq)$, they satisfy the relation 
\bea
\label{2.ABrel}
\prod_{I=1}^h \, \mA^I \star \mB_I \star (\mA^I)^{-1} \!\star (\mB_I)^{-1}=1.
\eea
Let $\tilde\Sigma_p$ denote the universal cover of $\Sigma_p$. Then, a choice of preferred pre-image of~$\qq\in \Sigma_p$ in $\tilde\Sigma_p$ induces a canonical identification of the fundamental group $\pi_1(\Sigma_p,\qq)$ with the automorphism group $\mathrm{Aut}(\tilde\Sigma_p /\Sigma_p)$. This canonical identification sets the action of  an element $\gamma \in \pi_1(\Sigma_p,\qq)$ on $\tilde\Sigma_p$ (see footnote~\ref{ftbasept}).
The image of $x\in\tilde\Sigma_p$ under $\gamma$ will be denoted  by $\gamma \cdot x$.  The preferred pre-image of $\qq$ in $\tilde\Sigma_p$, as well as its image in $\tilde\Sigma\smallsetminus \pi^{-1}(p)$ via the the natural map $\tilde\Sigma_p\to\tilde\Sigma\smallsetminus \pi^{-1}(p)$, will also be denoted by $\qq$, and is part of the topological setup of our construction. Part of this setting is illustrated in figure \ref{fig:1} for a surface of genus two.

\begin{figure}[htb]
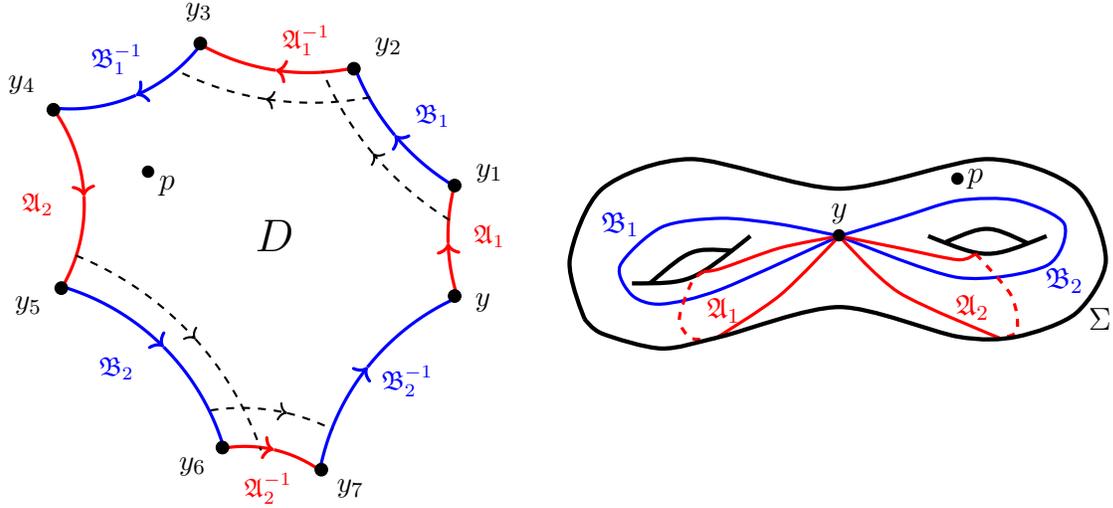

\begin{center}
\tikzpicture[scale=1.05]
\scope[xshift=0cm,yshift=0cm, scale=4]
\draw  [very  thick, color=red, <-, domain=182:198] plot ({cosh(0.6)*cos(0)+sinh(0.6)*cos(\x)}, {cosh(0.6)*sin(0)+sinh(0.6)*sin(\x)});
\draw  [very  thick, color=red, domain=166:182] plot ({cosh(0.6)*cos(0)+sinh(0.6)*cos(\x)}, {cosh(0.6)*sin(0)+sinh(0.6)*sin(\x)});
\draw  [very  thick, color=red, <-, domain=261.5:282] plot ({cosh(0.65)*cos(85)+sinh(0.65)*cos(\x)}, {cosh(0.65)*sin(85)+sinh(0.65)*sin(\x)});
\draw  [very  thick, color=red,  domain=241:261.5] plot ({cosh(0.65)*cos(85)+sinh(0.65)*cos(\x)}, {cosh(0.65)*sin(85)+sinh(0.65)*sin(\x)});
\draw  [very  thick, color=red, <-, domain=363:396] plot ({cosh(0.5)*cos(175)+sinh(0.5)*cos(\x)}, {cosh(0.5)*sin(175)+sinh(0.5)*sin(\x)});
\draw  [very  thick, color=red,  domain=330:363] plot ({cosh(0.5)*cos(175)+sinh(0.5)*cos(\x)}, {cosh(0.5)*sin(175)+sinh(0.5)*sin(\x)});
\draw  [very  thick, color=red, <-, domain=77:100] plot ({cosh(0.4)*cos(265)+sinh(0.4)*cos(\x)}, {cosh(0.4)*sin(265)+sinh(0.4)*sin(\x)});
\draw  [very  thick, color=red,  domain=54:77] plot ({cosh(0.4)*cos(265)+sinh(0.4)*cos(\x)}, {cosh(0.4)*sin(265)+sinh(0.4)*sin(\x)});
\draw  [very  thick, color=blue, <-, domain=220:239] plot ({cosh(0.7)*cos(40)+sinh(0.7)*cos(\x)}, {cosh(0.7)*sin(40)+sinh(0.7)*sin(\x)});
\draw  [very  thick, color=blue, domain=201:220] plot ({cosh(0.7)*cos(40)+sinh(0.7)*cos(\x)}, {cosh(0.7)*sin(40)+sinh(0.7)*sin(\x)});
\draw  [very  thick, color=blue, <-, domain=293.5:323] plot ({cosh(0.5)*cos(125)+sinh(0.5)*cos(\x)}, {cosh(0.5)*sin(125)+sinh(0.5)*sin(\x)});
\draw  [very  thick, color=blue, domain=264:293.5] plot ({cosh(0.5)*cos(125)+sinh(0.5)*cos(\x)}, {cosh(0.5)*sin(125)+sinh(0.5)*sin(\x)});
\draw  [very  thick, color=blue, <-, domain=45:73] plot ({cosh(0.7)*cos(225)+sinh(0.7)*cos(\x)}, {cosh(0.7)*sin(225)+sinh(0.7)*sin(\x)});
\draw  [very  thick, color=blue, domain=17:45] plot ({cosh(0.7)*cos(225)+sinh(0.7)*cos(\x)}, {cosh(0.7)*sin(225)+sinh(0.7)*sin(\x)});
\draw  [very  thick, color=blue, <-, domain=141.5:169] plot ({cosh(0.7)*cos(315)+sinh(0.7)*cos(\x)}, {cosh(0.7)*sin(315)+sinh(0.7)*sin(\x)});
\draw  [very  thick, color=blue, domain=114:141.5] plot ({cosh(0.7)*cos(315)+sinh(0.7)*cos(\x)}, {cosh(0.7)*sin(315)+sinh(0.7)*sin(\x)});
\draw  [thick, color=black, dashed, ->, domain=238:220] plot ({cosh(0.9)*cos(40)+sinh(0.9)*cos(\x)}, {cosh(0.9)*sin(40)+sinh(0.9)*sin(\x)});
\draw  [thick, color=black, dashed,  domain=220:203] plot ({cosh(0.9)*cos(40)+sinh(0.9)*cos(\x)}, {cosh(0.9)*sin(40)+sinh(0.9)*sin(\x)});
\draw (0.148,-0.74) [fill=black] circle(0.02cm) ;
\draw  [thick, color=black, dashed, ->, domain=68:42] plot ({cosh(0.9)*cos(225)+sinh(0.9)*cos(\x)}, {cosh(0.9)*sin(225)+sinh(0.9)*sin(\x)});
\draw  [thick, color=black, dashed,  domain=42:19] plot ({cosh(0.9)*cos(225)+sinh(0.9)*cos(\x)}, {cosh(0.9)*sin(225)+sinh(0.9)*sin(\x)});
\draw  [thick, color=black, dashed, ->, domain=281:261] plot ({cosh(0.85)*cos(85)+sinh(0.85)*cos(\x)}, {cosh(0.85)*sin(85)+sinh(0.85)*sin(\x)});
\draw  [thick, color=black, dashed, domain=261:243] plot ({cosh(0.85)*cos(85)+sinh(0.85)*cos(\x)}, {cosh(0.85)*sin(85)+sinh(0.85)*sin(\x)});
\draw  [thick, color=black, dashed, ->, domain=99:77] plot ({cosh(0.6)*cos(265)+sinh(0.6)*cos(\x)}, {cosh(0.6)*sin(265)+sinh(0.6)*sin(\x)});
\draw  [thick, color=black, dashed, domain=77:63] plot ({cosh(0.6)*cos(265)+sinh(0.6)*cos(\x)}, {cosh(0.6)*sin(265)+sinh(0.6)*sin(\x)});
\draw (0.57,0.16) [fill=black] circle(0.02cm) ;
\draw (0.251,0.53) [fill=black] circle(0.02cm) ;
\draw (-0.7,0.4) [fill=black] circle(0.02cm) ;
\draw (-0.235,0.61) [fill=black] circle(0.02cm) ;
\draw (-0.675,-0.165) [fill=black] circle(0.02cm) ;
\draw (-0.165,-0.67) [fill=black] circle(0.02cm) ;
\draw (0.57,-0.19) [fill=black] circle(0.02cm) ;

\draw (0.66, -0.21) node{\small $\qq$};
\draw (0.68, 0.2) node{\small $\qq_1$};
\draw (0.36, 0.6) node{\small $\qq_2$};
\draw (-0.24, 0.71) node{\small $\qq_3$};
\draw (-0.80, 0.48) node{\small $\qq_4$};
\draw (-0.78, -0.22) node{\small $\qq_5$};
\draw (-0.26, -0.73) node{\small $\qq_6$};
\draw (0.24, -0.8) node{\small $\qq_7$};

\draw [color=red] (0.68,0) node{\footnotesize $\mA_1$};
\draw [color=blue] (0.5,0.38) node{\footnotesize $\mB_1$};
\draw [color=red] (0.1,0.62) node{\footnotesize $\mA_1^{-1}$};
\draw [color=blue] (-0.5,0.56) node{\footnotesize $\mB_1^{-1}$};
\draw [color=red] (-0.75,0.1) node{\footnotesize $\mA_2$};
\draw [color=blue]  (-0.5,-0.42) node{\footnotesize $\mB_2$};
\draw [color=red] (-0.02,-0.8) node{\footnotesize $\mA_2^{-1}$};
\draw [color=blue]  (0.42,-0.46) node{\footnotesize $\mB_2^{-1}$};
\draw (0,0) node{{\Large $D$}};
\draw (-0.4,0.2) node{{$\bullet$}};
\draw (-0.34,0.16) node{{$p$}};
\endscope
\scope[xshift=3.4cm,yshift=-1.5cm, scale=0.75]
\draw[ultra thick] (1.5,1.2) .. controls (2.5, 1.2) .. (3.5,2);
\draw[ultra thick] (1.8,1.2) .. controls (2.5, 1.8) .. (3.2,1.75);
\draw[ultra thick] (6.5,2) .. controls (7.5, 1.6) .. (8.5,2);
\draw[ultra thick] (6.8,1.9) .. controls (7.5, 2.2) .. (8.2,1.87);
\draw[very thick, color=blue] plot [smooth] coordinates {(5,2) (4,2.2) (3, 2.3) (2, 2.15) (1.3,1.5) (1.5,1) ( 2, 0.85) (3,1.1) (5,2) };
\draw[very thick, color=blue] plot [smooth] coordinates {(5,2) (6,2.3) (7, 2.6) (8, 2.6) (8.8,2.2) (8.5,1.5)  (7,1.3) (5,2) };
\draw[very thick, color=red] plot [smooth] coordinates { (5,2) (4,1.8) (3,1.45) (2.7,1.4)};
\draw[very thick, color=red, dashed] plot [smooth] coordinates { (2.7,1.4) (2.5,1.2) (2.3,0.6) (2.5, 0.275) (2.9, 0.265)};
\draw[very thick, color=red] plot [smooth] coordinates { (2.9,0.265) (4,1) (5,2)};
\draw[very thick, color=red] plot [smooth] coordinates { (5,2) (6,1.8) (7,1.6) (7.3,1.7)};
\draw[very thick, color=red, dashed] plot [smooth] coordinates { (7.3,1.7) (7.9,1) (8,0.4) (7.7, 0.28)};
\draw[very thick, color=red] plot [smooth] coordinates { (7.7, 0.28) (6,1.05) (5,2)};
\draw[ultra thick] plot [smooth] coordinates {(0.7,0.6) (0.45,1.6) (1,2.8) (2.5, 3.3)  (5,2.8) (7.5, 3.3) (9,2.8) (9.5,1.6) (9,0.7) (8,0.3) (7, 0.3) (5,0.8) (3, 0.3) (2,0.1) (1,0.35) (0.7,0.6) };
\draw [color=red] (3.05,0.75) node{\small $\mA_1$};
\draw [color=red] (7.25,0.8) node{\small $\mA_2$};
\draw [color=blue] (1.3,2.2) node{\small $\mB_1$};
\draw [color=blue]  (8.8,1.25) node{\small $\mB_2$};
\draw (9.4,0.6) node{{$\Sigma$}};
\draw (5,2) node{$\bullet$};
\draw (5,2.38) node{ $\qq$};
\draw (7,2.95) node{ $\bullet$};
\draw (7.3,2.97) node{ $p$};
\endscope
\endtikzpicture
\caption{The left panel represents  a genus-two Riemann surface $\Sigma$ in terms of a fundamental domain $D \subset \tilde \Sigma$ for the action of $\mathrm{Aut}(\tilde\Sigma /\Sigma)\simeq \pi_1(\Sigma, \qq)$, which can be obtained by cutting $\Sigma$ along the cycles in the right panel. The surface $\Sigma$ may be reconstructed from $D$ by pairwise identifying inverse boundary components with one another under the dashed arrows; the projection $\pi: \tilde \Sigma \to \Sigma$ maps all the vertices $\qq$ and $\qq_i$ for $i=1,\cdots , 7$ of $D$ to the same point~$\qq$ in $\Sigma$.  The points $\qq_i \in \Tilde \Sigma$ are related to $\qq \in \tilde \Sigma$ by $\qq_1 =\mA_1 \cdot \qq$, $\qq_2=\mB_1 \cdot \qq_1$,    $\qq_3 = \mA_1^{-1} \cdot \qq_2$, $\qq_4=   \mB_1^{-1} \cdot \qq_3$, $\qq_5=\mA_2 \cdot \qq_4$, $\qq_6 = \mB_2 \cdot \qq_5$ and $\qq_7 =  \mA_2^{-1} \cdot \qq_6$,  the product of loops being understood here as a composition of the corresponding elements in $\mathrm{Aut}(\tilde\Sigma /\Sigma)$ with $\mB_2^{-1} \cdot \qq_7=\qq$ in view of (\ref{2.ABrel}). 
 \label{fig:1}}
\end{center}
\end{figure}

\sm

The homology groups $H_1(\Sigma, \ZZ)$ and $H_1(\Sigma_p, \ZZ)$ are both isomorphic to~$\mathbb Z^{2h}$, and one can choose the loops $\mA^I, \mB_I$ in such a way that their image in $H_1(\Sigma, \ZZ)=\pi_1^{\mathrm{ab}}(\Sigma,\qq)$ (resp.\ $H_1(\Sigma_p, \ZZ)=\pi_1^{\mathrm{ab}}(\Sigma_p,\qq)$) is a symplectic basis with respect to the canonical intersection pairing $\mJ$, i.e.\ $\mJ(\mA^I, \mA^J)=\mJ(\mB_I, \mB_J)=0$ and $\mJ(\mA^I, \mB_J)=\delta ^I_J$ for $I,J =1, \cdots, h$. The symplectic group $\Sp(2h,\ZZ)$ takes symplectic bases to symplectic bases as follows,\footnote{We reiterate that, unless otherwise indicated, we will follow the Einstein convention in which a pair of identical upper and lower indices are summed over, without writing the summation sign. Indices may be lowered or raised with the help of the \textit{metric} $Y = \Im (\Omega)$, whose  components $Y_{IJ}$ are defined in (\ref{permat}),  or its  inverse $Y^{-1}$, whose components are denoted by $Y^{IJ}$. Following these conventions we have, for example, $\om_I = Y_{IJ} \, \om^J$,  $\om^I= Y^{IJ} \om_J$,  $\bom^I = Y^{IJ} \bom_J$ and $Y_{IJ} Y^{JK} = \delta ^K_I$. \label{footnote:Y}}
\bea
\label{2.mod}
M: \begin{cases}
\mB_I \to   A_I{}^J \, \mB_J + B_{IJ} \, \mA^J
\cr
\mA^I \, \to  C^{IJ} \,  \mB_J + D^I{}_J  \, \mA^J
\end{cases}
\hskip 0.8 in 
M=\left ( \bma A & B \cr C & D \ema \right ) \in \Sp(2h,\ZZ).
\eea
The choice of a symplectic basis of $H_1(\Sigma, \ZZ)$ induces a canonical choice of representatives for the $h$ generators of $H^{1,0}(\Sigma)$, namely the holomorphic Abelian differentials~$\om_I$, with $I=1,\cdots, h$. These differentials are normalized on the $\mA^J$ cycles, and their integrals on the $\mB_J$ cycles give the components of the period matrix $\Omega$:
\bea
\oint _{\mA^J} \om_I = \delta ^J_I,
\hskip 0.8in 
\oint _{\mB_J} \om_I = \Omega _{IJ},
\hskip 0.8in 
Y_{IJ} = \Im (\Omega_{IJ}).
\label{permat}
\eea
By the Riemann relations we have $\Omega^t=\Omega$ and the matrix $Y= \Im (\Omega)$ is positive definite. Equivalently, the following pairing holds, 
\bea
\label{2.Riem}
{ i \over 2} \int _\Sigma \omega _J \wedge \bom^I = \delta ^I_J,
\eea
where, as mentioned in the footnote \ref{footnote:Y}, 
\bea 
\label{2.eqdefombar}
\bom^I = Y^{IJ} \bom_J.
\eea
Under a modular transformation $M \in \Sp(2h,\ZZ)$ of (\ref{2.mod}), the row matrix $\omega$ of holomorphic Abelian differentials and the period matrix  $\Omega$ transform by $M: \omega \to  \om (C\Omega+D)^{-1}$ and  $M: \Omega \to  (A \Omega+B)(C\Omega+D)^{-1}$, respectively.   

\sm

We recall from section \ref{sec:1.2} that the Lie algebra $\mg$, in which the connections that we will consider take their values, is freely generated by a set of $2h$ elements $a \cup b$ where $a=\{ a^1, \cdots, a^h\}$ and $b= \{ b_1, \cdots, b_h\}$ and is completed with respect to the degree, while the (completed) sub-algebra of $\mg$ that is freely generated by the elements of~$b$ alone is denoted $\mg_b$.  In the remainder of this section, we will review the construction of the two $\mg$-valued connections $d-\cK_\text{E}$ and $d-\cJ_{\rm DHS}$, and compare them with their more classical genus-one analogues. More precisely, in sections~\ref{sec:2.1} and~\ref{sec:2.2} we will uniquely characterize $\cK_\text{E}$ and $\cJ_{\rm DHS}$, respectively, through their functional properties, and review some further features. In section \ref{sec:2.3} we will compare them with their genus-one analogues, introduced by Calaque--Enriquez--Etingof and Brown--Levin, respectively.

\subsection{The Enriquez connection $d-\cK_\text{E}$}
\label{sec:2.1}

In this subsection we shall give the definition and list some basic properties of the Enriquez connection $d-\cK_\text{E}$ needed in this paper. The following result, which defines $\cK_\text{E}$ through a functional characterization, essentially follows from \cite{Enriquez:2011} but was never stated in this form. We therefore include also its proof for completeness. The flatness of the corresponding connection trivially follows from the meromorphicity of $\cK_\text{E}$.

{\thm 
\label{2.thm:1}
For any fixed $p\in\Sigma$ there exists a unique differential form (in the variable~$x$) $\cK_{\rm E}(x,p;a,b)$ which is multiple-valued on~$\Sigma$, meromorphic on~$\tilde \Sigma$ with simple poles at all points in $\pi^{-1}(p)$ and holomorphic elsewhere, takes values in~$\mg$ and satisfies:
\begin{enumerate}
\itemsep=-0.02in
\item the monodromy conditions
\bea
\label{2.E1}
\cK_{\rm E} (\mA^K \cdot x, p;a, b) & = & \cK_{\rm E} (x,p;a, b),
\no \\
\cK_{\rm E} (\mB_K \cdot x, p;a, b) & = & e^{-2\pi i b_K} \cK_{\rm E} (x,p;a, b) \, e^{2\pi ib_K} ;
\eea
\item the residue condition (where a preferred pre-image of $p$ in $\tilde \Sigma$ is also denoted by $p$)\footnote{\label{taupoles}The residues of the poles at the points in $\pi^{-1}(p)$, other than $p$ itself, may be obtained by combining (\ref{2.E2a}) with the monodromy relations of (\ref{2.E1}).} 
\bea
\label{2.E2a}
\mathrm{Res}_{x=p} \, \cK_{\rm E} (x,p;a, b) =
[b_I, a^I];
\eea
\item  it is linear in the generators $a^I$.
\end{enumerate}
}
\begin{proof} It follows from \cite{Enriquez:2011}, Lemma\footnote{Condition (a) in the original reference (which is  (\ref{A-cycleintg}) in this work) is not necessary, as it follows from the other properties (see our proof of the uniqueness of the family).} 6, that for fixed $p\in\Sigma$ there exists a family of multiple-valued differentials (in the variable $x$) $g^{I_1 \cdots I_r} {}_J (x,p)$, meromorphic on~$\tilde \Sigma$ with simple poles at (a subset of) points in $\pi^{-1}(p)$ and holomorphic elsewhere, which satisfy the monodromy conditions\footnote{Here the generalized Kronecker symbol $\delta ^{I_1 \cdots I_s }_K$  is defined for arbitrary $s\geq 1$  by the product of standard Kronecker symbols: $\delta^{I_1 \cdots I_s}_K = \delta ^{I_1}_K \cdots \delta ^{I_s}_K$.} 
\bea
\label{2.mon}
g^{I_1 \cdots I_r} {}_J (\mA^K \cdot x ,p) & = &g^{I_1 \cdots I_r} {}_J (x,p),
\no \\
g^{I_1 \cdots I_r} {}_J (\mB_K \cdot x ,p) & = &g^{I_1 \cdots I_r} {}_J (x,p) 
+ \sum _{s=1}^r { (-2\pi i)^s \over s!} \, \delta ^{I_1 \cdots I_s } _K \, g ^{I_{s+1} \cdots I_r} {}_J (x,p),
\eea 
and which are holomorphic at (the preferred pre-image of) $p$ for $r\geq 2$ (see footnote \ref{taupoles}), whereas for $r=1$ they satisfy the residue condition
\bea
\label{2.E3a}
\mathrm{Res}_{x=p} \, g^{I} {}_J (x,p) =  \delta^I_J.
\eea
It follows that, if we define $\cK_{\rm E}$ as the $\mg$-valued generating series of this family, 
\bea
\label{2.Kexp}
\cK_{\rm E} (x,p;a, b)=\sum_{r=0}^\infty g^{I_1 \cdots I_r} {}_J (x,p)  B_{I_1} \cdots B_{I_r} a^J,
\eea
where we set $B_I X = [b_I, X]$ for arbitrary $X \in \mg$, then $\cK_{\rm E}$ satisfies all the required properties, which proves the existence part of the statement. 

\sm

For the uniqueness, notice first that, by properties of Lie brackets, linearity in the generators $a^I$ is equivalent for $\cK_{\rm E}$ to have an expansion like \eqref{2.Kexp}, hence it is enough to verify the uniqueness of a family of differentials $g^{I_1 \cdots I_r} {}_J (x,p)$ as above. By Cauchy's residue theorem, one can prove, using \eqref{2.mon} and \eqref{2.E3a}, that
\bea\label{A-cycleintg}
\oint _{\mA^K} g^{I_1 \cdots I_r} {}_J (\ti,p) = (-2\pi i)^r { B_r \over r!} \, \delta ^{I_1 \cdots I_r K}_J,
\eea
where $\mA^K$ denote the pre-images of the $\mA$-cycles in a fundamental domain $D\subset \tilde\Sigma$ containing $p$ (see figure \ref{fig:1}) and $B_r$ denote the Bernoulli numbers. For $r=0$ the coefficients
\bea
\label{2.firstord}
g^{} {}_J (x,p) =  \omega_J(x),
\eea
are the holomorphic Abelian differentials on $\Sigma$. This implies that this family is uniquely determined by \eqref{2.mon} and \eqref{2.E3a}, which in turn implies the uniqueness of $\cK_{\rm E} (x,p;a, b)$.
\end{proof}
\begin{rmk}
We have departed from the conventions in the original reference \cite{Enriquez:2011}, as well as \cite{Enriquez:2021} and \cite{Baune:2024}, by factors of $-2\pi i$  to align the genus-one instance of the Enriquez connections and its expansion coefficients with the common  conventions of the particle-physics and string-theory literature reviewed in section \ref{sec:2.3} below. More specifically, the conventions of this work are obtained by rescaling the generators $b_I \rightarrow - 2\pi i b_I$ in\footnote{The notation used for these generators in the original reference \cite{Enriquez:2011} is actually $x_1,\ldots ,x_h$, and the translation into our notation is $x_I=- 2\pi i b_I$. Also, rather than considering a $\mg$-valued differential with a simple pole at $p$, one considers in \cite{Enriquez:2011} a differential valued in the (completion of) the quotient $\mathrm{Lie}(a,b)/[a^I,b_I]$, which turns out to be independent of $p$, but whose analytic construction is exactly the same as that presented here. One may also think of the differential $\cK_{\rm E}$ as obtained as the restriction from $\Sigma^2\smallsetminus \{z_1=z_2\}$ to $\Sigma_p$ of the two-variable version of the Enriquez connection of \cite{Enriquez:2011}, fixing the second component to $z_2=p$ and restricting the target Lie algebra accordingly.} \cite{Enriquez:2021,Baune:2024} (while
leaving the $a^J$ unchanged) which leads to slightly different conditions for the monodromies \eqref{2.E1} and the residue \eqref{2.E2a}, and to the dictionary
\[ g^{I_1 \cdots I_r} {}_J (x,p) = (-2\pi i)^r  \omega^{I_1 \cdots I_r} {}_J (x,p) \]
between the differentials $g^{I_1 \cdots I_r} {}_J (x,p)$ in (\ref{2.Kexp}) and the
$\omega^{I_1 \cdots I_r} {}_J (x,p)$ in \cite{Enriquez:2011}.
\end{rmk}

\sm

\begin{rmk}
It will be useful in the sequel to rephrase the residue conditions \eqref{2.E2a} and \eqref{2.E3a} in terms of distributions. Consider the Dirac $\delta$-function $\delta(x,y)$ which is of type $(1,1)$ in $x$ and type $(0,0)$ in $y$, normalized by $\int _\Sigma \delta (x,y) \phi(x) = \phi(y)$ for an arbitrary scalar test function $\phi$. It is given in local coordinates by $\delta (x,y) = \tfrac{i} {2} dx \wedge d \bar x \, \delta^{(2)}(x,y) $, where $\delta^{(2)}(x,y)$ is the standard coordinate $\delta$-function. Then \eqref{2.E2a} is equivalent to\footnote{Throughout, we shall use the notations $\p_x = dx \, \p / \p x$ and $\bar \p_x = d\bar x \, \p / \p \bar x$, so that the total differential is given by $d_x = \p_x + \bar \p_x$.}
\bea
\label{2.E2}
\bar \p_x \,  \cK_\text{E} (x,p;a, b) =
2\pi i\,  [b_I, a^I]\, \delta(x,p),
\eea
and \eqref{2.E3a} is equivalent to
\bea
\label{2.E3}
\bar \p_x \, g^I{}_J(x,p) = 2 \pi i \, \delta ^I_J \, \delta (x,p).
\eea
\end{rmk}

{\cor[See \cite{Enriquez:2011}, Lemma 9] For $r=0$ or if $I_r\neq J$ the differential form $g^{I_1 \cdots I_r} {}_J (x,p)$ is independent of $p$, otherwise it is a meromorphic multiple-valued function of $p$, whose monodromies are given by 
\bea
g^{I_1 \cdots I_r} {}_J (x,   \mA^K \cdot p) & = & g^{I_1 \cdots I_r} {}_J (x,p),
\no \\
g^{I_1 \cdots I_r} {}_J (x, \mB_K \cdot p) & = & g^{I_1 \cdots I_r} {}_J (x,p) 
+ \delta ^{I_r}_J \sum_{s=0}^{r-1} { (2\pi i)^{r-s} \over (r-s)!} \, g^{I_1 \cdots I_s}{}_K (x,p) \, \delta_K^{I_{s+1} \cdots I_{r-1}}.  \label{2.mon.p}
\eea
}

\noindent
Following (\ref{1.PE}) and (\ref{1.words}), the path-ordered exponential and its expansion in terms of words $\mw$ in non-commutative letters from the set $a\cup b$ and associated polylogarithms $\Gamma_\text{E} (\mw; x,y,p) $ for the connection $\cK_\text{E}$ are given by
\bea
\bGam _\text{E} (x,y,p;a,b) 
= \text{P} \exp \int^x _y \cK_\text{E} (t,p;a,b) 
= \sum _{\mw \, \in \, \cW(a \, \cup \, b)} \mw \, \Gamma_\text{E} (\mw; x,y,p).
\label{epoly}
\eea
The resulting Enriquez polylogarithms $\Gamma_\text{E} (\mw; x,y,p)$ are multiple-valued functions of $x$, $y$, $p \in \Sigma$, which for certain choices of $\mw$ have a logarithmic singularities at $x=p$ and $y=p$.  They can be straightforwardly expressed in terms of the iterated integrals introduced in \cite{Baune:2024, DHoker:2024ozn}
\begin{align}
 \tGamma{ \overrightarrow{I}_{\!\!1} &\overrightarrow{I}_{\!\!2} &\cdots &\overrightarrow{I}_{\!\!\ell} }{
J_1 &J_2 &\cdots &J_\ell }{ p_1 &p_2 &\cdots &p_\ell}{x,y} 
= 
\int^x_y dt \, g^{  \overrightarrow{I}_{\!\!1} }{}_{J_1} (t,p_1) \,
 \tGamma{  \overrightarrow{I}_{\!\!2} &\cdots &\overrightarrow{I}_{\!\!\ell} }{
 J_2 &\cdots &J_\ell }{ p_2 &\cdots &p_\ell}{t,y} \, , \ \ \ \ \ \
  \tGamma{ \emptyset }{ \emptyset }{ \emptyset }{x,y} = 1
  \label{gatw.01}
\end{align}
upon specializing $p_1=p_2=\cdots =p_\ell = p$, where the multi-indices $ \overrightarrow{I}_{\!\! j} = I_j^1 I_j^2 \ldots I_j^r$ may be empty to recover the integration kernels $\omega_J(t)$.  
It should be possible to define by tangential base point regularization \cite{Deligne:1989, Panzer:2015ida, Abreu:2022mfk} also the value at $x=p$ or $y=p$, where the integral is otherwise logarithmically divergent, see \cite{Broedel:2014vla, Matthes:thesis, Enriquez:2023} for the genus-one case.

\sm

Formulas for the integration kernels $g^{I_1\cdots I_r}{}_J(x,p)$ of the higher-genus polylogarithms $\Gamma_\text{E} (\mw; x,y,p)$ in terms of the fundamental form of the third kind and of iterated integrals of Abelian differentials, or in terms of averages on the Schottky cover in the case of real hyperelliptic curves, can be deduced from \cite{Enriquez:2021} (see section~5) and \cite{Baune:2024}, respectively.

\subsection{The DHS connection $d-\cJ_\text{DHS}$}
\label{sec:2.2}

The DHS connection $\cJ_\text{DHS}(x,p;a,b)$ is a single-valued smooth differential form in $x \in \Sigma_p$, which is the sum of a $(1,0)$-form $\cJ_\text{DHS} ^{(1,0)}$ that has a regular singularity at $x=p$, and a $(0,1)$-form $\cJ_\text{DHS} ^{(0,1)}$ that is purely anti-holomorphic and single-valued on the whole $\Sigma$. The connection $\cJ_\text{DHS}(x,p;a,b)$ takes values in the Lie algebra $\mg$. It was defined in \cite{DHS:2023} via an expansion, similar to the one given for $\cK_\text{E}$ in (\ref{2.Kexp}), which we shall repeat below, and may also be characterized by its functional properties as follows.

{\thm 
\label{2.thm:2}
For any fixed $p\in\Sigma$ there exists a unique differential form (in the variable~$x$) $\cJ_{\rm DHS}(x,p;a,b)$ which is smooth and single-valued on $\Sigma_p$, takes values in~$\mg$, and satisfies:
\begin{enumerate}
\itemsep=-0.02in
\item the Maurer--Cartan equation for $x \not = p$, and it has a simple pole in $x$ at $p$:
\bea
\label{2.MC1}
d_x \cJ_{\rm DHS} (x,p;a,b) - \cJ_{\rm DHS} (x,p;a,b) \wedge \cJ_{\rm DHS} (x,p;a,b)= 2 \pi i \, \delta (x,p) \, [b_I, a^I];
\qquad
\eea
\item its $(0,1)$ component is given by $\cJ_{\rm DHS} ^{(0,1)} (x,p;a,b) = - \pi b_I \bom^I(x)$ with $\bom^I$ given in (\ref{2.Riem});
\item its $(1,0)$ component $\cJ_{\rm DHS} ^{(1,0)} (x,p;a,b)$ is linear in the generators $a^J$.
\end{enumerate}

\sm

\begin{proof} The existence follows from the explicit construction of \cite{DHS:2023} of a family of differentials $f^{I_1 \cdots I_r} {}_J$ whose generating series, similar to the expansion of $\cK_\text{E}$ in (\ref{2.Kexp}), is a $(1,0)$ differential $\cJ_\text{DHS} ^{(1,0)} (x,p;a,b)$ which satisfies the required properties. Let us show here how the properties in the statement uniquely characterize the construction of \cite{DHS:2023}.
The combination of item 1 and item 2 gives a differential equation for the $(1,0)$ component,
\bea
\label{2.MC2}
\bar \p _x \, \cJ_\text{DHS} ^{(1,0)} (x,p;a,b) 
+ \pi \bom^I(x)  \wedge \big [ b_I, \cJ_\text{DHS} ^{(1,0)} (x,p;a,b) \big ]= 2 \pi i \, \delta (x,p) \, [b_I, a^I].
\eea
Linearity of $\cJ_\text{DHS} ^{(1,0)} (x,p;a,b)$ in the generators of $a$, as prescribed by item 3, combined with 
the structure of the Maurer--Cartan equation given in (\ref{2.MC2}) is equivalent to  the condition that $\cJ_\text{DHS} ^{(1,0)} (x,p;a,b) $ admits the following expansion
\bea
\cJ_\text{DHS} ^{(1,0)} (x,p;a,b) = 
\om_J(x) a^J + \sum_{r=1}^\infty f^{I_1 \cdots I_r} {}_J (x,p) B_{I_1} \cdots B_{I_r} a^J,
\label{jdhs10}
\eea
where we recall that $B_I X = [b_I, X]$ for $X \in \mg$, and where  $ f^{I_1 \cdots I_r} {}_J (x,p)$ are $(1,0)$ forms in~$x$ and scalars in $p$ that are single-valued for $(x,p) \in \Sigma \times \Sigma$, with a simple pole at $x=p$ for $r=1$ and smooth otherwise. The equation (\ref{2.MC2}) translates into the following set of differential equations for $ f^{I_1 \cdots I_r} {}_J (x,p)$, referred to as a Massey system,
\bea
\label{2.massey}
\bar \p_x  f ^I{} _J(x,p) & = & - \pi \bom ^I (x) \wedge \om_J (x) + 2\pi i \, \delta ^I_J \, \delta (x,p), 
\no  \\
\bar \p_x  f^{I_1 \cdots I_r} {}_J (x,p) & = & - \pi \bom ^{I_1} (x) \wedge  f^{I_2 \cdots I_r} {}_J (x,p), \hskip 1in r\geq 2. 
\eea
Note that integrability of the Massey system requires that $f^{I_1 \cdots I_r} {}_J (x,p)  \in {\rm Range} (\p_x)$ for all $r \geq 1$ since the integral of the left side of each equation vanishes and therefore so must the right side. The solution of this system is the content of the following lemma.

{\lem
\label{2.lem:1}
The solution for $f^I{}_J(x,p)$ of the first equation in the Massey system of (\ref{2.massey})  is unique and given in terms of the Arakelov Green function $\cG(x,y)$ as follows,
\bea
\label{2.r1}
f^I{}_J(x,p) = \p_x \int _\Sigma \cG(x,\ti) \Big (  {-} {i \over 2} \bom^I(\ti ) \wedge \om_J(\ti ) - \delta ^I_J \, \delta(\ti,p) \Big ).
\eea
The solution for $  f^{I_1 \cdots I_r} {}_J(x,p)$ with $r\geq 2$ of the second equation of the Massey system of (\ref{2.massey}) is unique and given recursively in $r$ by the absolutely convergent integrals for $x \not=p$\bea
\label{2.r2}
  f^{I_1 \cdots I_r} {}_J(x,p) 
= \p_x  \int _\Sigma  \cG(x,\ti) \Big ( {-} {i \over 2}  \bom^{I_1} (\ti) \wedge f^{I_2 \cdots I_r}{}_J(\ti,p)  \Big ).
\eea

\sm

\begin{proof}[Proof of the lemma] 
To prove the lemma, we use the Arakelov Green function $\cG(x,y)$, which  was introduced in mathematics in \cite{Faltings} and applied in physics in \cite{Alvarez-Gaume:1986nqf, DHoker:2017pvk}. It is the unique real-valued symmetric function of $(x,y) \in \Sigma \times \Sigma$ which is smooth for $y \not= x$ and satisfies the following equations, 
\bea
\label{2.Arak}
\bar \p_x  \p_x \, \cG(x,y) = 2 \pi i \Big ( \kappa (x) - \delta (x,y)  \Big ),
\hskip 1in 
\int _\Sigma \kappa (\ti ) \, \cG(\ti ,y)=0,
\eea
where $\kappa$ is the canonical volume form on $\Sigma$, given by,
\bea
\label{2.kappa}
\kappa = { i \over 2h} \, \om_I \wedge \bom^I,
\hskip 1in 
\int _\Sigma \kappa =1.
\eea 
Thus, $\cG(x,y)$ is a smooth function away from the diagonal $x=y$, where it has a logarithmic singularity given by $\cG(x,y) = - \ln |x-y|^2 + \hbox{regular}$, as a result of which its differential has a simple pole,
\bea
\p_x \, \cG(x,y) = - { dx \over x-y} + \hbox{regular},
\eea
and the integrals in (\ref{2.r2}) are absolutely convergent for $x \not=p$. The limit $x\to p$ of (\ref{2.r2}) and its regularization in case of $r=2$ were discussed in section 8 of \cite{DHoker:2024ozn}. 

\sm

Returning to establishing the solution to the Massey system, we readily obtain the solution (\ref{2.r1}) to the equation for $r=1$ in terms of $\cG(x,y)$  since the right side of the first equation in (\ref{2.massey})  integrates to zero. By construction, the solution for $f^I{}_J(x,p)$ given in (\ref{2.r1}) belongs to the range of the operator $\p_x$, as argued already earlier. As  a result the integral over $x$ of the right side of the $r=2$ equation in (\ref{2.r2}) vanishes, so that one may solve the $r=2$ equation for $f^{I_1 I_2}{}_J(x,p)$ in terms of the Arakelov Green function as well. By induction on $r$ one establishes (\ref{2.r2}) for all values of~$r$. This concludes the proof of Lemma~\ref{2.lem:1}.
\end{proof}}

\sm

The family of differentials $f^{I_1 \cdots I_r} {}_J(x,p)$ obtained from the lemma, which coincides with the family originally defined in \cite{DHS:2023}, is therefore uniquely determined by items 1--3, thus concluding the proof of the theorem.
\end{proof}}

\begin{rmk}
Note that linearity of $\cJ_{\rm DHS} ^{(1,0)} (x,p;a,b)$ in $a^J$ which is required in item 3  is consistent with the linearity of the right side of (\ref{2.MC1}). 
\end{rmk}

\subsubsection{Modular invariance of the DHS connection}

The  holomorphic Abelian differentials $\om_I$, and their conjugates $\bom^I$ introduced in (\ref{2.Riem}),  transform under non-linear realizations of the modular group $P,Q: \Sp(2h,\ZZ) \times \mathfrak H_h \to \GL(h,\CC)$, where $\mathfrak H_h$ denotes the Siegel upper half-space, whose action is given as follows, 
\bea
\label{2.mod2}
\begin{cases}  \om_I  \to P(M,\Omega) _I {}^J \om_J \cr  \bom^I \to Q(M,\Omega) ^I {}_J \bom^J \end{cases}
\hskip 0.3in 
\begin{cases} P(M,\Omega) = \big ( Q(M,\Omega)^t \big ) ^{-1} \cr Q(M,\Omega)  = C \Omega+D \end{cases}
\hskip 0.2in 
M= \left ( \bma A & B \cr C & D \ema \right ),
\quad
\eea
for $M \in \Sp(2h,\ZZ)$ and $\Omega \in \mH_h$. The composition law for $M_1, M_2 \in \Sp(2h,\ZZ)$ is given by,
\bea
\begin{cases} P(M_1 M_2 , \Omega) = P(M_1,   \Omega') P(M_2 , \Omega) \cr
Q(M_1 M_2 , \Omega) = Q(M_1,   \Omega') Q(M_2, \Omega) \end{cases}
\hskip 0.35in 
 \Omega' = (A_2 \Omega + B_2) (C_2 \Omega + D_2)^{-1}. \
\eea
The equations defining the Arakelov Green function in (\ref{2.Arak}) are modular invariant, and so is $\cG(x,y)$. Therefore, the functions $f^{I_1 \cdots I_r}{}_J(x,p)$ are actually modular tensors in the sense introduced in  \cite{DHoker:2020uid},  whose  transformation law may be deduced from that of the Abelian differentials in (\ref{2.mod2}),
\bea
M: f^{I_1 \cdots I_r} {}_J(x,p) \to Q^{I_1}{}_{K_1} \cdots Q^{I_r}{}_{K_r} P_J{}^L 
f^{K_1 \cdots K_r} {}_L(x,p),
\eea 
where we have abbreviated $P=P(M,\Omega)$ and $Q=Q(M,\Omega)$ for $M \in \Sp(2h,\ZZ)$.  Note that the dependence of both $P(M,\Omega)$ and $Q(M,\Omega)$ on $\Omega$ is holomorphic, as is the entire transformation factor of the modular tensors $f^{I_1 \cdots I_r} {}_J(x,p)$. The factors $P$ and $Q$ and their tensor products generalize the automorphy factors $(c\tau+d)^k$ associated with ${\rm SL}(2,\ZZ)$ modular transformations of genus one. 

\sm

The result may be summarized by the following proposition (see Theorem 3.2 of \cite{DHS:2023}).
{\prop
\label{2.cor:1}
The connection $\cJ_{\rm DHS} (x,p;a,b) $ is invariant under the action of the modular group $\Sp(2h,\ZZ)$ provided the generators of the Lie algebra $\mg$ transform as follows,
\bea
M: a^I \to Q(M,\Omega) ^I{}_J \, a^J,
\hskip 1in
M : b_I \to P(M,\Omega) _I{}^J \, b_J , 
\label{modabtrf}
\eea
where $P(M,\Omega)$ and $Q(M,\Omega)$ are defined in (\ref{2.mod2}) for $M \in \Sp(2h,\ZZ)$.}

\subsubsection{Modular properties of the DHS polylogarithms}

Following (\ref{1.PE}) and (\ref{1.words}), the path-ordered exponential and its expansion in terms of words $\mw$ in non-commutative letters from the set $a\cup b$ and associated polylogarithms $\Gamma_\text{DHS} (\mw; x,y,p) $ for the connection $\cJ_\text{DHS}$ are introduced as follows,
\bea
\bGam _\text{DHS} (x,y,p;a,b) 
= \text{P} \exp \int^x _y \cJ_\text{DHS} (t,p;a,b) 
= \sum _{ \mw \, \in \, \cW(a \, \cup \, b)} \mw \, \Gamma_\text{DHS} (\mw; x,y,p).
\label{dhspoly}
\eea
The implications for the modular transformation law of the polylogarithms $\Gamma_\text{DHS} (\mw;x,y,p)$  are summarized by the following proposition (see section 4.4 of \cite{DHS:2023}). 

{\prop
\label{2.cor:2}
The polylogarithms $\Gamma _{\rm DHS}(\mw;x,y,p)$ associated with words composed of the alphabet $a \cup b$ with $a=\{ a^1, \cdots, a^h \}$ and $b=\{ b_1, \cdots, b_h\}$ map to modular tensors,
\bea
\Gamma _{\rm DHS}(a^{I_1} \cdots a^{I_m} b_{J_1} \cdots b_{J_n} \cdots ;x,y,p)  
&  =  &
\Gamma _{I_1 \cdots I_m}{}^{J_1 \cdots J_n} {}_{\cdots} {}^{\cdots} (x,y,p),
\eea
whose transformation law under $M \in \Sp(2h,\ZZ)$ is    given by, 
\bea
 M:  ~ \Gamma _{I_1 \cdots I_m}{}^{J_1 \cdots J_n} {}_{\cdots} {} ^{\cdots} 
& \to  & P_{I_1} {}^{K_1} \cdots P_{I_m} {}^{K_m} Q^{J_1}{}_{L_1} \cdots Q^{J_n}{}_{L_n} \cdots 
\Gamma _{K_1 \cdots K_m}{}^{L_1 \cdots L_n} {}_{\cdots} {}^{\cdots},
\quad
\eea
where again $P=P(M,\Omega)$ and $Q= Q(M,\Omega)$ for brevity.}

\subsection{Restriction to genus one}
\label{sec:2.3}

We follow the customary notation $\tau = \Omega_{11}$ for the restriction of the period matrix (\ref{permat}) to genus $h=1$, and we identify~$\Sigma$ with its Jacobian, the complex torus $\CC/(\ZZ+ \tau \ZZ)$ for $\Im\tau >0$. The restriction of both the Enriquez connection ${\cal K}_{\rm E}$ and the DHS connection ${\cal J}_{\rm DHS}$ to genus one can be explicitly expressed in terms of the odd Jacobi theta function,
\bea
\vartheta_1(x) = 2 q^{1/8} \sin(\pi x) \prod_{n=1}^{\infty} (1-q^n) (1-e^{2\pi i x} q^n) (1-e^{-2\pi i x} q^n)\,, 
 \ \ \ \ q= e^{2\pi i \tau}.
\label{h1.01}
\eea
Specifically, the genus-one connections involve the Kronecker function $F(x;\a)$ (also known as Kronecker-Eisenstein series), which can be defined as \cite{Kron1, Kron2} 
\bea
\label{2.FG}
F(x;\a) = \frac{ \vartheta'_1(0) \vartheta_1(x+\a) }{ \vartheta_1(x) \vartheta_1(\a) }.
\eea
The function $F(x;\a)$ is  meromorphic on $\CC \times \CC$ and, viewed as a function on 
$\Sigma \times \Sigma$, is multiple-valued and has the following monodromies in the variable $x$, 
\bea
F(x+1;\a) = F(x;\a),
\hskip 1in
F(x+\tau;\a)= e^{-2\pi i \alpha} F(x;\a).
\label{h1.03}
\eea
These monodromies cancel if one considers the modified version\footnote{The notation $\Omega(x;\a)$ is customary and not to be confused with the period matrix.}
\bea
\label{2.FGa}
\Omega (x;\a) = \exp \left ( 2 \pi i  \, { \Im x \over \Im \tau} \, \a \right )\,F(x;\a),
\eea
which is doubly periodic but non-meromorphic in $x$. The Laurent expansions at $\alpha=0$ of \eqref{2.FG} and its single-valued analogue \eqref{2.FGa} produce the integration kernels $g^{(r)}$ and $f^{(r)}$, respectively:
\bea
\label{2.fg}
F(x;\a) = \sum_{r=0}^\infty \a^{r-1}  g^{(r)}(x),
\hskip 0.4in
\Omega (x;\a) = \sum_{r=0}^\infty \a^{r-1}  f^{(r)}(x),
\eea
with $g^{(0)}(x) = f^{(0)}(x)=1$, $g^{(1)}(x) = \tfrac{\partial}{\partial x} \log \vartheta_1(x)$ and $f^{(1)}(x) = \tfrac{\partial}{\partial x} \log \vartheta_1(x)+ 2\pi i \,  \frac{\Im x}{\Im \tau}$.

\sm

At genus one, the Lie algebra~$\mg$ in which the Enriquez and DHS connections take values is generated by two elements $a$ and $b$, with $a=a^1$ and $b=b_1$ in our earlier notation. In terms of the Kronecker function defined in (\ref{2.FG}), the Enriquez and DHS  connections 
in Theorems \ref{2.thm:1} and \ref{2.thm:2} reduce to the following expressions at genus one, 
\begin{align}
\label{h1.04}
\cK_\text{E} (x,p;a,b) \Big|_{h=1} & =  dx \,  F (x-p;  B)  B a,
 \\
\cJ_\text{DHS}(x,p;a,b) \Big|_{h=1} & = dx \, \Omega (x-p;B)  B a - d \bar x \,{ \pi \, b \over \Im \tau},
\no
\end{align}
where $BX = [b,X]$ for any $X \in \mg$. The right side of the first line is a connection introduced and generalized to multiple variables by\footnote{This connection was independently introduced in the one-variable case by Levin--Racinet~\cite{Levin:2007}.} Calaque--Enriquez--Etingof \cite{CEE}, whereas the right side of the second line coincides (upon adding $dx\, \pi \, b / \Im \tau$) with a connection introduced and generalized to multiple variables by Brown--Levin \cite{BrownLevin}.  The Enriquez  kernels $g^{I_1 \cdots I_r} {}_J (x,p) $ introduced  in (\ref{2.Kexp})  and the DHS kernels $f^{I_1 \cdots I_r} {}_J (x,p) $ introduced in  (\ref{jdhs10}) on a Riemann surface of arbitrary genus,  reduce at genus one to the genus-one kernels $ g^{(r)}(x-p)$  and $f^{(r)}(x-p)$ introduced in (\ref{2.fg}) as follows,
\bea
g^{I_1 \cdots I_r} {}_J (x,p) \Big|_{h=1} & = & g^{(r)}(x-p) \, dx,
\no \\
f^{I_1 \cdots I_r} {}_J (x,p) \Big|_{h=1} & = & f^{(r)}(x-p) \, dx.
\label{h1.05}
\eea
In (\ref{h1.04}) and (\ref{h1.05}) the restrictions to genus one of $\cK_\text{E}(x,p;a,b)$ and $\cJ_\text{DHS}(x,p;a,b)$ become dependent only on the difference $x-p$ thanks to translation invariance on the torus $\Sigma$. 
Without loss of generality one may fix $p$ to be at the origin of $\Sigma$. The monodromy conditions (\ref{2.mon}), (\ref{2.mon.p}) both  specialize to the genus one case as follows,
\bea
g^{(r)}(x+1) = g^{(r)}(x),  
\hskip 0.6in
g^{(r)}(x+\tau) = g^{(r)}(x) + \sum_{s=1}^{r} \frac{(-2\pi i)^s}{s!}  g^{(r-s)}(x),
\label{h1.06}
\eea
consistently with the Laurent expansion of (\ref{h1.03}) with respect to $\alpha$.

\sm

The meromorphic multiple-valued connection of Calaque--Enriquez--Etingof and the non-meromorphic single-valued connection of Brown--Levin on the right side of the two equations in (\ref{h1.04}) may be related by a gauge transformation and an automorphism of the Lie algebra $\mg$. This relation provides an explicit realization of the correspondence (\ref{1.KJ}) in the introduction. Indeed, one readily verifies that the gauge transformation
\bea
{\cal U}_{\rm BL}(x)  =\exp \bigg(  2\pi i \, \frac{\Im x}{\Im \tau} \, b \bigg)
\label{h1.11}
\eea
together with the automorphism $\hat a= a + \pi b / (\Im \tau)$ and $\hat b= b$ reproduce the relation 
\bea
\cK_\text{E} (x,p;a,b) \Big|_{h=1} &=&
\cU_\text{BL} (x{-}p)^{-1}  \cJ_\text{DHS}(x,p;\hat a,\hat b) \Big|_{h=1} \, \cU_\text{BL} (x{-}p) 
\no \\ &&
- \cU_\text{BL} (x{-}p)^{-1} d_x \, \cU_\text{BL} (x{-}p),
\label{h1.12}
\eea
which implies that the two connections are related by
\begin{equation}
d-\cK_\text{E} (x,p;a,b) \Big|_{h=1}\,=\,\cU_\text{BL} (x{-}p)^{-1}\,\Big(d- \cJ_\text{DHS}(x,p;\hat a,\hat b) \Big|_{h=1}\Big)\,\cU_\text{BL} (x{-}p).
\end{equation}
The anti-holomorphic dependence of (\ref{h1.11}) ensures that the $(0,1)$-form components
cancel between the two terms on the right side of (\ref{h1.12}).

\newpage

\section{Gauge transforming $\cJ_\text{DHS}$ to $\cK_\text{E}$ }
\setcounter{equation}{0}
\label{sec:3}

In this section, we give the first explicit construction of a gauge transformation and an automorphism of the Lie algebra $\mg$ that relate the connection $d-\cK_\text{E}$ with the connection $d-\cJ_\text{DHS}$. More precisely, we will relate the differential forms $ \cK_\text{E} (x,p;a,b)$ and $\cJ_\text{DHS}(x,p;\ta, \tb)$ by a gauge transformation $\cU_{\rm DHS}(x,p)$, where $a, b$ and $\ta, \tb$ are two distinct sets of generators of the algebra $\mg$, whose elements are given by
\begin{align} 
a & = \{ a^1 , \cdots, a^h \}, & \ta = \{ \ta ^1, \cdots, \ta^h \},
\no \\
b & = \{ b_1 , \cdots, b_h \}, & \tb = \{ \tb _1, \cdots, \tb_h \}.
\end{align}
The replacement $a \cup b\to \ta \cup \tb$ corresponds to an automorphism of $\mg$ whose explicit form we shall construct. The construction of the full relation between $ \cK_\text{E} (x,p;a,b)$ and $\cJ_\text{DHS}(x,p;\ta, \tb)$ may be conveniently  decomposed into two parts: first the construction of the gauge transformation $\cU_{\rm DHS}(x,p)$ and second the construction of the automorphism. We shall now proceed to each part in turn.

\subsection{Construction of the gauge transformation $\cU_{\rm DHS}$}

To prove the existence of the relation of (\ref{1.KJ}), in this section we shall construct a suitable gauge transformation $\cU_{\rm DHS}(x,p)$ in terms of the connection $\cJ_\text{DHS}$ subject to (\ref{1.KJ}), which we repeat here for convenience, 
\bea
\label{3.KJa}
\cK_\text{E} (x,p; a,b) & = & \cU_{\rm DHS} (x,p)^{-1}  \cJ _\text{DHS} (x,p;\ta,\tb) \, \cU_{\rm DHS} (x,p) 
\no \\ &&
- \cU_{\rm DHS} (x,p)^{-1} d_x \, \cU_{\rm DHS} (x,p).
\eea
The key role of $\cU_{\rm DHS}(x,p)$ is to produce the monodromy of $\cK_\text{E}$, given in equation (\ref{2.E1}) of Theorem \ref{2.thm:1},  starting from the connection  $\cJ_\text{DHS}$ which has trivial monodromy. Thus, we seek to construct a gauge transformation with the following monodromy, 
\bea
\label{monUDHS}
\cU_{\rm DHS}(\mA^K \cdot x, p) & = & \cU_{\rm DHS}(x,p),
\no \\
\cU_{\rm DHS}(\mB_K \cdot x, p) & = & \cU_{\rm DHS} (x,p) \, e^{ 2\pi i b_K} .
\eea

\sm

To obtain $\cU_{\rm DHS}$ in terms of $\cJ_\text{DHS}$,  we begin by considering the solution $\gDHS$  to the following differential equation in the variable $x$,
\bea
\label{3.Ueq}
d_x \, \gDHS (x,y,p; \xi, \eta) & = & \cJ _\text{DHS} (x,p;\xi, \eta) \, \gDHS (x,y,p; \xi, \eta),
\eea
along with the initial condition $\gDHS (y,y,p; \xi, \eta)=1$. Here, $\xi$ and $\eta$ are taken to be arbitrary elements of $\mg^h$ so that $\cJ_\text{DHS}$ and $\gDHS $ take values in $\mg$ and $\exp(\mg)$, respectively. While $\cJ_\text{DHS}(x,p;\xi, \eta)$ is single-valued in $x$ on $\Sigma_p$, the function $\gDHS(x,y,p;\xi,\eta)$ is multiple-valued in $x,y$ on $\Sigma_p$. Therefore, we shall consider (\ref{3.Ueq}) for $x,y \in \tilde \Sigma_p$, and represent~$\Sigma$ by the fundamental domain~$D$ illustrated in figure~\ref{fig:1}, where the points~$y$ in the left and right panels of figure \ref{fig:1} denote the image of $y$ under the natural maps $\tilde \Sigma_p\to\tilde\Sigma\smallsetminus\pi^{-1}(p)$ and $\tilde \Sigma_p\to\Sigma_p$, respectively. The presence of the pole in~$x$ at the point~$p$, 
\bea
\label{3.pole}
\cJ _\text{DHS} (x,p;\xi, \eta) = { [\eta_J, \xi^J] \over x-p} dx + \hbox{ regular},
\eea 
implies that, for generic $\xi,\eta$, the path-ordered exponential $\gDHS (x,y,p; \xi, \eta) $ is singular in~$x$ at~$p$ and has non-trivial monodromy as $x \in D$ circles around $p$.  The solution to (\ref{3.Ueq}) and the initial condition is given by the path-ordered exponential, 
\bea
\label{3.po}
\gDHS (x,y,p;\xi, \eta) = \text{P} \exp \int _y^x \cJ_\text{DHS} (\ti ,p;\xi, \eta).
\eea

\subsubsection{Implementing the monodromy conditions}

In the remainder of this subsection, and in subsection \ref{sec:3.2} below, we will develop a systematic 
method to determine special elements $\xih$ and $\etah$ of $\mg^h_b$ such that $\gDHS(x,y,p;\xih,\etah)$
specializes for $y=p$ to the desired gauge transformation $\cU_{\rm DHS}(x,p)$ in (\ref{3.KJa}).

\sm

To proceed, recall from section \ref{sec:1.1} that the monodromy action of an element $\gamma \in \pi_1 (\Sigma_p, \qq)$ on the end-point $x \in \tilde \Sigma _p$ of the path-ordered exponential $\gDHS$ for arbitrary elements $\xi, \eta \in \mg^h$ is such that
\begin{align}
\gDHS(\gamma \cdot x, y ,p;\xi, \eta) &= \gDHS (x,y, p;\xi, \eta) \, \mu_{\rm DHS}(\gamma, y, p;\xi, \eta),
\label{3.pomu}
\end{align}
where we set $\mu_{\rm DHS}(\gamma, y, p;\xi, \eta)=\gDHS(\gamma \cdot y, y ,p;\xi, \eta)$, inspired by the notations of \eqref{1.eqmon}. 
The monodromy conditions of (\ref{monUDHS}) are implemented using the following lemma. 

{\lem
\label{3.lem:1} For any $y\neq p$ there exist unique elements $\xih(y,p)= \{ \xih^1(y,p), \cdots, \xih^h(y,p)\}$ and $\etah(y,p) =\{ \etah _1(y,p), \cdots, \etah _h(y,p)\}$ of $\mg_b^h$ such that the monodromy of the path-ordered exponential $\gDHS(x,y,p;\xih(y,p), \etah(y,p))$ is given by
\bea
\label{3.mon}
\mu_{\rm DHS}\big(\mA^K,y, p;\xih(y,p), \etah(y,p) \big) & = & 1,
\no \\
\mu_{\rm DHS}\big(\mB_K,y, p;\xih(y,p), \etah(y,p) \big) & = &  e^{2\pi i b_K}.
\eea}
Before proceeding to the proof of this lemma we notice that, intuitively, its statement is justified by the fact that, given $b_K$ for $K=1,\cdots, h$, there are $2h$ equations for $2h$ unknowns $\xih, \etah$.

\begin{proof} By construction, for any $\gamma$ the element $\mu_{\rm DHS}(\gamma, y, p;\xi, \eta)$ belongs to $\mathbb C\langle\!\langle \xi,\eta\rangle\!\rangle$, namely,  it is a formal series in $2h$ non-commutative variables $\xi^I,\eta_I$. The first order in $\xi$ and $\eta$ is obtained by expanding the path-ordered exponential of (\ref{3.pomu}) to first order in $\cJ_\text{DHS}$, and then retaining from the latter only its part linear in $\xi$ and $\eta$, namely $\cJ_\text{DHS}(x,p;\xi, \eta) = \om_J(x) \xi^J - \pi \bom^I (x) \eta_I + \cO(\xi\eta)$,  which gives the following contributions,
\bea
\mu_{\rm DHS} (\mA^K, y, p; \xi, \eta) & = & 1 - \xi^K + \pi \eta ^K + \cO( \xi^2, \eta^2, \xi \eta),
\no \\
\mu_{\rm DHS} (\mB_K, y, p; \xi, \eta) & = & 1 - \Omega _{KJ} \xi^J + \pi \bar \Omega_{KJ} \eta ^J + \cO( \xi^2, \eta^2, \xi \eta).
\eea
Combining this with the desired monodromy relations of \eqref{3.mon} leads to simple first order relations between $\xi, \eta$ and $b$,
\bea
\label{3.xietaA}
\xi^K - \pi \eta ^K + \cO( \xi^2, \eta^2, \xi \eta) & = & 0,
\no \\
\Omega _{KJ} \xi^J - \pi \bar \Omega_{KJ} \eta ^J + \cO( \xi^2, \eta^2, \xi \eta) & = & 2\pi i b_K,
\eea
whose unique solution $\xih,\etah$ to this order is given by
\bea
\label{3.xieta}
 \etah_I(y,p) =  b_I +  \cO(b^2),
\hskip 1in
\xih^I(y,p) =  \pi b^I + \cO(b^2).
\eea
This can be used to prove inductively that there is a unique solution $\xih(y,p),\etah(y,p)$ to all orders whose components belong to $\mathbb C\langle\!\langle b\rangle\!\rangle$. The fact that these components actually belong to the subset $\mg_b$ of $\mathbb C\langle\!\langle b\rangle\!\rangle$ follows from the fact that both sides of \eqref{3.mon} belong to $\exp(\mg_b)$, because one can take the logarithm on both sides and inductively deduce that the components of $\xih(y,p)$ and $\etah(y,p)$ are Lie series.
\end{proof}

\subsubsection{Regularity of the gauge transformation}
\label{sec:3.2}

Although for generic $\xi,\eta\in\mg^h$ the differential $\cJ_\text{DHS}(x,p;\xi,\eta)$ has a pole in $x$ at $p$, which produces a logarithmic singularity in $\gDHS$, the lemma below guarantees that, for $\xih=\xih(y,p), \ \etah=\etah(y,p)$ as in Lemma \ref{3.lem:1}, the coefficient $[\etah_J, \xih^J]$ of this singularity vanishes, so that in this case $\gDHS$ is smooth and single-valued on $\tilde \Sigma$. 

{\lem
\label{3.lem:2}
The quantity $[\etah_J, \xih^J]$ vanishes when $\xih=\xih(y,p),\ \etah=\etah(y,p)\in\mg_b^h$ are such that $\gDHS(x,y,p;\xih,\etah) $ satisfies the monodromy conditions (\ref{3.mon}).}

\begin{proof} To prove this lemma, we observe that $[\etah_J, \xih^J]$ is the residue of $\cJ_{\rm DHS}(x,p;\xih,\etah)$ at the pole at $p$ and consider the fundamental domain $D_p \subset \tilde\Sigma$ for $\Sigma$ depicted in figure \ref{fig:2},  whose boundary curve $\p D_p$ is the $4h$-gon obtained from the union of the curves $\mA^I$, $\mB_I$ and their inverses, with vertices in $\pi^{-1}(\pi(\qq))$, as illustrated in figure~\ref{fig:2}. The curves are chosen  such that a preferred preimage $p\in\tilde\Sigma$ of the point $p\in\Sigma$ is in the interior of~$D_p$. The closed boundary curve $\p D_p\subset\tilde\Sigma\smallsetminus\pi^{-1}(p)$ is homotopic to a small circle $\mC_p$ around the point $p$ (see figure \ref{fig:2}), so that the homotopy invariance of the integral defining $\gDHS$ implies the relation\footnote{Here and elsewhere, $\p D_p$ and $\mC_p$ denote also the lifts of $\p D_p$ and $\mC_p$ from $\tilde\Sigma\smallsetminus\pi^{-1}(p)$ to~$\tilde\Sigma_p$ which are uniquely determined by the choice of a preferred preimage of~$\qq$ in~$\tilde\Sigma_p$.} 
\bea
\gDHS (\p D_p \cdot x, y, p;\xih, \etah) = \gDHS (\mC_p \cdot  x, y, p;\xih, \etah).
\eea
Using the composition law \eqref{1.eqcomplawmu} of the map $\mu_{\rm DHS}$, considered here as a map on $\pi_1(\Sigma_p,\qq)$ with every other dependence omitted, we evaluate $\gDHS (\p D_p \cdot x, y, p;\xih, \etah) $ explicitly, 
\bea
\label{3.Umu}
\gDHS (\p D_p \cdot x, y, p;\xih, \etah) & = & 
\gDHS (x,y,p;\xih,\etah)  
 \\ && \times 
\prod _{K=1}^h \mu_{\rm DHS}(\mA^K) \mu_{\rm DHS} (\mB_K) 
\mu_{\rm DHS} (\mA^K)^{-1} \mu_{\rm DHS} (\mB_K)^{-1} .
\no
\eea
Since $\xih$ and $\etah$ are the solutions to the monodromy relations of (\ref{3.mon}), the factors of $\mu_{\rm DHS}(\gamma) = \mu_{\rm DHS}(\gamma, y, p;\xih, \etah)$  take the values $ \mu_{\rm DHS}(\mA^K) =1$ and  $ \mu_{\rm DHS} (\mB_K) = e^{2 \pi i b_K}$, so that their product on the right side of (\ref{3.Umu}) cancels. As a result, the monodromy of $\gDHS(x,y,p;\xih,\etah)$ around the point $p$ is trivial,
\bea
\label{3.Unomon}
\gDHS (\mC_p \cdot x, y, p;\xih, \etah)  = \gDHS (x, y, p;\xih, \etah).
\eea
Finally, the monodromy may also be evaluated via explicit calculation by choosing local coordinates $x$ in a neighborhood of $p$,  parametrizing the circle $\mC_p$ by polar coordinates $x_\ep (\theta) -p=\ep \, e^{i \theta}$ for $\ep >0$, and deriving from (\ref{3.Ueq}) a differential equation in $\theta$, 
\bea
{ \p \over \p \theta} \, \gDHS (x_\ep (\theta), y, p;\xih, \etah) 
= \Big ( i [\etah_J, \xih^J] + \cO(\ep) \Big ) \, \gDHS (x_\ep (\theta), y, p;\xih, \etah).
\eea
Integrating this equation in $\theta$ from 0 to $2 \pi$,  in the limit $\ep \to 0$, we obtain, 
\bea
\label{3.Umon}
\gDHS (\mC_p \cdot x, y, p;\xih, \etah)  
= e^{ 2 \pi i [\etah_J , \xih^J]} \, \gDHS (x, y, p;\xih, \etah).
\eea
Combining equations (\ref{3.Unomon}) and (\ref{3.Umon}) imposes the condition $e^{2\pi i \, [\etah_J, \xih^J]}=1$. Since the Lie algebra $\mg_b$ is freely generated this implies that the residue of $\cJ_{\rm DHS}(x,p;\xih,\etah)$ vanishes,
\bea
{} [\etah_J, \xih^J]=0,
\eea
and therefore that $\cJ_{\rm DHS}(x,p;\xih,\etah)$ is smooth for $x \in \Sigma$. 
\end{proof}

\begin{figure}[htb]
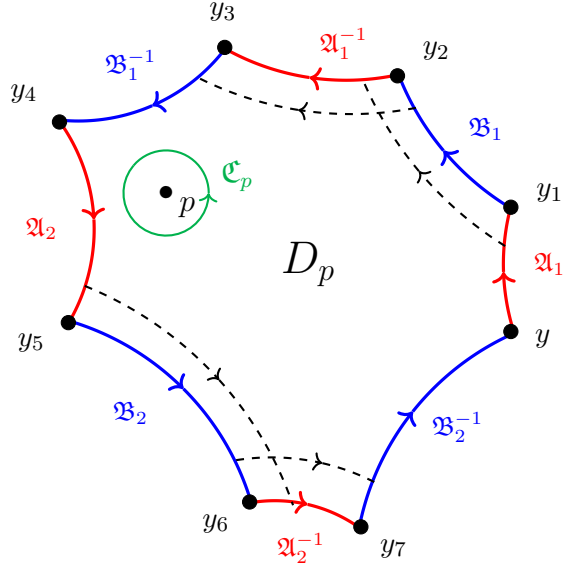

\begin{center}
\tikzpicture[scale=1.05]
\scope[xshift=0cm,yshift=0cm, scale=4.5]
\draw  [very  thick, color=red, <-, domain=182:198] plot ({cosh(0.6)*cos(0)+sinh(0.6)*cos(\x)}, {cosh(0.6)*sin(0)+sinh(0.6)*sin(\x)});
\draw  [very  thick, color=red, domain=166:182] plot ({cosh(0.6)*cos(0)+sinh(0.6)*cos(\x)}, {cosh(0.6)*sin(0)+sinh(0.6)*sin(\x)});
\draw  [very  thick, color=red, <-, domain=261.5:282] plot ({cosh(0.65)*cos(85)+sinh(0.65)*cos(\x)}, {cosh(0.65)*sin(85)+sinh(0.65)*sin(\x)});
\draw  [very  thick, color=red,  domain=241:261.5] plot ({cosh(0.65)*cos(85)+sinh(0.65)*cos(\x)}, {cosh(0.65)*sin(85)+sinh(0.65)*sin(\x)});
\draw  [very  thick, color=red, <-, domain=363:396] plot ({cosh(0.5)*cos(175)+sinh(0.5)*cos(\x)}, {cosh(0.5)*sin(175)+sinh(0.5)*sin(\x)});
\draw  [very  thick, color=red,  domain=330:363] plot ({cosh(0.5)*cos(175)+sinh(0.5)*cos(\x)}, {cosh(0.5)*sin(175)+sinh(0.5)*sin(\x)});
\draw  [very  thick, color=red, <-, domain=77:100] plot ({cosh(0.4)*cos(265)+sinh(0.4)*cos(\x)}, {cosh(0.4)*sin(265)+sinh(0.4)*sin(\x)});
\draw  [very  thick, color=red,  domain=54:77] plot ({cosh(0.4)*cos(265)+sinh(0.4)*cos(\x)}, {cosh(0.4)*sin(265)+sinh(0.4)*sin(\x)});
\draw  [very  thick, color=blue, <-, domain=220:239] plot ({cosh(0.7)*cos(40)+sinh(0.7)*cos(\x)}, {cosh(0.7)*sin(40)+sinh(0.7)*sin(\x)});
\draw  [very  thick, color=blue, domain=201:220] plot ({cosh(0.7)*cos(40)+sinh(0.7)*cos(\x)}, {cosh(0.7)*sin(40)+sinh(0.7)*sin(\x)});
\draw  [very  thick, color=blue, <-, domain=293.5:323] plot ({cosh(0.5)*cos(125)+sinh(0.5)*cos(\x)}, {cosh(0.5)*sin(125)+sinh(0.5)*sin(\x)});
\draw  [very  thick, color=blue, domain=264:293.5] plot ({cosh(0.5)*cos(125)+sinh(0.5)*cos(\x)}, {cosh(0.5)*sin(125)+sinh(0.5)*sin(\x)});
\draw  [very  thick, color=blue, <-, domain=45:73] plot ({cosh(0.7)*cos(225)+sinh(0.7)*cos(\x)}, {cosh(0.7)*sin(225)+sinh(0.7)*sin(\x)});
\draw  [very  thick, color=blue, domain=17:45] plot ({cosh(0.7)*cos(225)+sinh(0.7)*cos(\x)}, {cosh(0.7)*sin(225)+sinh(0.7)*sin(\x)});
\draw  [very  thick, color=blue, <-, domain=141.5:169] plot ({cosh(0.7)*cos(315)+sinh(0.7)*cos(\x)}, {cosh(0.7)*sin(315)+sinh(0.7)*sin(\x)});
\draw  [very  thick, color=blue, domain=114:141.5] plot ({cosh(0.7)*cos(315)+sinh(0.7)*cos(\x)}, {cosh(0.7)*sin(315)+sinh(0.7)*sin(\x)});
\draw  [thick, color=black, dashed, ->, domain=238:220] plot ({cosh(0.9)*cos(40)+sinh(0.9)*cos(\x)}, {cosh(0.9)*sin(40)+sinh(0.9)*sin(\x)});
\draw  [thick, color=black, dashed,  domain=220:203] plot ({cosh(0.9)*cos(40)+sinh(0.9)*cos(\x)}, {cosh(0.9)*sin(40)+sinh(0.9)*sin(\x)});
\draw (0.148,-0.74) [fill=black] circle(0.02cm) ;
\draw  [thick, color=black, dashed, ->, domain=68:42] plot ({cosh(0.9)*cos(225)+sinh(0.9)*cos(\x)}, {cosh(0.9)*sin(225)+sinh(0.9)*sin(\x)});
\draw  [thick, color=black, dashed,  domain=42:19] plot ({cosh(0.9)*cos(225)+sinh(0.9)*cos(\x)}, {cosh(0.9)*sin(225)+sinh(0.9)*sin(\x)});
\draw  [thick, color=black, dashed, ->, domain=281:261] plot ({cosh(0.85)*cos(85)+sinh(0.85)*cos(\x)}, {cosh(0.85)*sin(85)+sinh(0.85)*sin(\x)});
\draw  [thick, color=black, dashed, domain=261:243] plot ({cosh(0.85)*cos(85)+sinh(0.85)*cos(\x)}, {cosh(0.85)*sin(85)+sinh(0.85)*sin(\x)});
\draw  [thick, color=black, dashed, ->, domain=99:77] plot ({cosh(0.6)*cos(265)+sinh(0.6)*cos(\x)}, {cosh(0.6)*sin(265)+sinh(0.6)*sin(\x)});
\draw  [thick, color=black, dashed, domain=77:63] plot ({cosh(0.6)*cos(265)+sinh(0.6)*cos(\x)}, {cosh(0.6)*sin(265)+sinh(0.6)*sin(\x)});
\draw  [thick, color=dgreen,  domain=0:140] plot ({-0.4+0.12*cos(\x)}, {0.2+0.12*sin(\x)});
\draw  [thick, color=dgreen,  domain=140:360,->] plot ({-0.4+0.12*cos(\x)}, {0.2+0.12*sin(\x)});
%
%
\draw (0.57,0.16) [fill=black] circle(0.02cm) ;
\draw (0.251,0.53) [fill=black] circle(0.02cm) ;
\draw (-0.7,0.4) [fill=black] circle(0.02cm) ;
\draw (-0.235,0.61) [fill=black] circle(0.02cm) ;
\draw (-0.675,-0.165) [fill=black] circle(0.02cm) ;
\draw (-0.165,-0.67) [fill=black] circle(0.02cm) ;
\draw (0.57,-0.19) [fill=black] circle(0.02cm) ;

\draw (0.66, -0.21) node{\small $\qq$};
\draw (0.68, 0.2) node{\small $\qq_1$};
\draw (0.36, 0.6) node{\small $\qq_2$};
\draw (-0.24, 0.71) node{\small $\qq_3$};
\draw (-0.80, 0.48) node{\small $\qq_4$};
\draw (-0.78, -0.22) node{\small $\qq_5$};
\draw (-0.26, -0.73) node{\small $\qq_6$};
\draw (0.24, -0.8) node{\small $\qq_7$};

\draw [color=red] (0.68,0) node{\footnotesize $\mA_1$};
\draw [color=blue] (0.5,0.38) node{\footnotesize $\mB_1$};
\draw [color=red] (0.1,0.62) node{\footnotesize $\mA_1^{-1}$};
\draw [color=blue] (-0.5,0.56) node{\footnotesize $\mB_1^{-1}$};
\draw [color=red] (-0.75,0.1) node{\footnotesize $\mA_2$};
\draw [color=blue]  (-0.5,-0.42) node{\footnotesize $\mB_2$};
\draw [color=red] (-0.02,-0.8) node{\footnotesize $\mA_2^{-1}$};
\draw [color=blue]  (0.42,-0.46) node{\footnotesize $\mB_2^{-1}$};
\draw (0,0) node{{\Large $D_p$}};
\draw (-0.4,0.2) node{{$\bullet$}};
\draw (-0.34,0.16) node{{$p$}};
\draw[color=dgreen]  (-0.2,0.25) node{{$\mC_p$}};
\endscope
\endtikzpicture
\caption{A genus-two Riemann surface $\Sigma_p$ with puncture $p$ can be represented in terms of a (punctured) fundamental domain $D_p \subset \tilde \Sigma\smallsetminus\pi^{-1}(p)$. The surface $\Sigma_p$ may be reconstructed from $D_p$ by pairwise identifying inverse cycles with one another under the dashed arrows.  The points $\qq_i \in \tilde \Sigma$ are related to $\qq \in \tilde \Sigma$  as detailed in the caption of figure \ref{fig:1}. The curve $\mC_p$ is homotopic to the boundary curve $\p D_p$.
 \label{fig:2}}
\end{center}
\end{figure}

It is instructive to verify that the relation $[\etah_J, \xih^J]=0$ is obeyed to low orders of $\xih$ and $\etah$ in $b$. To first order, the result readily follows from (\ref{3.xieta}). The verification of the relation to second order is relegated to appendix \ref{sec:A}.

{\lem
\label{lem:benj:vanishing} 
\begin{enumerate}
\itemsep=-0.02in
 \item In \cite{DHS:2023} one defines inductively a family of smooth functions $\Phi ^{I_1 \cdots I_r} {}_J(x)$ of $x\in\Sigma$ by setting for $r\geq 2$
 \bea \label{eq:251114n1}
 \Phi ^{I_1 \cdots I_r} {}_J(x)=\int_\Sigma\mathcal G(x,t)\, \overline\omega^{I_1}(t) \wedge \partial_t\Phi ^{I_2 \cdots I_{r}} {}_J(t),
 \eea
 where $\mathcal G$ is the Arakelov Green function defined in the proof of Lemma \ref{2.lem:1}, and by setting for $r=1$
\bea \label{eq:251114n2}
 \Phi ^{I} {}_J(x)=\int_\Sigma\mathcal G(x,t)\,\overline\omega^{I}(t) \wedge \omega_J(t).
 \eea 
Then one has 
\bea
\p_x \Phi ^{I_1 \cdots I_r}{}_J(x) = f^{I_1 \cdots I_r} {}_J (x,p)  
- { 1 \over h}  \, f^{I_1 \cdots I_{r-1} K} {}_K (x,p) \, \delta ^{I_r}_J ,
\label{phiastrless}
\eea
where the differentials $f^{I_1 \cdots I_r} {}_J (x,p) $ are as in Lemma \ref{2.lem:1}. 
In particular, the right side in this equation is independent of $p$. 
\item If $\lambda=\{ \lambda^1 , \cdots, \lambda ^h \}$ and  $\mu=\{ \mu _1, \cdots, \mu_h\}$  
are in $\mathfrak g^h$ such that $[\mu_I, \lambda^I]=0$,
then the differential $\cJ_{\rm DHS}(x,p;\lambda,\mu)$ is 
independent of $p$, and has the expansion 
\bea
\label{3.lem:3.3}
\cJ_{\rm DHS} (x,\cdot\,;\lambda,\mu) = 
\om_J(x) \lambda^J - \pi \bom^I(x) \mu_I + \sum _{r=1}^\infty \p_x \Phi ^{I_1 \cdots I_r} {}_J(x) 
M_{I_1} \cdots M_{I_r} \lambda^J,
\quad
\eea
where $M_I X = [\mu_I, X]$ for all $X \in \mg$. 
\end{enumerate}}

\begin{proof} 
1. We first prove the case $r=1$ and then show how the case for arbitrary $r \geq 2$ can be derived  from the case $r=1$ via a recursion relation. Equations \eqref{eq:251114n1} and \eqref{eq:251114n2} imply the  following relations,  
\bea
\label{rec:benj}
 \p_x \Phi ^{I}{}_J(x) & = &
- { i \over 2}  \int_\Sigma \partial_x\mathcal G(x,z) \, \bom^{I}(z) \wedge \om_J(z)  ,
\no \\
 \p_x \Phi ^{I_1 \cdots I_r}{}_J(x) & = &
- { i \over 2}  \int_\Sigma \partial_x\mathcal G(x,z) \, \bom^{I_1}(z) \wedge \partial_z\Phi ^{I_2 \cdots I_r}{}_J(z)  ,
\eea
while the analogous equations for $f^I{}_J(x,p)$ and $f ^{I_1 \cdots I_r}{}_J(x,p)$ are given in (\ref{2.r1}) and (\ref{2.r2}) of  Lemma \ref{2.lem:1}, respectively. Expressing the combination on the right side of (\ref{phiastrless}) with the help of (\ref{2.r1}), we observe that the term proportional to $\delta^I_J$ under the parentheses in the integrand of (\ref{2.r1}) cancels so that we obtain,
\bea
f^I{}_J(x,p) - { 1 \over h }  \, f^K{}_K(x,p) \, \delta ^I_J
= \p_x \int _\Sigma \cG(x,t) \Big ( -{i \over 2} \bom^I(t) \wedge \om_J(t) \Big )
\eea
which is readily seen to coincide with $\p_x \Phi^I{}_J(x)$, thereby establishing item 1 for the case $r=1$. For the case $r \geq 2$, we express the combination on the right side of (\ref{phiastrless}) with the help of (\ref{2.r2}) and obtain the following recursion relation, 
\bea
&&  f^{I_1 \cdots I_r} {}_J(x,p)-{1 \over h}   f^{I_1 \cdots I_{r-1}K} {}_K(x,p) \, \delta ^{I_r}_J
\no \\ && \qquad \! \! \!
= {-} {i \over 2} \int _\Sigma  \partial_x\cG(x,\ti) \left (   \bom ^{I_1}(t) \wedge 
\Big \{ f^{I_2 \cdots I_r}{}_J(\ti,p)-{1 \over h} \, \delta ^{I_r}_Jf^{I_2 \cdots 
I_{r-1}K}{}_K(\ti,p) \Big \} \right ) .
\eea
This recursion relation is identical to the second line in \eqref{rec:benj}, so that the identity for $r=1$ implies the identity for any $r$. The result is manifestly independent of $p$,  thereby completing the proof of item 1 of the lemma.\footnote{An alternative proof may be obtained by combining formula (30) of \cite{DHoker:2023khh} with its trace over the indices~$I_r$ and $J$ and using the tracelessness of $\p_x \Phi ^{I_1 \cdots I_r}{}_J(x)$ which follows from the tracelessness of $\p_x \Phi^I{}_J(x)$ with the recursion relation in the second line of (\ref{rec:benj}).}  

\sm

2. Expressing the coefficients $f^{I_1 \cdots I_r}{}_J(x,p)$ in \eqref{jdhs10} in terms of $\p_x \Phi ^{I_1 \cdots I_r}{}_J(x) $ and $ f^{I_1 \cdots I_{r-1} K} {}_K (x,p) \, \delta ^{I_r}_J$ using  (\ref{phiastrless}) and setting $a =\lambda$ and $b = \mu$ we obtain, 
\bea
\cJ_\text{DHS} ^{(1,0)} (x,p;\lambda,\mu) & = &
\om_J(x) \lambda^J 
+ \sum_{r=1}^\infty \Big (\p_x \Phi ^{I_1 \cdots I_r} {}_J(x) 
\no \\ && \hskip 1.1in 
+ {1\over h}f^{I_1...I_{r-1}K}{}_K(x,p) \delta^{I_r}_J \Big ) 
M_{I_1} \cdots M_{I_r} \lambda^J.
\eea
Using the assumption $[ \mu_I ,  \lambda^I]=0$ of the lemma, we see that the effect of the Kronecker $\delta ^{I_r}_J$ factor is to produce the combination $\delta ^{I_r }_J M_{I_r} \lambda ^J = [\mu_J, \lambda ^J]=0$, so that the terms in $1/h$ cancel. Combining the result with the $\cJ_\text{DHS} ^{(0,1)}(x,p;\lambda, \mu)= - \pi \bom ^J(x) \mu_J$ part proves the expression claimed in (\ref{3.lem:3.3}).
\end{proof}

{\cor
\label{3.cor:1}
For $\xih=\xih(y,p)$ and $\etah=\etah(y,p)$ obeying the monodromy conditions (\ref{3.mon}) of Lemma \ref{3.lem:1}, the differential $\cJ_{\rm DHS}(x,p;\xih,\etah)$  is given by 
\bea
\cJ_{\rm DHS} (x,\cdot\, ;\xih, \etah) = \om_J(x) \xih^J - \pi \bom^I(x) \etah_I + \sum _{r=1}^\infty \p_x \Phi ^{I_1 \cdots I_r} {}_J(x) \hat H_{I_1} \cdots \hat H_{I_r} \xih^J. 
\label{dhsxieta}
\eea
where $\hat H _I X = [\hat \eta _I, X]$ for $X \in \mg$.}
\sm

\begin{proof}
This follows from combining item 2 of Lemma \ref{lem:benj:vanishing} and Lemma \ref{3.lem:2}. 
\end{proof}

{\lem
\label{new:lem:benjamin} 
If $\lambda,\mu \in \mg^h$ are such that $[\mu_I, \lambda^I]=0$, then for any  $y\in\tilde\Sigma$ and any $K \in \{ 1, \cdots, h\}$, the functions  $p\mapsto \mu_{\mathrm{DHS}} (\mA^K,y,p;\lambda,\mu)$ and $p\mapsto\mu_{\mathrm{DHS}} (\mB_K,y,p;\lambda,\mu)$ (defined for $p$ in the complement in $\tilde\Sigma$ of $\mathrm{Aut}(\tilde\Sigma/\Sigma)\cdot y$) are constant.  The maps taking $(y,\lambda,\mu)$, where $y\in\tilde\Sigma$ and $\lambda,\mu \in \mg^h$  such that $[\mu_I, \lambda^I]=0$, to these constant values will be denoted
$(y,\lambda,\mu)\mapsto \mu_{\mathrm{DHS}} (\mA^K,y,\cdot \, ;\lambda,\mu)$ and $(y,\lambda,\mu)\mapsto \mu_{\mathrm{DHS}} (\mB_K,y,\cdot \, ;\lambda,\mu)$. 
}

\begin{proof}
It follows from Lemma \ref{lem:benj:vanishing} and from the assumption on $\lambda,\mu$ that 
$\mathcal J_{\mathrm{DHS}}(x,p;\lambda,\mu)$ is independent on $p$. 
These functions, being defined as the holonomies of the connection 
$d_x-\mathcal J_{\mathrm{DHS}}(x,p;\lambda,\mu)$ based at $y$, are therefore also independent of $p$. 
\end{proof}

{\cor
\label{3.coroll2:BIS}
For any $y\in\tilde\Sigma$, the maps
$p\mapsto \xih(y,p), \ \etah(y,p)\in\mg_b^h$, where $p\in\tilde\Sigma\smallsetminus \mathrm{Aut}(\tilde\Sigma/\Sigma)\cdot y$ given by
Lemma \ref{3.lem:1}, are constant; we denote by $\xih(y,\cdot), \ \etah(y,\cdot)$
these constant values. The maps $y\mapsto \xih(y,\cdot), \ \etah(y,\cdot)$
are smooth functions of $y\in\tilde\Sigma$, and therefore well-defined for $y=p$; their values 
at this point satisfy (and are uniquely determined by) the monodromy conditions
\bea
\label{3.monp}
\mu_{\rm DHS}\big(\mA^K,p, \cdot\,;\xih(p,\cdot\,), \etah(p,\cdot\,) \big) & = & 1,
\no \\
\mu_{\rm DHS}\big(\mB_K,p, \cdot\,;\xih(p,\cdot\,), \etah(p,\cdot\,) \big) & = &  e^{2\pi i b_K}. 
\eea
}

\begin{proof}
The identity $[\hat\eta_I(y,p), \hat\xi^I(y,p)]=0$ follows from Lemma \ref{3.lem:2}.  This fact, combined with
Lemma \ref{new:lem:benjamin}, enables us to rewrite the system \eqref{3.mon} as follows 
\bea
\mu_{\rm DHS}\big(\mA^K,y, \cdot\,;\xih(y,p), \etah(y,p) \big) & = & 1,
\no \\
\mu_{\rm DHS}\big(\mB_K,y, \cdot\,;\xih(y,p), \etah(y,p) \big) & = &  e^{2\pi i b_K}.
\eea
The self-map of $\mathfrak g^{2h}_b$ given by 
\bea
(\lambda, \mu)\mapsto 
\Big ( \mathrm{log} \, \mu_{\rm DHS}\big(\mA^K,y, \cdot \, ;\lambda,\mu \big),
\mathrm{log} \, \mu_{\rm DHS}\big(\mB_K,y, \cdot \, ;\lambda,\mu \big) \Big )_{K=1,\cdots, h}
\eea 
is smooth in $y$, triangular with respect to the  degree filtration of $\mathfrak g^{2h}_b$, 
and bijective. The inverse of this self-map is therefore also smooth in $y$. 
The above system implies that the pair $(\xih(y,p), \etah(y,p))$ is the image of $(0,2\pi i b)$ by this 
inverse map, which proves at the same time its independence in $p$ and its smoothness in $y$. 
\end{proof}

{\deff \label{3deffff}
Henceforth, the notation $\xih,\etah$ will refer to the elements $\xih(p,\cdot\,), \etah(p,\cdot\,)$ of~$\mg_b^h$  which solve the monodromy conditions \eqref{3.monp} from Corollary \ref{3.coroll2:BIS} (i.e.\ abandoning the notation of Lemma \ref{3.lem:2} and Corollary \ref{3.cor:1}).
}

\sm

An explicit evaluation of $\xih$ and $\etah$ to second order in $b$ is relegated to appendix \ref{sec:A}. Unlike the contributions linear in $b_I$ spelled out in (\ref{3.xieta}) which are independent of the moduli\footnote{\label{topfoot} Among the moduli of $\Sigma_p$, we include here and elsewhere also the topological datum of a choice of a preferred preimage $\qq\in\tilde\Sigma_p$ which determines the action of the fundamental group $\pi_1(\Sigma_p,\qq)$ on~$\tilde\Sigma_p$.} of $\Sigma_p$ (besides $b^I= Y^{IK} b_K$), all the higher order contributions in~$b$ to~$\etah_I$ and~$\xih^I$ will be non-trivial  functions of the moduli of $\Sigma_p$.

\subsubsection{Construction of the gauge transformation $\cU_{\rm DHS}(x,p)$}

Finally, we define the gauge transformation $\cU_{\rm DHS}(x,p)$ as follows.
{\deff
\label{3.defU}
For $\xih$ and $\etah$ as in Definition~\ref{3deffff}, the gauge transformation $\cU_{\rm DHS}(x,p) = \cU_{\rm DHS}(x,p;\xih,\etah)$ is defined to be the specialization $\gDHS(x,p,p;\xih,\etah)$ of $\gDHS$, namely the unique solution of the differential equation
\bea
\label{3.cUa}
d_x \, \cU_{\rm DHS}(x,p;\xih,\etah) = \cJ_{\rm DHS}(x,\cdot\,;\xih, \etah) \, \cU_{\rm DHS}(x,p;\xih,\etah)
\eea
with the boundary condition $\cU_{\rm DHS}(p,p;\xih,\etah)=1$.}

\sm

It follows from Lemma \ref{3.lem:2} that the gauge transformation $\cU_{\rm DHS}(x,p)$ is a smooth function of $x,p \in \tilde \Sigma$. It can be written as a path-ordered exponential,
\bea
\label{3.cUb}
\cU_{\rm DHS}(x,p;\xih, \etah) = \text{P} \exp \int ^x_p \cJ_\text{DHS} (\ti,\cdot\,; \xih, \etah),
\eea
with $\cJ_\text{DHS} (\ti,\cdot\,; \xih, \etah)$ as in (\ref{dhsxieta}). Notice that $\cU_{\rm DHS}(x,p;\xih,\etah)$ obeys the following monodromy conditions equivalent to (\ref{3.monp}) which determine both $\xih^I$ and $\etah_I$ as a Lie series in $b_K$ (see section \ref{sec:3.6} below for details),
\bea
\label{3.mon1}
\cU_{\rm DHS}(\mA^K \cdot p, p;\xih, \etah) & = & 1,
\no \\
\cU_{\rm DHS}(\mB_K \cdot p , p;\xih, \etah) & = &  e^{2 \pi i b_K}.
\eea
{\cor \label{dhsrmk}
The function $ \cU_{\rm DHS}(x, p;\xih, \etah) $ in (\ref{3.cUb}) is a solution of the equations (\ref{monUDHS}), namely it satisfies the monodromy conditions
\bea
\label{3.eqmon}
\cU_{\rm DHS}(\mA^K \cdot x, p;\xih, \etah) & = & \cU_{\rm DHS}(x, p;\xih, \etah) ,
\no \\
\cU_{\rm DHS}(\mB_K \cdot x , p;\xih, \etah) & = & \cU_{\rm DHS}(x, p;\xih, \etah)   e^{2 \pi i b_K}.
\eea
}
\begin{proof}
The corollary is readily proven by combining the general property (\ref{1.eqmon}) of path-ordered exponentials of single-valued connections with the monodromy conditions (\ref{3.mon1}).
\end{proof}

\subsection{Relating the connections $d-\cJ_{\rm DHS}$ and $d-\cK_{\rm E}$}

In this subsection we shall show that the gauge transformation $\cU_{\rm DHS}(x,p)$ introduced in Definition \ref{3.defU}, combined with a suitable automorphism of the Lie algebra~$\mg$, relates $\cJ_\text{DHS}$ and $\cK_\text{E}$. The result is summarized in the following theorem.
{\thm 
\label{3.thm:1}
The flat connections $d_x - \cK_{\rm E} (x,p; a,b) $ and $d_x - \cJ _{\rm DHS} (x,p;\ta,\tb)$ are related by the gauge transformation $\cU_{\rm DHS}(x,p) = \cU_{\rm DHS}(x,p;\xih,\etah)$ defined by \eqref{3.cUb}, whose arguments $\xih, \etah \in \mg_b^h $ are the uniquely determined solutions
of \eqref{3.mon1}, so that
\bea
\label{3.KJ}
d_x - \cK_{\rm E} (x,p; a,b)= 
\cU_{\rm DHS} (x,p;\xih,\etah)^{-1} \Big (d_x - \cJ _{\rm DHS} (x,p;\ta,\tb) \Big ) \cU_{\rm DHS}  (x,p;\xih,\etah).
\eea
The elements $\ta, \tb \in \mg^h$ are uniquely determined in terms of $a,b$, and the moduli of~$\Sigma_p$,
 by the linearity of $\cK_{\rm E}(x,p;a,b)$ in~$a$, the linearity of $\cJ_{\rm DHS}^{(1,0)}(x,p;\ta,\tb)$ in $\ta$ and the following residue matching relations, 
\bea
\label{3.hats}
\tb_I = \etah_I,
\hskip 1in 
[b_I, a^I] =  [\etah_I, \ta^I - \xih^I].
\eea
These conditions for $\hat a^I$ and $\hat b_I$ may be recursively solved as formal series in (the non-commutative components of) $b$, the leading order solution being given by $\hat a^I= a^I+ \cO(b)$ and $\hat b_I = b_I + \cO(b^2)$. }

\begin{proof} We begin by noting that the right side of (\ref{3.KJ}) has the same monodromy as~$\cK_\text{E}$ given in (\ref{2.E1}) for \textit{arbitrary} $\ta, \tb$, thanks to the monodromy condition on the gauge transformation 
$\cU_{\rm DHS}(x,p;\xih,\etah)$ of Corollary \ref{dhsrmk}. Next, requiring the right side of (\ref{3.KJ}) to be a form of type $(1,0)$ in~$x$, as indeed $\cK_\text{E}$ should be, is equivalent to requiring the vanishing on its $(0,1)$ component, which is equivalent to the following relation
\bea
\label{3.JU}
\bar \p_x \, \cU_{\rm DHS}(x,p;\xih,\etah)\, \, \cU_{\rm DHS}(x,p;\xih,\etah)^{-1} = \cJ _\text{DHS}^{(0,1)} (x,p;\ta,\tb).
 \eea
By construction in (\ref{3.cUa}), we have 
 \bea
 \bar \p_x \, \cU_{\rm DHS}(x,p;\xih,\etah)\, \, \cU_{\rm DHS}(x,p;\xih,\etah)^{-1}= \cJ _\text{DHS}^{(0,1)} (x,p;\xih,\etah),
 \eea 
 while we have  $\cJ _\text{DHS}^{(0,1)} (x,p;\ta,\tb) = - \pi \bom^I (x) \tb_I$ in view of item 2 of Theorem \ref{2.thm:2}. Combining these results with (\ref{3.JU}) reduces to $\tb _I = \etah_I$, which gives the first relation in (\ref{3.hats}). 
 
 \sm
 
To determine $\ta$, we substitute  the $(1,0)$ component of (\ref{3.cUa}) into (\ref{3.KJ}) and use the linearity of $\cJ_\text{DHS}^{(1,0)}  (x,p; \ta, \etah)$ in $\ta$,  to arrive at the following simplified relation,
\bea
\label{3.JKa}
\cK_\text{E} (x,p;a,b) = 
\cU_{\rm DHS}(x,p;\xih, \etah)^{-1} \,  \cJ_\text{DHS}^{(1,0)}  (x,p; \ta - \xih, \etah) \, \cU_{\rm DHS}(x,p;\xih, \etah).
\eea
By Corollary \ref{dhsrmk}, the right side has the monodromies required on $\cK_\text{E}(x,p;a,b)$ and is meromorphic in~$x$, because it is a $(1,0)$ form which satisfies the Maurer--Cartan equation \eqref{1.MC} (recall that gauge transformations preserve flatness). Moreover, it is clear by construction that its only poles in~$x$ are placed at all pre-images $\pi^{-1}(p)$ of~$p$ on $\tilde\Sigma$, and are all simple. In view of Theorem \ref{2.thm:1}, to identify the right side with $\cK_\text{E}(x,p;a,b)$ it remains to match the residues of their poles in~$x$, which leads to the condition
\bea
\label{3.res-m}
{} [b_I, a^I] =   [\etah_I, \ta^I - \xih^I]  
\eea
and gives the second relation in (\ref{3.hats}). Using the lowest-order solution of the monodromy conditions for $\xih$ and $\etah$, given in (\ref{3.xieta}), reduces the conditions of  (\ref{3.hats}) to $\hat b_I = b_I + \cO(b^2) $ and $[b_I, a^I]= [b_I, \hat a^I] + \cO(b)$ whose solution is the one stated in the last line of Theorem \ref{3.thm:1},  thereby completing its proof. 
\end{proof}

{\rmk
Equation (\ref{3.KJ}) of Theorem \ref{3.thm:1}  is invariant under right-multiplication of the gauge transformation by an $x$-independent element $\cV \in \exp(\mg)$ accompanied by a conjugation of $a$ and $b$, while leaving $\xih^I, \etah_I, \hat a^I$ and $\hat b_I$ invariant,
\bea
\cU_{\rm DHS}(x,p;\xih,\etah) \to \cU_{\rm DHS}(x,p;\xih, \etah) \, \cV ,
\hskip 1in
\begin{cases} a ^J  \to \,  \cV \, a ^J \, \cV^{-1} , \\ b_I \,  \to \, \cV \, b_I \, \cV^{-1} . \end{cases}
\eea
The residue-matching conditions transform covariantly as follows,
\bea
\label{3.hats.1}
\tb_I = \etah_I,
\hskip 1in
\cV [b_I, a^I] \cV^{-1} =  [\etah_I, \ta^I - \xih^I].
\eea
Linearity of $\cK_{\rm E}(x,p;a,b)$ in the generators $a^I$ requires restricting to $\cV \in \exp(\mg_b)$. In particular, changing the base point of the path-ordered integral in (\ref{3.cUb}) from $p$ to an arbitrary point $p' \in \tilde \Sigma$ is equivalent to multiplying $\cU(x,p;\xih,\etah)$ to the right by the $x$-independent factor $\cV= \cU(p',p;\xih,\etah)^{-1} \in \exp(\mg_b)$. }

{\rmk
\label{BLrem}
Upon restriction to genus $h=1$, the results in Theorem \ref{3.thm:1} and its proof reproduce
the automorphism $(\hat a,\hat b) = (a{+}\pi b/(\Im \tau), b)$ and the gauge element (\ref{h1.11}) 
mapping the Brown--Levin connection $d-{\cal J}_{\rm DHS} |_{h=1}$ to that of Calaque--Enriquez--Etingof 
$d-{\cal K}_{\rm E} |_{h=1}$ via (\ref{h1.12}). This can be seen from the fact that any bracket $[b_I,b_J]$
vanishes at $h=1$, which truncates the Lie-series expansion of $\xih$ and $\etah$ to their first order $\pi b/(\Im \tau)$ and~$b$, respectively, see (\ref{3.xieta}). As a consequence, the gauge transformation (\ref{3.cUb}) reduces to the path-ordered exponential of $\pi (dx {-} d\bar x) b/(\Im \tau)$ obtained from (\ref{dhsxieta}) which matches 
$\exp \{   2\pi i \, b \, (\Im x)/(\Im \tau)   \}$ in (\ref{h1.11}).

\sm

The first part $\hat a = a{+}\pi b/(\Im \tau)$ of the automorphism at $h=1$ follows
from comparison of $a$ and $\hat a$ in (\ref{3.JKa}) and recalling the restriction 
$\xih = \pi b/(\Im \tau)$ at genus one. Note that the genus-one generators correspond
to the placement of uppercase and lowercase indices according to $a= a^1, \ b= b_1$ as well as 
$\xih = \xih^1, \ \etah = \etah_1$.
}

\subsection{Iterative construction relating  $\cK_\text{E}$ and $\cJ_\text{DHS}$}
\label{sec:3.6}

In this last subsection, we shall outline an iterative procedure for the explicit construction of the gauge transformation $\cU_{\rm DHS}(x,p)$ in (\ref{3.JKa}) and the automorphism  $a \cup b \rightarrow \ta \cup \tb$ to arbitrary  order as a Lie series in $b$. This procedure will express the Enriquez differentials $g^{I_1 \cdots I_r}{}_J(x,p)$ in (\ref{2.Kexp})  in terms of the DHS differentials $f^{I_1 \cdots I_r}{}_J(x,p)$ in (\ref{2.lem:1}) and their iterated integrals. Calculations to low orders are relegated to appendices~\ref{sec:A} and~\ref{sec:B}.

\subsubsection{Solving for $\hat \xi$ in terms of $\hat \eta$}

The starting point is the path-ordered exponential $\cU_\text{DHS}(x,p;\xih,\etah)$ given in (\ref{3.cUb}), where $\xih, \etah \in \mg_b^h$ satisfy the relation $[\etah_I, \xih^I]=0$. The monodromy conditions of (\ref{3.mon1}) may be solved in two step. First, the $\mA$-monodromy conditions are solved for $\xih$ in terms of $\etah$, as stated in the lemma below. Second, the $\mB$-monodromy conditions may subsequently be solved for both $\xih$ and $\etah$ as functions of $b$, which will be done in subsection \ref{3.sec:3.3} below.
{\lem
\label{3.lem:4}
The solution to the monodromy conditions $\cU_{\rm DHS} ( \mA^K \cdot  p  , p; \xih, \etah)=1$ of (\ref{3.mon1}) for $\xih$ as a function of $\etah$ is  unique  and  given by the following associative series 
\bea
\label{A.Lie}
\xih^I =  \sum_{ r=1 }^\infty \XX^{I J_1 \cdots J_r}  \etah _{J_1} \cdots \etah_{J_r}. 
\eea

\begin{enumerate}
\itemsep=0in
\item  The coefficients $ \XX $ depend on the moduli of $\Sigma$ but are independent of the point $p$.
\item They obey the following shuffle relations (see (\ref{appbb.02}), (\ref{appbb.in}) for the shuffle product)
\bea
\XX^{I (J_1 \cdots J_r \shuffle  K_1 \cdots K_s)} =0, \hskip 1in r,s \geq 1 .
\label{shprop}
\eea
\item They are invariant under cyclic permutations
of their indices, 
\bea
\XX^{I J_1 J_2 \cdots J_r} = \XX^{ J_1 J_2 \cdots J_r I } .
 \label{appbb.06}
\eea
\end{enumerate}
}
The proof of the lemma is presented in appendix \ref{sec:C}.   The leading order solution of (\ref{3.xieta}) gives $\XX^{IJ} = \pi Y^{IJ}$ while higher orders will be evaluated in appendices \ref{sec:A} and \ref{sec:B}.

\subsubsection{Applying the gauge transformation $\cU_\text{DHS}$ to $\cJ_\text{DHS}^{(1,0)}$}

The $\xih$ dependence of $\cU_{\rm DHS}(x,p;\xih, \etah)$ may be eliminated  using the expansion (\ref{A.Lie}) of Lemma \ref{3.lem:4}. The resulting  gauge transformation may be expanded in an associative power series in $\etah$ as follows,
\bea
 \label{reinstated} 
\cU_{\rm DHS} \big(x,p;\xih,\etah \big)^{-1} 
= 1 + \sum_{r=1}^{\infty}  \TT^{I_1 \cdots I_r}(x,p) \, \etah_{I_1} \cdots \etah_{I_r} .
\eea
The path-ordered exponential (\ref{3.cUb}) determines the coefficients $ \TT^{I_1 \cdots I_r}(x,p)$  in terms of DHS polylogarithms and the  coefficients $\XX^{I J_1 \cdots J_r}$ defined in  (\ref{A.Lie}), see (\ref{appbb.19}) for ranks $r=1,2$. The key relation (\ref{3.JKa}) is expressed instead in terms of the  following Lie series 
\bea
\cU_{\rm DHS} \big(x,p;\xih,\etah \big)^{-1} X \, \cU_{\rm DHS} \big(x,p;\xih,\etah \big) 
= X + \sum_{r=1}^{\infty}  \TT^{I_1 \cdots I_r}(x,p) \, \hat H_{I_1} \cdots \hat H_{I_r} X,
 \label{dhse.G} 
\eea
where again $\hat H_I X = [\etah_I, X]$ for all $X \in \mg$.

\sm

Applying the gauge transformation $\cU_{\rm DHS}(x,p;\xih, \etah)$ to $\cJ_\text{DHS}(x,p;X{+}\xih, \etah)$ produces a connection, denoted by $\cK(x,p;X, \etah)$, for an arbitrary $X \in \mg^h$,
\bea
\cK(x,p;X, \etah) 
= \cU_{\rm DHS}(x,p;\xih, \etah)^{-1} \cJ_\text{DHS}^{(1,0)}(x,p;X, \etah) \, \cU_{\rm DHS}(x,p;\xih, \etah)
\label{KnonE}
\eea
which is meromorphic in $x$ and  whose Lie series  in powers of $\hat H_I$ may be written as follows,
\bea
\label{3.K}
\cK(x,p;X, \etah) = \om_J(x) X^J + \sum _{r=1}^\infty h^{I_1 \cdots I_r}{}_J (x,p) \hat H_{I_1} \cdots \hat H_{I_r} X^J . 
\eea
The quantities $h^{I_1 \cdots I_r}{}_J(x,p)$ are $(1,0)$-forms in $x$ and scalars in $p$, which may be viewed as intermediate objects  between the  $f$-tensors and the $g$-kernels. They are meromorphic in $x$ but neither in $p$ nor in the moduli of $\Sigma$, and are simply related to the $f$-tensors, as may be seen by combining the expansions for $\cU_{\rm DHS}(x,p;\xih, \etah)^{-1}$ and $\cJ_\text{DHS}^{(1,0)}(x,p;X, \etah)$,
\begin{align}
 h^{I_1 \cdots I_r}{}_J(x,p) &=   f^{I_1  \cdots I_r }{}_J(x,p)  + \sum_{j=1}^{r} \TT^{I_1  \cdots I_j }(x,p)   f^{ I_{j+1}\cdots I_r}{}_J(x,p)  ,
 \label{dhse.u}
\end{align} 
where we set $f^\emptyset {}_J(x,p)=\om_J(x)$, as usual.

\subsubsection{Implementing the automorphism $a \cup b \to \hat a \cup \hat b$}
\label{3.sec:3.3}

To extract the connection $\cK_\text{E}$ from $\cK$ in (\ref{KnonE}), it remains to implement four operations
\begin{enumerate}
\itemsep 0in
\item solving the $\mB$-monodromy condition $\cU_\text{DHS} (\mB_K \cdot p, p; \xih, \etah)=e^{2 \pi i b_K}$ of (\ref{3.mon1}), after having eliminated $\xih$ using (\ref{A.Lie}), for $\etah$ in terms of Lie series in $b$ as follows
\bea
\label{exphatb}
\etah_I = b_I  + \sum_{r=1}^\infty \LL_I{}^{ J_1 \cdots J_r K} B_{J_1} \cdots B_{J_{r}} b_{K},
\eea
where $B_I X = [b_I, X]$ for all $X \in \mg$,  the coefficients $ \LL_I{}^{ J_1 \cdots J_r K}$ are non-holomorphic functions of the moduli of $\Sigma_p$, but this dependence is suppressed throughout;
\item setting $\etah = \hat b$ so that the last argument of  $\cJ_\text{DHS}^{(1,0)}$ equals $\hat b$;
\item setting $X^J = \hat a^J {-} \xih^J$ in (\ref{3.K}) and solving the residue matching equation $[b_I, a^I] = [ \hat b_I , \hat a^I {-} \xih^I]$ of (\ref{3.res-m}), with the help of the linearity of $\cK_\text{E}(x,p;a,b)$ in $a$, the linearity of $\cK(x,p;X,\etah)$ in $X$, and the fact that $\ta,\xih,\etah \in \mg^h$, to obtain a Lie series in $b$
\bea
\label{exphatb1}
\ta^I - \xih^I = a^I + \sum_{r=1}^\infty \MM^{I I_1 \cdots I_r} {}_J  B_{I_1} \cdots B_{I_r} a^J,
\eea 
which determines the coefficients $\MM^{I I_1 \cdots I_r} {}_J $
in terms of the $ \LL_I{}^{ J_1 \cdots J_r K}$ in (\ref{exphatb});
\item identifying
\bea
\label{3.KKE}
\cK_\text{E}(x,p;a,b) = \cK(x,p;\hat a-\xih, \hat b) .
\eea
\end{enumerate}
The residue matching condition provides a convenient relation between $\hat H$ and  $B$, 
\bea
\label{3.HB}
\hat H_J = B_J - \sum_{r=1}^\infty \cM^{K I_1 \cdots I_r}{}_J  \hat H_K B_{I_1} \cdots B_{I_r}, 
\eea
which may be solved iteratively for $\hat H $ as a series in $B$.

{\cor \label{cor:mshuf} The coefficients $\cL$ in the Lie series expansion (\ref{exphatb}) of $\etah_I$ may alternatively be expressed in terms of the following associative series expansion
\bea
\label{3.bsh}
\etah_J = b_J - \sum_{r=2}^\infty \cM_\shuffle^{I_1 \cdots I_r}{}_J  b_{I_1} \cdots b_{I_r}, 
\eea
where the coefficients $\cM_\shuffle^{I_1 \cdots I_r}{}_J $ are given
in terms of the coefficients $\cM^{I_1 \cdots I_r}{}_J$  in the expansion (\ref{exphatb1}) of  $\hat a^I- \xih^I$ by
\begin{align}
\cM_\shuffle^{I_1 \cdots I_r}{}_J  &= \cM^{I_1 \cdots I_r}{}_J  + \sum_{\ell=1}^{r-2} (-1)^\ell \sum_{2\leq j_1 < j_2 < \cdots < j_\ell}^{r-1}
\cM^{I_1 I_2 \cdots I_{j_1}}{}_{K_1}  \cM^{K_1 I_{j_1+1}  \cdots I_{j_2}  }{}_{K_2}  \times  \cdots \notag \\
&\quad\quad\quad\quad\quad\quad\quad
\times \cdots
\cM^{K_{\ell-1} I_{j_{\ell-1}+1}  \cdots I_{j_\ell}  }{}_{K_\ell} 
\cM^{K_\ell I_{j_\ell+1}  \cdots I_{r}  }{}_{J} 
\label{mshuf}
\end{align}
and obey shuffle relations similar to those in (\ref{shprop}),
\bea
\cM_\shuffle^{I_1 \cdots I_r \shuffle K_1 \cdots K_s}{}_J = 0, \ \ \ \ \ \ r,s\geq 1.
\label{mshufrel}
\eea}
The proof of the corollary is given  in appendix \ref{sec:cor}. We reiterate that, even though this is
not exposed in our notation, all of $\cL_J{}^{I_1\cdots I_r},\cM^{I_1\cdots I_r}{}_J,\cM_\shuffle^{I_1\cdots I_r}{}_J$ with $r\geq 2$ depend non-holomorphically on $p$ and the moduli of $\Sigma$.

\subsubsection{Expressing the $g$-differentials in terms of the $f$-differentials}

The explicit relations between the  $g$-differentials and the  $f$-differentials may be obtained by combining the relation (\ref{3.KKE})  with the expansion (\ref{3.K}),
\bea
\cK_\text{E}(x,p;a,b) = \om_J(x) (\hat a^J - \xih^J) + \sum_{r=1}^ \infty h^{I_1 \cdots I_r} {}_J(x,p) \hat H_{I_1} \cdots \hat H_{I_r} ( \hat a^J - \xi^J)
\eea
and then expressing $\hat a - \xih$ in terms of $a,b$ using (\ref{exphatb1}), eliminating $\hat H$ in favor of $B$ using (\ref{3.HB}), and expressing the functions $h$ in terms of the tensors $f$ using (\ref{dhse.u}). 
{\prop
\label{3.prop:5}
Up to rank four, the expressions for $g$ in terms of $f$ are determined by combining the relations between $g$ and $h$ given to rank four by 
 \bea
 \label{gtoh.main} 
g^{I_1}{}_J(x,p) &= & h^{I_1}{}_J(x,p) + \omega_K(x) \MM^{KI_1}{}_J(p),
\no  \\
g^{I_1 I_2}{}_J(x,p) &= & h^{I_1 I_2}{}_J(x,p) 
 \! + \! h^{I_1 }{}_K(x,p) \MM^{KI_2}{}_J(p)
 \no \\ &&
 -  h^{K }{}_J(x,p) \MM^{I_1 I_2}{}_K(p)
\! + \! \omega_K(x) \MM^{KI_1 I_2}{}_J(p),
\no \\
g^{I_1 I_2 I_3}{}_J(x,p) & = &
 h^{I_1 I_2 I_3}{}_J(x,p)  
 - \MM^{I_1 I_2}{}_K(p)  h^{K I_3}{}_J(x,p) 
  - \MM^{I_2 I_3}{}_K(p) h^{ I_1 K }{}_J(x,p) 
 \no \\ &&
+ h^{I_1 I_2}{}_K(x,p) \MM^{K I_3}{}_J(p)
- \MM^{I_1 I_2}{}_K(p) h^{ K }{}_L(x,p) \MM^{L I_3}{}_J(p)
\no \\ &&
+ \big \{  \MM^{I_1 I_2}{}_L(p)  \MM^{L I_3}{}_K(p) 
- \MM^{I_1 I_2 I_3}{}_K(p) \big \} h^{ K }{}_J(x,p) 
\no \\ &&
 + h^{ I_1 }{}_K(x,p) \MM^{K I_2 I_3}{}_J(p)
+ \omega_K(x) \MM^{K I_1 I_2 I_3}{}_J(p) ,
 \eea
 and the relations (\ref{dhse.u}) between $h$ and $f$ expanded  to rank four
 \bea
h^{I_1} {}_J (x,p) & = & f^{I_1}{}_J(x,p) + \cT^{I_1}(x,p) \om_J(x) ,
\no \\
h^{I_1 I_2} {}_J(x,p) & = & f^{I_1 I_2} {}_J(x,p) + \cT^{I_1} (x,p) f^{I_2}{}_J(x,p) + \cT^{I_1 I_2}(x,p) \om_J(x),
\no \\
h^{I_1 I_2 I_3} {}_J(x,p) & = & f^{I_1 I_2 I_3} {}_J(x,p) + \cT^{I_1} (x,p) f^{I_2I_3}{}_J(x,p) 
\no \\ &&
+ \cT^{I_1 I_2}(x,p) f^{I_3} {}_J(x,p) + \cT^{I_1 I_2 I_3}(x,p) \om_J(x),
\eea
resulting, for example, in the lowest two ranks in the formulas of  (\ref{intro.41}).}  

\sm

The detailed proof for rank two, including
the calculation of the coefficients $\XX$, $\cM$, $\cT$ and $\cL$ involved in the derivation, will be given in appendix \ref{sec:A}. The proof for rank three will be presented in appendix \ref{sec:B}, while  the derivation for rank four is left to the reader.

\sm

Further details on the implementation of the procedure outlined in this subsection may be found in appendix \ref{sec:B}, including 
\begin{itemize}
\itemsep=0in
\item a more extensive discussion of the symmetry properties of the coefficients
$\LL_I{}^{ J_1 \cdots J_r K}$ and $ \MM^{I I_1 \cdots I_r} {}_J$ in section \ref{sec:B.4}
(see (\ref{shprop}), (\ref{appbb.06}) and section \ref{sec:shprf} for those of $\XX^{I J_1 \cdots J_r}$); 
\item the explicit form of the gauge transformation (\ref{reinstated}) in sections \ref{sec:B.2} and \ref{sec:B.3};  and
\item the detailed construction of the automorphism in sections \ref{sec:B.4} and \ref{sec:B.5}. 
\end{itemize}

\newpage

\section{Gauge transforming $\cK_\text{E}$ to $\cJ_\text{DHS} $ }
\setcounter{equation}{0}
\label{sec:4}

In this section we give the second explicit construction of a gauge transformation and of an automorphism of the Lie algebra $\mg$ that relate the connection $d-\cK_\text{E}$ with the connection $d-\cJ_\text{DHS}$. Although the construction of this section will turn out to be inverse to that of the previous one by the results of section \ref{sec:new5}, its separate construction has advantages for both structural understanding of and practical calculations with the polylogarithm functions associated with $\cK_\text{E}$ and $\cJ_\text{DHS}$, see section \ref{sec:5}.

The construction of the gauge transformation in section \ref{sec:4.1} is divided into two steps. In the first step we exploit the $(0,1)$ component of $\cJ_{\rm DHS}$, which is purely anti-holomorphic, to reproduce the anti-holomorphic part of $d-\cJ_\text{DHS}$. 
In the second step, we then exploit the differential $\cK_\text{E}$ to complete the construction of a gauge transformation $\cU_{\rm E}(x,p)$ which reproduces the desired monodromies. Finally, in section \ref{4.ssecRelCon} we construct a Lie algebra automorphism in the form of an appropriate redefinition $a\cup b\to \ach\cup\bch$ to match the residue of the gauge transformation of $\cK_{\rm E}(x,p;\ach,\bch)$ with that of $\cJ_{\rm DHS}(x,p;a,b)$. Explicit formulas at low degree for the gauge transformation and for the automorphism are given in section \ref{sec4.5}.

\subsection{Construction of the gauge transformation $\cU_{\rm E}$}
\label{sec:4.1}

We now proceed with presenting the two steps of the construction of the gauge transformation $\cU_{\rm E}(x,p)$, which will be combined to define $\cU_{\rm E}(x,p)$ in subsection \ref{sec:5.1.3}.

\subsubsection{Step 1: holomorphicity}
\label{J01sec}

Recall from Theorem \ref{2.thm:2} that $\cJ_{\rm DHS}^{(0,1)}(x,p;a,b)=-\pi b_I\bar\omega^I(x)$. It is therefore independent of the point~$p$ and the generators~$a^J$ of $\mg$, which will be removed from the notation throughout this section. We define $\gminus(x,y;b)$ as the unique solution of
\begin{equation}\label{4.defU-}
d_x\gminus(x,y;b)=\cJ_{\rm DHS}^{(0,1)}(x;b)\,\gminus(x,y;b),
\end{equation}
such that $\gminus(y,y;b)=1$. It can be explicitly constructed as a path-ordered exponential,
\begin{equation}
\gminus(x,y;b)=\text{P} \exp \int _y^x\cJ_{\rm DHS}^{(0,1)}(t;b).
\end{equation}
The role of $\gminus(x,y;b)$ for the construction of the gauge transformation $\cU_{\rm E}(x,p)$ will be to recast the anti-holomorphic part of $\cJ_{\rm DHS}$ out of a differential which is purely holomorphic in $x \in \tilde\Sigma_p$. This can be intuitively understood as follows: consider the differential $\cJ(x,p;a,b)$, defined by\footnote{Note that $\gminus(x,y;b)$ is well-defined also when taking the base-point $y=p$ as will be done in Definition~\ref{nicedef} below.}
\begin{equation}\label{4:defJ}
\cJ(x,p;a,b)=\gminus(x,p;b)^{-1}\cJ^{(1,0)}_{\rm DHS}(x,p;a,b)\,\gminus(x,p;b).
\end{equation}
It follows from its definition that $\cJ$ is a $(1,0)$-form. Moreover, one has
\begin{align}
\gminus(x,p;b)^{-1} & (d_x-\cJ_{\rm DHS}(x,p;a,b))\,\gminus(x,p;b)
\no \\ & =
d_x-\cJ(x,p;a,b) +\gminus(x,p;b)^{-1}d_x\gminus(x,p;b)
\no \\ & \quad 
-\gminus(x,p;b)^{-1}\cJ^{(0,1)}_{\rm DHS}(x,p;a,b)\gminus(x,p;b),
\end{align}
and since by the definition \eqref{4.defU-} of $\gminus$ the last two terms cancel each other, we obtain
\begin{equation}\label{4.gauge1}
d_x-\cJ(x,p;a,b)= \gminus(x,p;b)^{-1} \,\big(d_x-\cJ_{\rm DHS}(x,p;a,b) \big)\,\gminus(x,p;b).
\end{equation} 
Since gauge transformations preserve flatness, it follows that $d_x-\cJ(x,p;a,b)$ is flat which, combined with the fact that $\cJ$ is a $(1,0)$-form, implies that $\cJ$ must be purely holomorphic in $x \in \tilde\Sigma_p$ (but not in $p$ and the moduli of the surface), namely $\gminus(x,p;b)$ can be used to gauge transform $\cJ_\text{DHS}$ to a holomorphic $(1,0)$-form in $x \in \tilde\Sigma_p$.

\sm

Let us also consider the monodromy representation (see section \ref{sec:1.1}) 
\bea\label{5:eqmumin}
\mu_{-} (\gamma,y;b) = \gminus(\gamma \cdot y,y;b),
\eea
so that one has
\begin{equation}\label{eqmongminus}
\gminus(\gamma  \cdot x,y;b)=\gminus(x,y;b)\,\mu_{-} (\gamma,y;b).
\end{equation}
It follows from its definition and from the single-valuedness of $\cJ_{\rm DHS}$ that $\cJ$ is multiple-valued in~$x$, with monodromies given by
\begin{align}
\label{4.monJ}
\cJ(\mA^K \cdot x ,p;a,b)&=\mu_{-} (\mA^K,p;b)^{-1}\cJ(x,p;a,b)\,\mu_{-} (\mA^K,p;b), 
\notag \\
\cJ(\mB_K \cdot x,p;a,b)&=\mu_{-} (\mB_K,p;b)^{-1}\cJ(x,p;a,b)\,\mu_{-} (\mB_K,p;b).
\end{align}
The idea of the second step below will then be to gauge transform $\cK_{\rm E}$ into a holomorphic $(1,0)$-form in $x \in \tilde\Sigma_p$ with the same monodromy properties \eqref{4.monJ} as $\cJ$, because then we know that by further applying the gauge transformation $\gminus(x,p;b)^{-1}$ we would obtain a single-valued smooth differential with the same $(0,1)$ component as $\cJ_{\rm DHS}$, thus matching two fundamental properties which uniquely characterize $\cJ_{\rm DHS}$ in Theorem \ref{2.thm:2}.

\subsubsection{Step 2: monodromies}

Let us now consider the solution $\gE(x,y,p;\xi,\eta)$  to the differential equation
\bea\label{4:defUE}
d_x \, \gE(x,y,p; \xi, \eta) & = & \cK_{\rm E}(x,p;\xi,\eta) \,\gE(x,y,p;\xi,\eta),
\eea
along with the initial condition $\gE (y,y,p; \xi, \eta)=1$. As in section \ref{sec:3}, $\xi$ and $\eta$ are arbitrary elements of $\mg^h$, so that $\cK_{\rm E}$ and $\cU_{\rm E}$ take values in $\mg$ and $\exp(\mg)$, respectively, and we will later make a specific choice by requiring $\gE(x,y,p; \xi, \eta)$ to satisfy suitable monodromy properties. The function $\gE$ can be explicitly constructed as the path-ordered integral
\bea\label{eqtoday}
\gE(x,y,p; \xi, \eta) = \text{P} \exp \int _y^x \cK_{\rm E}(\ti ,p;\xi, \eta),
\eea
the integral being taken as usual in the universal cover $\tilde\Sigma_p$, and can be seen as a multiple-valued function of $x,y,p\in\Sigma$, with logarithmic singularities at $x=p$ and $y=p$ and holomorphic elsewhere. 

\sm

Combining the path-concatenation formula \eqref{1.comp} with the monodromies \eqref{2.E1} of $\cK_{\rm E}$, one obtains that, for every $K=1,\cdots ,h$,
\begin{align}
\label{4.monUE}
\gE(\mA^K \cdot x ,y,p;\xi, \eta)&=\gE(x,y,p;\xi,\eta)\,\mu_{\rm E} (\mA^K,y,p;\xi, \eta), 
\notag\\
\gE(\mB_K \cdot x ,y,p;\xi, \eta)&=e^{-2\pi i\eta_K}\,\gE(x,y,p;\xi,\eta)\,e^{2\pi i\eta_K}\,\mu_{\rm E} (\mB_K,y,p;\xi,\eta),
\end{align}
where for $\gamma\in\pi_1(\Sigma_p, \qq)$ we define $\mu_{\rm E} (\gamma,y,p;\xi, \eta)$ to be the element\footnote{Notice that, since $\cK_{\rm E}$ is multiple-valued, the map $\mu_{\rm E}:\pi_1(\Sigma_p, \qq)\to\exp(\mg)$ given by setting $\mu_{\rm E}(\gamma)=\mu_{\rm E} (\gamma,y,p;\xi, \eta)$ does not preserves multiplication, and is therefore not a homomorphism.\label{4.footnote}} of $\exp(\mg)$ given by
\bea
\mu_{\rm E} (\gamma,y,p;\xi, \eta) = \gE (\gamma \cdot y,y,p;\xi, \eta).
\eea
As will become clearer in the sequel, one might be tempted  to impose the condition that the monodromies of the differential
\begin{equation}
\gE(x,y,p;\xi,\eta)^{-1}d_x\gE(x,y,p;\xi,\eta)=\gE(x,y,p;\xi,\eta)^{-1}\cK_{\rm E}(x,p;\xi,\eta)\,\gE(x,y,p;\xi,\eta)
\end{equation}
should match the monodromies \eqref{4.monJ} of the differential $\cJ$ defined in \eqref{4:defJ}.
If the path-ordered exponentials involved had the same base-point, one would readily obtain the desired condition by combining \eqref{2.E1} with \eqref{4.monUE}. However, in the absence of further constraints on $\xi$ and $\eta$, the path-ordered exponential in \eqref{eqtoday} diverges for $y=p$, thereby invalidating this approach. Instead, we shall seek to specialize $\xi$ and $\eta$ to ($p$ and $y \not= p$ dependent) elements $\xich(y,p),\etach(y,p)$ of $\mg^h$ which satisfy the monodromy conditions adapted to $\cJ$ at integration base-point $y$, namely that, for every $K=1,\cdots ,h$,
\begin{align}\label{4.monconds}
\mu_{\rm E} \big(\mA^K,y,p;\xich(y,p), \etach(y,p) \big)&=\mu_{-}(\mA^K,y;b),\notag \\
e^{2\pi i\etach_K}\mu_{\rm E} \big(\mB_K,y,p;\xich(y,p), \etach(y,p) \big)&=\mu_{-}(\mB_K,y;b),
\end{align}
with $\mu_-$ given by (\ref{5:eqmumin}).
Since the right sides of these equations are independent of~$a$, the elements $\xich(y,p)$ and $\etach(y,p)$ naturally belong to the subspace $\mg_b^h$.
\begin{lem}\label{4.lemxieta}  For any $y \neq p$ there exists a unique pair of elements $\xich=\xich(y,p),\etach=\etach(y,p)$ of $\mg_b^h$ which satisfy \eqref{4.monconds}. Their components satisfy the identity
\begin{equation}\label{4.eqvanbracket}
[\etach_I,\xich^I]=0.
\end{equation}
\end{lem}
Similarly to the case of Lemma \ref{3.lem:1}, notice that, intuitively, the first statement of this lemma is justified by the fact that there are $2h$ equations for $2h$ unknowns $\xich,\etach$.
\begin{proof}
The proof of the existence and uniqueness of $\xich=\xich(y,p)$ and $\etach=\etach(y,p)$ is similar to the proof of Lemma~\ref{3.lem:1}, because both sides of the equations \eqref{4.monconds} belong to $\exp(\mg)$, and one can recursively (on the degree of the Lie monomials in the components of $b$) prove that their logarithms, which belong to $\mg$, have a unique solution $\xich,\etach\in\mg_b^h$. For the first step of the induction, notice that we have, at order one,
\begin{align}
\log\big(\mu_{\rm E} (\mA^I,y,p;\xich, \etach) \big)&=\xich^I+\mathcal O([\xich,\xich],[\etach,\etach],[\xich,\etach]),\notag\\
\log\big(e^{2\pi i \etach_I} \mu_{\rm E} (\mB_I,y,p;\xich, \etach) \big)&=\Omega_{IJ}\xich^J + 2\pi i \etach_I+\mathcal O([\xich,\xich],[\etach,\etach],[\xich,\etach]),
\label{noeta1}
\end{align}
as well as
\begin{align}
\log\big(\mu_{-} (\mA^I,y;b) \big)&=-\pi Y^{IJ}b_J+\mathcal O([b,b]),\notag\\
\log\big(\mu_{-} (\mB_I,y;b) \big)&=-\pi\bar\Omega_{IK}Y^{KJ}b_J+\mathcal O([b,b]),
\label{noeta2}
\end{align}
which we can equate to find that the components of $\xich,\etach$ that satisfy \eqref{4.monconds} are constrained at order one to be
\bea\label{4.firstordxieta}
\xich^I=-\pi Y^{IJ}b_J+\mathcal O([b,b]),\quad \quad \etach_I=b_I+\mathcal O([b,b]).
\eea

We are now left with proving that these elements $\xich,\etach\in\mg_b^h$ must necessarily satisfy the condition \eqref{4.eqvanbracket}. First of all, even though the map $\mu_{\rm E}:\pi_1(\Sigma_p,\qq)\to \exp(\mg)$ is not a homomorphism (see footnote \ref{4.footnote}), one can verify using the path-concatenation property \eqref{1.comp} of path-ordered exponentials that, for any  $y,p\in\Sigma$ and any $\xi,\eta\in\mg^h$, the map $\tilde\mu_{\rm E}$ defined on the generators of $\pi_1(\Sigma_p,\qq)$ by setting 
\bea
\tilde\mu_{\rm E}(\mA^I)=\mu_{\rm E} (\mA^I,y,p;\xi, \eta),\quad\quad \tilde\mu_{\rm E}(\mB_I)=e^{2\pi i\eta_I}\mu_{\rm E} (\mB_I,y,p;\xi, \eta)
\eea
preserves multiplication and is therefore a homomorphism from $\pi_1(\Sigma_p,\qq)$ to $\exp(\mg_b)$. It follows from the first part of the statement that $\xich$ and $\etach$ are such that
\begin{align}\label{4.eqproof1}
\tilde\mu_{\rm E} (\mA^I)=\mu_{-}(\mA^I),\quad\quad
\tilde\mu_{\rm E} (\mB_I)=\mu_{-}(\mB_I).
\end{align}

Because $\tilde \mu_{\rm E}$ is a homomorphism, one can use the same argument inspired by Cauchy's theorem exploited in the proof of Lemma \ref{3.lem:2} to deduce that
\bea\label{4.eqproof2}
\prod_{I=1}^h\tilde \mu_{\rm E}(\mA^I)\tilde \mu_{\rm E}(\mB_I)\tilde \mu_{\rm E}(\mA^I)^{-1}\tilde \mu_{\rm E}(\mB_I)^{-1}= e^{2\pi i [\etach_I,\xich^I]}
\eea
Similarly, the fact that $\mu_{-}:\pi_1(\Sigma,\qq)\to \exp(\mg_b)$ is a homomorphism implies that 
\bea\label{4.eqproof3}
\prod_{I=1}^h \mu_{-}(\mA^I) \mu_{-}(\mB_I) \mu_{-}(\mA^I)^{-1} \mu_{-}(\mB_I)^{-1}=1.
\eea
The identity \eqref{4.eqvanbracket} then follows from combining equations \eqref{4.eqproof1}, \eqref{4.eqproof2} and \eqref{4.eqproof3}.
\end{proof}

\begin{rmk}
Similar to the case of $\xih(y,p)$ and $\etah(y,p)$ from the first construction in section~\ref{sec:3}, the defining equations (\ref{4.monconds}) for the elements $\xich=\xich(y,p)$ and $\etach=\etach(y,p)$ may in principle lead to a separate dependence on the base-point~$y$ of the $\mA$- and $\mB$-cycles and on $p$ (as well as on the moduli of~$\Sigma$). However, the result (\ref{4.eqvanbracket}) of Lemma \ref{4.lemxieta} will lead to Corollary \ref{4.coroll2} below showing that $\xich$ and $\etach$ do not depend on their second argument $p$.
\end{rmk}

\subsubsection{Combining the two steps}
\label{sec:5.1.3}

The fact that, by Lemma \ref{4.lemxieta}, the elements $\xich(y,p),\etach(y,p)$ of $\mg_b^h$ which satisfy \eqref{4.monconds} are such that $[\etach_I(y,p), \xich^I(y,p)]=0$ implies that the residue of the pole of the differential $\cK_{\rm E}(x,p;\xich(y,p),\etach(y,p))$ at $x=p$ vanishes, and therefore that this differential is holomorphic
on the whole $\tilde\Sigma$. This implies the following analogue of Corollary \ref{3.cor:1}.
\begin{cor}\label{4.cor:1an}
For $\xich=\xich(y,p)$ and $\etach=\etach(y,p)$ obeying \eqref{4.monconds}, $\cK_{\rm E}(x,p;\xich,\etach)$ is independent of its second argument $p$  (and henceforth will be denoted $\cK_{\rm E}(x,\cdot\,;\xich,\etach)$), so that it may depend on $p$ only through the corresponding dependence of $\xich$ and $\etach$.
\end{cor}
\begin{proof} One could set up a proof based on uniqueness arguments of \cite{Enriquez:2011}, but
we shall instead give a constructive proof here: analogously to Corollary \ref{3.cor:1}, one can establish
the independence of $\cK_{\rm E}(x,p;\xich,\etach)$ of its second argument $p$ through the
expansion
\bea
\cK_{\rm E}(x,\cdot\,;\xich,\etach) = \omega_J(x) \xich^J + \sum_{r=1}^{\infty} \varpi^{I_1\cdots I_r}{}_J(x)
\check{H}_{I_1} \cdots \check{H}_{I_r}  \xich^J,
\label{varpieq.1}
\eea
where $\check{H}_I = [\etach_I,X]$ for all $X\in \mg$. Moreover, the differentials 
$\varpi^{I_1\cdots I_r}{}_J(x)$ are the traceless parts
of $g^{I_1\cdots I_r}{}_J(x,p)$ in the following sense
\bea
\varpi^{I_1\cdots I_r}{}_J(x) = g^{I_1\cdots I_r}{}_J(x,p) - \frac{1}{h} \delta^{I_r}_J g^{I_1\cdots I_{r-1}K}{}_K(x,p) 
\label{varpieq.2}
\eea
and independent on $p$ as the notation suggests (see \cite{Enriquez:2011} or section 9.1 of \cite{DHoker:2024ozn}). The manifestly $p$-independent expansion (\ref{varpieq.1}) of $\cK_{\rm E}$ follows from
that in (\ref{2.Kexp}) since the factors of $\delta^{I_r}_J$ produced by the difference $\varpi^{I_1\cdots I_r}{}_J(x) - g^{I_1\cdots I_r}{}_J(x,p)$ give rise to the vanishing commutator $\delta^{I_r}_J \check{H}_{I_r}  \xich^J = [\etach_I , \xich^I] =0$.
\end{proof}

\sm

Thanks to the result of Corollary \ref{4.cor:1an}, we can deduce also the following corollary, whose proof is analogous to that of Corollary \ref{3.coroll2:BIS}, and will therefore be omitted.
\begin{cor}\label{4.coroll2}
The elements $\xich(y,p), \ \etach(y,p)\in\mg_b^h$ determined by Lemma \ref{4.lemxieta} are independent of~$p$. They extend to smooth functions of $y\in\tilde\Sigma$, so that $\xich(y,\cdot), \etach(y,\cdot)$ are well-defined elements of $\mg_b^h$ also for $y=p$ which satisfy (and are uniquely determined by) the monodromy conditions
\bea
\label{4.monp}
\mu_{\rm E} \big(\mA^K,p, \cdot\,;\xich(p,\cdot\,), \etach(p,\cdot\,) \big) & = & \mu_{-}(\mA^K,p;b),
\no \\
e^{2\pi i\etach_K}\mu_{\rm E} \big(\mB_K,p, \cdot\,;\xich(p,\cdot\,), \etach(p,\cdot\,) \big) & = &  \mu_{-}(\mB_K,p;b).
\eea
Here again the dot in the argument of a function is meant to stress the independence of the function on that argument.
\end{cor}

{\deff \label{4deffff} Henceforth, the notation $\xich,\etach$  will be used for the elements $\xich(p,\cdot\,), \etach(p,\cdot\,)$ of $\mg_b^h$  which solve the monodromy conditions
 \eqref{4.monp} from Corollary \ref{4.coroll2} (i.e.\ abandoning the notation of Lemma \ref{4.lemxieta} and Corollary \ref{4.cor:1an}).
}

\begin{deff}
\label{nicedef}
For $\xich$ and $\etach$ as in Definition \ref{4deffff}, we define the gauge transformation $\cU_{\rm E}(x,p) = \cU_{\rm E}(x,p;\xich,\etach)$ as the product of the specialization $\gE(x,p,\cdot\,;\xich,\etach)$ of $\gE(x,y,p;\xi,\eta)$ from \eqref{4:defUE} with the inverse of $\gminus(x,p;b)$, obtained by specializing \eqref{4.defU-} to $y=p$,
\bea\label{4:defgauge}
 \cU_{\rm E}(x,p;\xich,\etach)=\gE(x,p,\cdot\,;\xich,\etach)\,\gminus(x,p;b)^{-1}
\eea
with $\xich^I,\etach_I$ determined as a Lie series in $b_K$ by the following equivalent of (\ref{4.monp}), 
\bea
\label{mero.mon1}
\cU_{\rm E}(\mA^K \cdot p , p; \xich, \etach) & = & 1,
\no \\
 \cU_{\rm E}(\mB_K \cdot p , p; \xich, \etach) & = & 
 e^{-2\pi i \etach_K}.
\eea
\end{deff}
Note that (\ref{mero.mon1}) is the direct analogue of the monodromy conditions (\ref{3.mon1})
that determine the generators $\xih^I,\etah_I$ of $\cU_{\rm DHS}$ in terms of $b_K$.

\subsection{Relating the connections $d-\cK_{\rm E}$ and $d-\cJ_{\rm DHS}$}
\label{4.ssecRelCon}

We are left with constructing suitable elements $\ach, \bch\in\mg^h$ that, combined with the previously constructed gauge transformations, will enable to obtain $\cJ_\text{DHS}(x,p;a,b)$ out of $\cK_{\rm E}(x,p;\ach, \bch)$. As we will see below, a key ingredient is the following statement.

\begin{lem}\label{4.lemtildea}
There is a unique element~$\tilde a$ of $\mg^h$, linear in $a$, which satisfies
\begin{equation}\label{4.eqrestildea}
[b_J,a^J]=[\etach_J,\tilde a^J],
\end{equation}
with $\etach$ as in Lemma \ref{4.lemxieta}, and $a\cup b\to \tilde a\cup \etach$ is a Lie automorphism of~$\mg$.
\end{lem}
\begin{proof}
This result can be proven along similar lines to the proof of Theorem \ref{3.thm:1} (see the discussion on the algorithmic determination of~$\tilde a$ in section \ref{sec4.5}). However,  we will take a different route: to prove existence, we will construct an automorphism~$\theta$ of~$\mg$ and show that its inverse~$\psi$ satisfies $\psi(b_I)=\etach_I$ and is such that the element $\tilde a=\psi(a)$ satisfies the conditions of the statement. We will conclude the proof by showing the uniqueness of~$\tilde a$.

\sm

Recall from the proof of Lemma \ref{4.lemxieta}, equation \eqref{4.firstordxieta}, that $\etach_J=b_J+\mathcal O([b,b])$ for $J=1,\cdots ,h$, with $\etach\in\mg_b^h$. This implies that the map $b\to\etach$ can be inverted and the elements $\etach_J$ generate $\mg_b$, so it can be viewed as an automorphism of $\mg_b$. One can therefore view each $b_J$ as an element of the completed free Lie algebra $\mg_{\etach}\simeq \mg_b$ generated by $\etach$. Since elements of $\mg_{\etach}$ can also be viewed as elements of the (associative) algebra $\mathbb C\langle \! \langle \etach\rangle \! \rangle$ of power series in the (non-commutative) variables $\etach_1,\cdots ,\etach_h$, we can write 
\bea\label{eq:bineta}
b_J=\etach_J+\sum_{r=2}^{\infty}\sigma^{I_1\cdots I_r}{}_J\,\etach_{I_1}\cdots\etach_{I_r},
\eea
with $\sigma^{I_1\cdots I_r}{}_J\in\mathbb C$ depending on the moduli of $\Sigma_p$. If we set 
\bea
s^I{}_J=\sum_{r=2}^{\infty}\sigma^{II_2\cdots I_r}{}_J\,b_{I_2}\cdots b_{I_r}\in \mathbb C\langle \! \langle b\rangle \! \rangle, 
\eea
then since $s^I{}_J=\mathcal O(b)$ one can recursively (on the degree in $b$) construct elements $t^I{}_J\in\mathbb C\langle \! \langle b\rangle \! \rangle$ that satisfy the relation 
\begin{equation}
\label{INVERT}
t^I{}_J+s^I{}_K\, t^K{}_J=\delta^I_{J},
\end{equation}
so that the $h\times h$ matrix $t$ with entries $t^I{}_J$ is the inverse of the matrix $I + s$ with entries $\delta^I_{J}+s^I{}_J$. Up to and including the second order in $b_K$, we find
\bea
t^I{}_J = \delta^I_J - \sigma^{I K}{}_J b_K + (\sigma^{I K_1}{}_L  \sigma^{L K_2}{}_J - \sigma^{I K_1 K_2}{}_J ) b_{K_1} b_{K_2} + {\cal O}(b^3).
\eea
One can extend the usual adjoint action on $\mg$ given by $\mathrm{ad}(X)(Y)=[X,Y]$, with $X,Y\in\mg$, to elements $X=\sum_{r=0}^{\infty}X^{I_1\cdots I_r}b_{I_1}\cdots b_{I_r}$ of $\mathbb C\langle \!\langle b\rangle \!\rangle$ by setting 
\bea\label{4.eq1dec1}
\mathrm{ad}\bigg(\sum_{r=0}^{\infty}X^{I_1 I_2 \cdots I_r}b_{I_1} b_{I_2}\cdots b_{I_r}\bigg)(Y)=\sum_{r=0}^{\infty}X^{I_1 I_2 \cdots I_r}\,[b_{I_1},[b_{I_2},\cdots ,[b_{I_r},Y] \cdots ]].
\eea 
Notice that, if~$X$ is in the first place an element of $\mg$, then one has the non-trivial identity, also known as Dynkin's lemma,
\bea\label{eq:Dynkin}
\bigg[\sum_{r=0}^{\infty}X^{I_1\cdots I_r}b_{I_1}\cdots b_{I_r},\,Y\bigg]=\mathrm{ad}\bigg(\sum_{r=0}^{\infty}X^{I_1\cdots I_r}b_{I_1}\cdots b_{I_r}\bigg)(Y),
\eea
where the right side of the equation is defined by \eqref{4.eq1dec1}.

\sm

Let us now define a Lie algebra endomorphism $\theta$ of $\mg$ by setting $\theta(a^I)=  \mathrm{ad}(t^I{}_J)(a^J)$ and $\theta(\etach_I)=b_I$, with $\mathrm{ad}(t^I{}_J)(a^J)$ as in \eqref{4.eq1dec1}. As remarked above, the assignment $\etach\to b$ induces an automorphism of $\mg_b$, which combined with the invertibility of the matrix $(t^I{}_J)$ implies that $\theta$ is invertible, and is therefore an automorphism of $\mg$. Notice that $\theta$ preserves the degree in the variables $a^I$ of the generators ($\theta(a^I)$ has degree one, $\theta(b_I)$ has degree zero), hence it preserves the degree of every element. In the remainder we will show that, if we define $\psi$ to be the automorphism given by the inverse of $\theta$, then $\psi$ satisfies all the desired properties. 

\sm

By construction we immediately obtain $\psi(b_I)=\etach_I$. Moreover, notice that, since $\theta$ preserves the $a$-degree, then so does~$\psi$, and therefore $\tilde a=\psi(a)$ is linear in $a$, as requested. We are left to show that eq.\ \eqref{4.eqrestildea} holds for this $\tilde a$, namely that 
\bea
\psi([b_J,a^J])=[b_J,a^J].
\eea

As a first step, notice that we can combine the formula \eqref{eq:bineta} with Dynkin's lemma \eqref{eq:Dynkin} to obtain the identity
\begin{align}\label{(*):0805}
[b_J,a^J]=[\etach_J,a^J]+\sum_{r=2}^{\infty}\sigma^{I_1\cdots I_r}{}_J[\etach_{I_1}\cdots\etach_{I_r},a^J]=\mathrm{ad}\bigg(\etach_J+\sum_{r=2}^{\infty}\sigma^{I_1\cdots I_r}{}_J\,\etach_{I_1}\dots\etach_{I_r}\bigg)(a^J).
\end{align}
Applying $\theta(\etach_I)=b_I$ on these identities we find
\begin{align}
\theta([b_J,a^J])&
=\mathrm{ad}\bigg(b_J+\sum_{r=2}^{\infty}\sigma^{I_1\cdots I_r}{}_J\,b_{I_1}\cdots b_{I_r}\bigg)\big(\theta(a^J) \big)
\notag\\ & =\mathrm{ad}(b_J)\big(\theta(a^J) \big)
+\mathrm{ad}\bigg(\sum_{r=2}^{\infty}\sigma^{I I_2\cdots I_r}{}_J\,b_I\, b_{I_2}\cdots b_{I_r}\bigg)\big(\theta(a^J) \big)
\notag\\ & =\mathrm{ad}(b_I)\bigg( \theta(a^I)+\mathrm{ad}\Big(\sum_{r=2}^{\infty}\sigma^{I I_2\cdots I_r}{}_J\,b_{I_2}\cdots 
b_{I_r}\Big) \big(\theta(a^J) \big)\bigg)
\notag\\ & =\mathrm{ad}(b_I)\Big(\theta(a^I)+\mathrm{ad}(s^I{}_J) \big(\theta(a^J) \big)\Big)
=\mathrm{ad}(b_I)\Big(\mathrm{ad}(\delta^I_{J}+s^I{}_J) \big(\theta(a^J) \big) \Big)
\notag\\&
=\mathrm{ad}(b_I)\Big(\mathrm{ad}(\delta^I_{J}+s^I{}_{J}) \big(\mathrm{ad}(t^J{}_{K})(a^K) \big) \Big) 
=\mathrm{ad}(b_I)\Big(\mathrm{ad}\big((\delta^I_{J}+s^I{}_{J})t^J{}_{K}\big)(a^K)\Big)
\notag\\& =\mathrm{ad}(b_I) \big(\mathrm{ad}(\delta^{I}_{K})(a^K) \big)=[b_I,a^I].
\end{align}
Since $\psi\circ\theta=\mathrm{Id}$, applying $\psi$ to the first and the last term of this chain of identities we obtain $[b_J,a^J]=\psi([b_J,a^J])$.

\sm

Finally, let us prove uniqueness of~$\tilde a\in\mg^h$ as in the statement. Suppose that~$\tilde a'\in\mg^h$ satisfies the same properties, then we would get an element $c=\tilde a'-\tilde a\in\mg^h$ which is linear in $a$ and such that $[\etach_J,c^J]=0$, but this is impossible, because the left side would have degree~1 in~$a$, whereas the right side has degree~0.
\end{proof}

We shall now prove that $\cU_{\rm E}(x,p;\xich,\etach)$ indeed provides the gauge transformation required in the relation (\ref{1.JK}) between $\cK_\text{E}(x,p;\ach, \bch)$ and $\cJ_\text{DHS}(x,p;a,b)$, for suitable $\ach, \bch\in~\!\mg^h$ constructed out of the elements $\xich, \etach\in\mg_b^h$ from Definition \ref{4deffff} and of the element $\tilde a\in\mg^h$ from Lemma \ref{4.lemtildea}. The result is summarized in the following theorem.
\begin{thm}\label{4.thm:1}
The flat connections $d_x - \cJ_{\rm DHS} (x,p; a,b) $ and $d_x - \cK _{\rm E} (x,p;\ach,\bch)$ are related by the gauge transformation $\cU_{\rm E}(x,p) = \cU_{\rm E}(x,p;\xich,\etach)$ defined by \eqref{4:defgauge}, whose arguments $\xich, \etach \in \mg_b^h$ are the uniquely determined solutions of (\ref{mero.mon1}), namely we have 
\begin{align}\label{4.maineq}
d_x - \cJ _{\rm DHS} (x,p;a,b)&=\cU_{\rm E} (x,p;\xich,\etach)^{-1} \Big (d_x -\cK_{\rm E} (x,p; \ach,\bch) \Big ) \cU_{\rm E} (x,p;\xich,\etach).
\end{align}

The elements $\ach$ and $\bch$ are defined as $\ach=\tilde a+\xich$ and $\bch=\etach$, with $\tilde a\in\mg^h$ as in Lemma~\ref{4.lemtildea} and $\xich,\etach\in\mg_b^h$ as in Definition \ref{4deffff}.
\end{thm}
\begin{proof}
We want to prove that the connection on the right side of \eqref{4.maineq} satisfies the properties of Theorem \ref{2.thm:2} which uniquely characterize the DHS connection. Recall from its defining equation \eqref{4:defgauge} that $ \cU_{\rm E}(x,p;\xich,\etach)=\gE(x,p,\cdot\,;\xich,\etach)\,\gminus(x,p;b)^{-1}$. 

\sm

First of all, by definition of $\gE$ and by the linearity of $\cK_{\rm E}(x,p;\xi,\eta)$ in its argument $\xi$, one has
\begin{align}
\gE(x,p,\cdot\,;&\, \xich,\etach)^{-1}\cK_{\rm E} (x,p; \tilde a+\xich,\etach) \,\gE(x,p,\cdot\,;\xich,\etach)-\gE(x,p,\cdot\,;\xich,\etach)^{-1}\,d_x\,\gE(x,p,\cdot\,;\xich,\etach) \notag \\
&=\gE(x,p,\cdot\,;\xich,\etach)^{-1}\Big(\cK_{\rm E} (x,p; \tilde a+\xich,\etach)-\cK_{\rm E} (x,p;\xich,\etach)\Big) \,\gE(x,p,\cdot\,;\xich,\etach) \notag \\
&=\gE(x,p,\cdot\,;\xich,\etach)^{-1}\cK_{\rm E} (x,p; \tilde a,\etach) \,\gE(x,p,\cdot\,;\xich,\etach).
\end{align}
Equation \eqref{4.maineq} is therefore equivalent to proving that
\begin{align}\label{4.eqn1}
\cJ _{\rm DHS} (x,p;a,b) &= \gminus(x,p;b)\,\gE (x,p,\cdot\,;\xich,\etach)^{-1}\cK_{\rm E} (x,p; \tilde a,\etach)\, \gE (x,p,\cdot\,;\xich,\etach)\,\gminus(x,p;b)^{-1}\notag\\
&\quad - \gminus(x,p;b)\,d_x\,\gminus(x,p;b)^{-1}.
\end{align}

Notice that the first term on the right side of \eqref{4.eqn1} is a $(1,0)$-form, whereas the second term can be rewritten as $(d_x\,\gminus(x,p;b))\, \gminus(x,p;b)^{-1}$, which by definition of $\gminus$ is equal to $\cJ _{\rm DHS}^{(0,1)} (x,p;a,b)$. This implies that item 2 of Theorem \ref{2.thm:2} is satisfied. Moreover, putting together eqs.\ \eqref{2.E1}, \eqref{eqmongminus}, \eqref{4.monUE} and \eqref{4.monp}, it follows that the first term on the right side of \eqref{4.eqn1} is single-valued as a differential in~$x$, and by construction of $\tilde a$, together with the fact that $\xich,\etach\in\mg^h_b$, it is linear in $a$, thus verifying item 3. Finally, item 1 follows by combining \eqref{4.eqrestildea} with the flatness of the right side of \eqref{4.maineq}, which in turn follows from the flatness of $\cK_{\rm E}$ and the fact that gauge transformations preserve flatness.
\end{proof}

As an intermediate step of the proof, we have verified the validity of \eqref{4.eqn1}, which can be combined with the identity $(d_x\,\gminus(x,p;b))\gminus(x,p;b)^{-1}=\cJ _{\rm DHS}^{(0,1)} (x,p;a,b)$ to obtain the following immediate consequence.
\begin{cor}\label{4.cormainthm}
For $\xich,\etach\in\mg_b^h$ as in Definition \ref{4deffff} and $\tilde a\in\mg^h$ as in Lemma \ref{4.lemtildea}, one has
\begin{equation}
\cJ _{\rm DHS}^{(1,0)} (x,p;a,b)=\cU_{\rm E} (x,p;\xich,\etach)^{-1}\cK_{\rm E} (x,p; \tilde a,\etach)\,\cU_{\rm E} (x,p;\xich,\etach).
\label{DHSfromE}
\end{equation}
\end{cor}
This identity can be exploited to compare the expansion coefficients $f^{I_1  \cdots I_r }{}_J(x,p)$ of $\cJ _{\rm DHS}^{(1,0)}$ with the expansion coefficients $g^{I_1  \cdots I_r }{}_J(x,p)$ of $\cK_{\rm E}$ (see section~\ref{sec4.5}).

\begin{rmk}
\label{blagain}
Similar to Remark \ref{BLrem}, restricting the conjugation (\ref{DHSfromE}) to
genus ${h=1}$ reproduces the gauge transformation (\ref{h1.11}) relating the Brown--Levin and Calaque--Enriquez--Etingof connections. More specifically, it is convenient to rearrange (\ref{h1.12}) into
\bea
 \cJ^{(1,0)}_{\rm DHS}(x,p; a,b) \Big|_{h=1}
=  \cU_{\rm BL} (x{-}p) \,\cK_{\rm E} (x,p;a,b) \Big|_{h=1}  
\cU_{\rm BL} (x{-}p)^{-1} ,
\label{h1gauge}
\eea
where $ \cU_{\rm BL} (x{-}p)=\exp(\frac{2\pi i  b}{\Im \tau} \Im\!(x{-}p))$, and the $(1,0)$-form part of $\cU_{\rm BL} (x{-}p)^{-1} d_x \, \cU_{\rm BL} (x{-}p)$ cancels the admixture of~$b$ in the genus-one generator $\hat a = a + \pi b/(\Im \tau)$. In order to recover (\ref{h1gauge}) from the genus-one instance of
(\ref{DHSfromE}), we specialize the factorized form $\cU_{\rm E} (x,p;\xich,\etach)^{-1} = \gminus(x,p;b)\,\gE(x,p,\cdot\,;\xich,\etach)^{-1} $ of (\ref{4:defgauge}) to $h=1$. Since $\xich = - \pi b/(\Im \tau)$ at genus one, 
the two factors reduce to $\gE(x,p,\cdot\,;\xich,\etach)^{-1}  |_{h=1} = \exp( \frac{\pi b}{\Im \tau}(x{-}p))$ as well as $\gminus(x,p;b)  |_{h=1} =  \exp( {-}\frac{\pi b}{\Im \tau}(\bar x{-} \bar p))$, respectively. As a result, we have $\cU_{\rm E} (x,p;\xich,\etach)^{-1}  |_{h=1}  =  \cU_{\rm BL} (x{-}p)$. Since the element $\tilde a = \psi(a)$ of Lemma \ref{4.lemtildea} reduces to $a$ at genus one\footnote{This follows from the fact
that the expansion coefficients ${\cal R}^{I I_1 \cdots I_r}{}_J$ of $\tilde a^I$ in (\ref{4.expatilde}) below
with $r\geq 1$ are polynomials in coefficients ${\cal Q}_I{}^{J_1 \cdots J_r}$ of a Lie 
series (\ref{expcheck}) which obey
shuffle relations in their upper indices $J_1,\cdots,J_r$ and thereby vanish at genus one.}
this completes our derivation of (\ref{h1gauge}) from the specialization of (\ref{DHSfromE}) to $h=1$.

\sm

Note that the genus-one generators correspond to the placement of uppercase and lowercase indices according to $a= a^1, \ b= b_1$ and $\xich = \xich^1, \ \etah = \etach_1$.
\end{rmk}

\subsection{Iterative construction relating $\cJ_\text{DHS}$ and $\cK_\text{E}$}
\label{sec4.5}

As in the case of the first construction of a gauge transformation, we conclude this section by outlining an algorithmic procedure to iteratively construct the gauge transformation~$\cU_{\rm E}$ and the Lie algebra automorphism $a\, \cup \, b\to\ach \, \cup \, \bch$. The steps of the procedure are similar to those spelled out in section \ref{sec:3.6} (and associated appendices), therefore we will give fewer details and present just a summary, together with the resulting low-degree formulas. 

\sm

The main challenge is to obtain expressions for the first orders of the series expansion in the set of generators $b$ of the elements $\xich,\etach\in\mg_b^h$ from Lemma \ref{4.lemxieta}, as well as the series expansion of $\tilde a\in\mg^h$ from Lemma \ref{4.lemtildea} (we recall that $\ach=\tilde a-\xich$, and $\bch=\etach$). Notice that, if we simply want to relate the differentials $f^{I_1 \cdots I_r} {}_J (x,p)$ and $g^{I_1 \cdots I_r} {}_J (x,p)$, we can bypass Theorem \ref{4.thm:1} and use instead Corollary \ref{4.cormainthm}.

\subsubsection{Second order of $\xich$ and $\etach$ in $b$}

We have already spelled out in the proof of Lemma \ref{4.lemxieta} how to obtain the first order of the solutions $\xich,\etach$ to the system of equations \eqref{4.monp}, which we repeat here for convenience,
\bea\label{4.firstordxieta2}
\xich^I=-\pi Y^{IJ}b_J+\mathcal O(b^2),\quad \quad \etach_I=b_I+\mathcal O(b^2)
\eea
and which takes the same form with the choice of basepoint $y=p$ fixed in Definition \ref{4deffff}.
More generally, we need to determine two families of coefficients $\cP^{I J_1 \cdots J_r}$ and $\cQ_I{}^{J_1\cdots J_r}$, dependent on the moduli of~$\Sigma_p$, such that
\begin{align}
\xich^I&=\sum_{r=1}^{\infty} \cP^{I J_1 \cdots J_r}b_{J_1}\cdots b_{J_r},\notag\\
\etach_I&=\sum_{r=1}^{\infty} \cQ_I{}^{J_1\cdots J_r}b_{J_1}\cdots b_{J_r} 
\label{expcheck}
\end{align}
and such that the monodromy conditions \eqref{4.monp} hold, and we already know by \eqref{4.firstordxieta2} that $\cP^{IJ}=-\pi Y^{IJ}$ and $\cQ_I{}^{J}=\delta_I^J$. Combining the degree-one terms with the first equation in \eqref{4.monp}, one obtains
\bea
\label{eq:43P}
\cP^{I J_1J_2}&= &-2\pi^2
i \, \mathrm{Im} \bigg ( \int_{p}^{\mA^I \cdot \, p } \! \! \! \omega^{J_1}(t_1) \int_p^{t_1}\omega^{J_2}(t_2) \bigg)
\no \\ &&
 +\pi\int_{p}^{\mA^I \cdot \, p }\ \! \! \! \big(  \varpi^{J_1}{}_{K}(t)  Y^{KJ_2}-    \varpi^{J_2}{}_{K}(t)  Y^{KJ_1}\big)
 \notag\\
&= & i\pi^2\bigg \{ \delta^{IJ_1}_{K}Y^{KJ_2}-\delta^{IJ_2}_{K}Y^{KJ_1}
-2 \, \mathrm{Im} \bigg(\int_{p}^{\mA^I \cdot \, p }\omega^{J_1}(t_1)\int_p^{t_1}\omega^{J_2}(t_2)\bigg) \bigg \} ,
\eea
where the traceless part $\varpi^{J_1}{}_{K}(t)$ of the differential $g^{J_1}{}_{K}(t,p)$ is defined in (\ref{varpieq.2}). The second equality in (\ref{eq:43P}) follows from the formula \eqref{A-cycleintg} for the $\mA$-cycle integrals of the $g$-differentials.

\sm

Combining the degree-one terms with the second equation in \eqref{4.monp} one obtains
\begin{align}\label{eq:43Q}
\cQ_I{}^{J_1J_2}&=-\pi \, 
\mathrm{Im}\bigg(\int_{p}^{\mB_I \cdot \, p}\omega^{J_1}(t_1)\int_p^{t_1}\omega^{J_2}(t_2)\bigg)
-\pi i\delta^{J_1}_I\delta^{J_2}_I -\frac{1}{2\pi i}\Omega_{IK}\cP^{KJ_1J_2} 
\notag\\
&\quad +\pi\delta^{J_1}_I\Omega_{IK}Y^{KJ_2}-\frac{i}{2}\int_{p}^{\mB_I \cdot \, p }\big(
  \varpi^{J_1}{}_{K}(t)  Y^{KJ_2} -   \varpi^{J_2}{}_{K}(t)   Y^{KJ_1}\big),
\end{align}
which in turn determines $\cQ_I{}^{J_1J_2}$ upon substituting \eqref{eq:43P} into \eqref{eq:43Q}. Similarly, one may recursively obtain the higher-order terms.

\subsubsection{Implementing the automorphism $a \cup b \to \check a \cup \check b$}

As for the computation of the element $\tilde a\in\mg^h$ from Lemma \ref{4.lemtildea}, one can either follow the steps of the proof of Lemma \ref{4.lemtildea}, or directly use the condition
\bea\label{repeateq}
[b_I,a^I]=[\etach_I,\tilde a^I],
\eea 
together with the Ansatz\footnote{Notice that this is the analogue in this setting of (\ref{exphatb1}).} (implied by imposing linearity of~$\tilde a$ in~$a$) 
\bea
\label{4.expatilde}
\tilde a^I=\sum_{r=0}^{\infty}\cR^{I I_1\cdots I_r}{}_JB_{I_1}\cdots B_{I_r}\,a^J,
\eea
where as usual $B_I X = [b_I, X]$ for all $X \in \mg$, and $\cR^{I I_1\cdots I_r}{}_J$ depend on the moduli of~$\Sigma_p$. We will take here the second route. It follows immediately from \eqref{4.firstordxieta2} and \eqref{4.expatilde} that
\bea
[\etach_I,\tilde a^I]=[b_I,\cR^I{}_J\,a^J]+\mathcal O(b^2),
\eea
from which we deduce that \eqref{repeateq} requires $\cR^I{}_J=\delta^I_J$ and therefore, at first order,
\bea
\tilde a^I=a^I+\mathcal O(b).
\eea
This can in turn be used to compute the second order of the right side of \eqref{repeateq},
\begin{align}
[\etach_I,\tilde a^I]&=[b_I,a^I]+[\cQ_I{}^{J_1J_2}b_{J_1}b_{J_2},a^I]+[b_I,\cR^{I I_1}{}_J[b_{I_1},a^{J}]]+\mathcal O(b^3)\notag\\
&=[b_I,a^I]+\cQ_I{}^{J_1J_2}[b_{J_1},[b_{J_2},a^I]]+\cR^{I I_1}{}_J[b_I,[b_{I_1},a^{J}]]+\mathcal O(b^3),
\end{align}
where in the second equality we made use of Dynkin's lemma \eqref{eq:Dynkin}. Comparing this with the left side of \eqref{repeateq} immediately yields
\bea
\label{newrefer}
\cR^{I I_1}{}_J=-\cQ_J{}^{I I_1},
\eea
with $\cQ_J{}^{I I_1}$ as in \eqref{eq:43Q}, and one can recursively iterate this procedure to obtain the coefficient $\cR^{I I_1\cdots I_r}{}_J$ in terms of $\cQ_J{}^{I_1\cdots I_p}$ and of $\cR^{I I_1\cdots I_q}{}_J$ with $p\leq r$ and $q<r$. In fact, the computations of section \ref{3.sec:3.3} can be
straightforwardly adapted from the case of $[\hat \eta_I , \hat a^I - \xih^I] = [b_I , a^I]$ to the
present case of $[\check \eta_I , \tilde a^I ] = [b_I , a^I]$. More specifically, the results of Corollary
\ref{cor:mshuf} translate into the following all-order relation between 
the expansion coefficients $\cQ_J{}^{I_1\cdots I_r}$ of $\etach_I$ 
and $\cR^{I I_1\cdots I_q}{}_J$ of $\tilde a^I$:
\begin{align}
\cQ_J{}^{I_1 \cdots I_r}  &= - \cR^{I_1 \cdots I_r}{}_J  - \sum_{\ell=1}^{r-2} (-1)^\ell \sum_{2\leq j_1 < j_2 < \cdots < j_\ell}^{r-1}
\cR^{I_1 I_2 \cdots I_{j_1}}{}_{K_1}  \cR^{K_1 I_{j_1+1}  \cdots I_{j_2}  }{}_{K_2} \times   \cdots \notag \\
&\quad\quad\quad\quad\quad\quad\quad
\times \cdots
\cR^{K_{\ell-1} I_{j_{\ell-1}+1}  \cdots I_{j_\ell}  }{}_{K_\ell} 
\cR^{K_\ell I_{j_\ell+1}  \cdots I_{r}  }{}_{J}.
\label{qshuf}
\end{align}
Note that, as a consequence of the expansion (\ref{expcheck}) and $\etah_I \in \mg_b$, the
coefficients $\cQ_J{}^{I_1 \cdots I_r} $ obey shuffle relations (see (\ref{appbb.02}), 
(\ref{appbb.in}) for the shuffle product)
\bea
\cQ_J{}^{I_1 \cdots I_r \shuffle K_1 \cdots K_s} = 0, \ \ \ \ \ \ r,s\geq 1,
\label{qshufrel}
\eea
which via (\ref{qshuf}) imply similar relations with admixtures of lower-rank terms for the
coefficients $\cR^{I I_1\cdots I_q}{}_J$ in (\ref{4.expatilde}).

\subsubsection{Expressing the $f$-differentials in terms of the $g$-differentials}

We are now in the position to apply Corollary \ref{4.cormainthm} to relate the $f$-differentials and the $g$-differentials at degree two\footnote{At degree one, it can easily be checked that Corollary \ref{4.cormainthm} yields the identity $f_J(x,p)=g_J(x,p)$, which is correct because we have seen that both sides are equal to the normalized holomorphic Abelian differential $\omega_J(x)$.}. To do so, first notice that the degree-one computation of $\xich,\etach$ is sufficient to obtain the expression
\bea
\cU_{\rm E}(x,p)=1-2\pi i\,\mathrm{Im}\bigg(\int_p^x\omega^I(t)\bigg) b_I +\mathcal O(b^2).
\eea
Recall that, if $X,Y\in\mg$, then $\exp(X)Y\exp(X)^{-1}=\exp(\mathrm{ad}(X))(Y)$. The expression $\cU_{\rm E} (x,p;\xich,\etach)^{-1}\cK_{\rm E} (x,p; \tilde a,\etach)\,\cU_{\rm E} (x,p;\xich,\etach)$ on the right side of \eqref{DHSfromE} is therefore equal to 
\begin{align}
&g_J(x,p)\tilde a^J+\bigg(g^I{}_J(x,p)+2\pi i\,\mathrm{Im}\bigg(\int_p^x\omega^I(t)\bigg)g_J(x,p) \bigg)[\etach_I,\tilde a^J]+\mathcal O(b^2)\notag\\
&= \omega_J(x)a^J+\bigg(g^I{}_J(x,p)+2\pi i\,\mathrm{Im}\bigg(\int_p^x\omega^I(t)\bigg)\omega_J(x)-\cQ_J{}^{KI}\omega_K(x) \bigg)[b_I,a^J]+\mathcal O(b^2),
\end{align}
where to pass from the first to the second line we used the known low-degree expressions for~$\etach_I$ and~$\tilde a^J$, as well as the fact that $g_J(x,p)=\omega_J(x)$. Corollary \ref{4.cormainthm} then implies that
\begin{align}
f^{I}{}_J(x,p)=g^I{}_J(x,p)+2\pi i\,\mathrm{Im}\bigg(\int_p^x\omega^I(t)\bigg)\omega_J(x)+
\omega_K(x)\cR^{KI}{}_J(p),
\label{newfgrel}
\end{align}
where we have exposed the $p$-dependence of the coefficients $\cR^{KI}{}_J(p)$
to highlight the analogy with the relation between $f^{I}{}_J(x,p)$ and $g^I{}_J(x,p)$ in the
first line of (\ref{intro.41}):
using the expression $\TT^I(x,p)= -2\pi i \mathrm{Im} \int_p^x\omega^I$ by (\ref{appbb.19}),
we find the relation
\bea
\cR^{KI}{}_J(p) = - \cM^{KI}{}_J(p)
\label{cRiscM}
\eea
by matching (\ref{newfgrel}) with the first line of (\ref{intro.41}), where
the explicit form of the coefficient $ \cM^{KI}{}_J(p)$ in the expansion (\ref{exphatb1}) 
can be found in (\ref{appbb.37}) and (\ref{dhse.21}). Note that both sides of
(\ref{cRiscM}) additionally depend on the moduli of $\Sigma$.

\newpage

\section{Uniqueness of gauge relations between $\mathcal K_{\mathrm E}$ and 
$\mathcal J_{\mathrm{DHS}}$}
\setcounter{equation}{0}
\label{sec:new5}

In this section, we shall demonstrate that the pairs of gauge transformations and
Lie algebra automorphisms constructed in the previous sections \ref{sec:3} and \ref{sec:4}
are inverses of one another. 
For this, we develop a formalism that shows that both the flat connections on $\tilde \Sigma_p$ and the elements of the gauge group which act on them can be consistently constrained to enjoy covariance properties under the action of $\boldsymbol{\pi}={\rm Aut}(\tilde \Sigma_p/\Sigma_p)$; this will set the stage for the proof of the uniqueness of such pairs of a gauge element and an automorphism (see Theorem \ref{sec04} and also Remark \ref{whrprf}).

\sm

Recall from section \ref{sec:2} the universal cover map $\tilde\Sigma_p\to\Sigma_p$. 
We denote by $E^k(\tilde\Sigma_p)$ the space of smooth differential $k$-forms on 
$\tilde\Sigma_p$. Recall that $\mathfrak g$ is the topologically free Lie algebra over $2h$ 
generators $(a^1,\ldots,b_h)$, see section \ref{sec:1.2}. It is a complete graded Lie 
algebra, i.e.\ $\mathfrak g=\hat\oplus_{d\geq1}\mathfrak g_d$, where $\mathfrak g_d$ is the degree $d$ part of
$\mathfrak g$ (the generators being of degree 1).

\subsection{Action of the gauge group on flat connections}
 
\begin{deff}
    (a) For $k\geq0$, set $E^k(\tilde\Sigma_p,\mathfrak g):=\hat\oplus_{d\geq1}E^k(\tilde\Sigma_p)\otimes\mathfrak g_d$.

    (b) Denote by $E^0(\tilde\Sigma_p,\mathrm{exp}(\mathfrak g))$ the set of maps 
$\tilde\Sigma_p\to\mathrm{exp}(\mathfrak g)$ of the form $\tilde\Sigma_p \ni x\mapsto e^{v(x)}\in\mathrm{exp}(\mathfrak g)$, 
where $v\in E^0(\tilde\Sigma_p,\mathfrak g)$.
\end{deff} 
 The map 
$E^0(\tilde\Sigma_p,\mathfrak g)\to E^0(\tilde\Sigma_p,\mathrm{exp}(\mathfrak g))$, 
$v\mapsto e^v$ is a bijection, and $E^0(\tilde\Sigma_p,\mathrm{exp}(\mathfrak g))$
is a subgroup of the group of all maps $\tilde\Sigma_p\to\mathrm{exp}(\mathfrak g)$. 

\begin{deff}
Denote by $\mathrm{Flat}(\tilde\Sigma_p)$ the set of flat connections on $\tilde\Sigma_p$, i.e.\ the set of 
$\mathcal J\in E^1(\tilde\Sigma_p,\mathfrak g)$ such that $d\mathcal J=\mathcal J\wedge \mathcal J$ (equality in $E^2(\tilde\Sigma_p,\mathfrak g)$).    
\end{deff}

There is an action of $E^0(\tilde\Sigma_p,\mathrm{exp}(\mathfrak g))$ on $\mathrm{Flat}(\tilde\Sigma_p)$, 
given by $(u,\mathcal J)\mapsto u\bullet \mathcal J:=-ud(u^{-1})+u\mathcal Ju^{-1}$. Then $d-u\bullet \mathcal J=u\circ (d-\mathcal J)\circ u^{-1}$
(equality of maps $E^0(\tilde\Sigma_p,U\mathfrak g)\to E^1(\tilde\Sigma_p,U\mathfrak g)$, where 
$u,u^{-1},\mathcal J,u\bullet \mathcal J$ stand for the left multiplication operators by the corresponding elements). 

\begin{deff}
    (a) For $\theta\in \mathrm{Aut}(\mathfrak g)$ and 
    $u\in E^0(\tilde\Sigma_p,\mathrm{exp}(\mathfrak g))$, define $\theta(u)$ 
    to be the map $\tilde\Sigma_p\ni x\mapsto \theta(u(x))\in \mathrm{exp}(\mathfrak g)$.  The 
    assignment $(\theta,u)\mapsto \theta(u)$ defines an action of $\mathrm{Aut}(\mathfrak g)$ on 
    $E^0(\tilde\Sigma_p,\mathrm{exp}(\mathfrak g))$. 

    (b) The corresponding semidirect group  $E^0(\tilde\Sigma_p,\mathrm{exp}(\mathfrak g))\rtimes
    \mathrm{Aut}(\mathfrak g)$ is the set of pairs $(u,\theta)\in E^0(\tilde\Sigma_p,\mathrm{exp}
    (\mathfrak g))\times \mathrm{Aut}(\mathfrak g)$, equipped with the product $(u,\theta)\cdot(u',\theta')=(u\cdot \theta(u'),\theta\theta')$. 
\end{deff}

\begin{lem}\label{lem14}
    (a) An action of $E^0(\tilde\Sigma_p,\mathrm{exp}(\mathfrak g))\rtimes
    \mathrm{Aut}(\mathfrak g)$ on the set $\mathrm{Flat}(\tilde\Sigma_p)$ is defined by
    $(u,\theta)\bullet\mathcal J:=u\bullet\theta(\mathcal J)$. 

    (b) A group morphism $\mathrm{exp}(\mathfrak g)\to E^0(\tilde\Sigma_p,
    \mathrm{exp}(\mathfrak g))\rtimes    \mathrm{Aut}(\mathfrak g)$ is defined by 
        $c\mapsto (c^{-1},x\mapsto c x c^{-1})$; it is injective and its image is a normal subgroup 
    of its target. We denote by $(E^0(\tilde\Sigma_p,
    \mathrm{exp}(\mathfrak g))\rtimes    \mathrm{Aut}(\mathfrak g))/\mathrm{exp}(\mathfrak g)$
    the corresponding quotient group.  

    (c) The image of the morphism from (b) acts trivially on $\mathrm{Flat}(\tilde\Sigma_p)$
    by the action from (a), so that this action induces an action of 
    $(E^0(\tilde\Sigma_p,
    \mathrm{exp}(\mathfrak g))\rtimes    \mathrm{Aut}(\mathfrak g))/\mathrm{exp}(\mathfrak g)$
    on $\mathrm{Flat}(\tilde\Sigma_p)$. 
\end{lem}

\begin{proof}
    (a) One computes $(u,\theta)\bullet((u',\theta')\bullet\mathcal J):=
    (u,\theta)\bullet u'\bullet\theta'(\mathcal J)=u\bullet\theta(u'\bullet\theta'(\mathcal J))
    =u\bullet(\theta(u')\bullet\theta\theta'(\mathcal J))=(u\cdot \theta(u'))\bullet\theta\theta'(\mathcal J)
    =(u\cdot \theta(u'),\theta\theta')\bullet\mathcal J$. 

    (b) For $c\in \mathrm{exp}(\mathfrak g)$, set $\mathrm{Ad}_c:=(x\mapsto c x c^{-1})$. 
    Let $\Psi(c):=(c^{-1},\mathrm{Ad}_c)$. For $c,d\in \mathrm{exp}(\mathfrak g)$, one has 
    $\Psi(c)\Psi(d)=(c^{-1},\mathrm{Ad}_c)(d^{-1},\mathrm{Ad}_d)
    =(c^{-1}\mathrm{Ad}_c (d^{-1}),\mathrm{Ad}_c\mathrm{Ad}_d)
    =(d^{-1}c^{-1},\mathrm{Ad}_{cd})=\Psi(cd)$. Therefore $\Psi$ is a group
    morphism. Its injectivity is obvious. For $(u,\theta)\in E^0(\tilde\Sigma_p,
    \mathrm{exp}(\mathfrak g))\rtimes    \mathrm{Aut}(\mathfrak g)$ and 
    $c\in \mathrm{exp}(\mathfrak g)$, one computes $(u,\theta)\cdot \Psi(c)=
    (u,\theta)\cdot (c^{-1},\mathrm{Ad}_c)=(u\theta(c)^{-1},\theta\circ \mathrm{Ad}_c)
    =(\theta(c)^{-1}\mathrm{Ad}_{\theta(c)}(u),\mathrm{Ad}_{\theta(c)}\circ\theta)
    =(\theta(c)^{-1},\mathrm{Ad}_{\theta(c)})\cdot (u,\theta)=\Psi(\theta(c))\cdot (u,\theta)$, which proves the normality of the image of 
    $\Psi$ in its target. 

    (c) The statement follows from the equality $\Psi(c)\bullet \mathcal J=(c^{-1},\mathrm{Ad}_c)\bullet\mathcal J
    :=c^{-1}\bullet\mathrm{Ad}_c(\mathcal J)=\mathcal J$ for any $c \in\mathrm{exp}(\mathfrak g)$
    and $\mathcal J\in\mathrm{Flat}(\tilde\Sigma_p)$. 
\end{proof}

\begin{rmk}
    For $q$ a point in $\tilde\Sigma_p$, let $E^0(\tilde\Sigma_p,\mathrm{exp}(\mathfrak g))_q$
    be the subset of $E^0(\tilde\Sigma_p,\mathrm{exp}(\mathfrak g))$ of all maps $u$ such that 
    $u(q)=1$. Then $E^0(\tilde\Sigma_p,\mathrm{exp}(\mathfrak g))_q$ is a subgroup, stable under the 
    action of $\mathrm{Aut}(\mathfrak g)$, and the semidirect product 
    $E^0(\tilde\Sigma_p,\mathrm{exp}(\mathfrak g))_q\rtimes\mathrm{Aut}(\mathfrak g)$
    is isomorphic to the quotient of Lemma \ref{lem14}(c). 
\end{rmk}

\subsection{Covariant version of the gauge group action}

Let ${\boldsymbol \pi}:=\pi_1(\Sigma_p,y)=\mathrm{Aut}(\tilde\Sigma_p/\Sigma_p)$. There are isomorphisms 
${\boldsymbol \pi}^{ab}=H_1(\Sigma_p,\mathbb Z)\simeq
\mathbb Z^{2h}$, where ${\boldsymbol \pi}^{ab}$ is the abelianization of ${\boldsymbol \pi}$. The abelianization 
$(\mathrm{exp}(\mathfrak g))^{ab}$ of $\mathrm{exp}(\mathfrak g)$ is 
$\mathfrak g_1=(\oplus_{I=1}^h\mathbb Ca^I)\oplus
(\oplus_{I=1}^h\mathbb Cb_I)$. The abelianization of a 
group morphism $\alpha : {\boldsymbol \pi}\to \mathrm{exp}(\mathfrak g)$
is a morphism of abelian groups $\alpha^{ab} : {\boldsymbol \pi}^{ab}\to (\mathrm{exp}(\mathfrak g))^{ab}=\mathfrak g_1$, which 
gives rise to a linear map $\alpha^{ab}\otimes\mathbb C : {\boldsymbol \pi}^{ab}\otimes\mathbb C
\to (\mathrm{exp}(\mathfrak g))^{ab}=\mathfrak g_1$. 

\begin{deff}
    (a) $\mathrm{Hom}({\boldsymbol \pi},\mathrm{exp}(\mathfrak g))$ is 
the set of group morphisms ${\boldsymbol \pi}\to\mathrm{exp}(\mathfrak g)$, i.e.\ of maps $\alpha$ such that 
$\alpha(\gamma\delta)=\alpha(\gamma)\alpha(\delta)$.

(b) The group $\mathrm{exp}(\mathfrak g)$ acts by conjugation on the set
$\mathrm{Hom}({\boldsymbol \pi},\mathrm{exp}(\mathfrak g))$, namely 
$c \bullet\rho:=(\gamma\mapsto c \rho(\gamma)c^{-1})$; we denote by 
$\mathrm{Hom}({\boldsymbol \pi},\mathrm{exp}(\mathfrak g))/\mathrm{exp}(\mathfrak g)$
the corresponding quotient set. 

(c) $\alpha\in \mathrm{Hom}({\boldsymbol \pi},\mathrm{exp}(\mathfrak g))$ is called nondegenerate if
its abelianization $\alpha^{ab}\otimes\mathbb C : {\boldsymbol \pi}^{ab}\otimes\mathbb C
\to\mathfrak g_1$ is a
vector space isomorphism.  

(d) A class in $\mathrm{Hom}({\boldsymbol \pi},\mathrm{exp}(\mathfrak g))/\mathrm{exp}(\mathfrak g)$ 
is called nondegenerate if one of its 
representative (equivalently, all of them) is nondegenerate.  
\end{deff}

Then ${\boldsymbol \pi}$ acts on $E^\bullet(\tilde\Sigma_p)$; the action $(\gamma,f)\mapsto \gamma^*f$ is such that 
$(\gamma\delta)^*f=\delta^*\gamma^*f$. For example, the action on functions is given by 
$(\gamma^*f)(x):=f(\gamma x)$. 

For $\alpha,\beta\in\mathrm{Hom}({\boldsymbol \pi},\mathrm{exp}(\mathfrak g))$, define 
$E^0(\tilde\Sigma_p,\mathrm{exp}(\mathfrak g))_{\alpha,\beta}$ as the subset of 
$E^0(\tilde\Sigma_p,\mathrm{exp}(\mathfrak g))$ of elements $u$ such that 
$\gamma^*u=\alpha(\gamma)\cdot u\cdot \beta(\gamma)^{-1}$. Then the product in 
$E^0(\tilde\Sigma_p,\mathrm{exp}(\mathfrak g))$ is such that 
\begin{equation}\label{pdt:auto}
E^0(\tilde\Sigma_p,\mathrm{exp}(\mathfrak g))_{\alpha,\beta}
\cdot E^0(\tilde\Sigma_p,\mathrm{exp}(\mathfrak g))_{\beta,\lambda}
\subset
E^0(\tilde\Sigma_p,\mathrm{exp}(\mathfrak g))_{\alpha,\lambda}
\end{equation}
for any $\alpha,\beta,\lambda$. 

For $\alpha\in\mathrm{Hom}({\boldsymbol \pi},\mathrm{exp}(\mathfrak g))$, let 
$\mathrm{Flat}_\alpha(\Sigma_p)$ be the set of $\mathcal J\in E^1(\tilde\Sigma_p,\mathfrak g)$ such that 
$d\mathcal J=\mathcal J\wedge \mathcal J$ and 
$\gamma^*\mathcal J=\alpha(\gamma)\mathcal J\alpha(\gamma)^{-1}$ for any $\gamma\in {\boldsymbol \pi}$. The action of 
$E^0(\tilde\Sigma_p,\mathrm{exp}(\mathfrak g))$ on $\mathrm{Flat}(\tilde\Sigma_p)$ is then such that 
\begin{equation}\label{act:gpoid}
E^0(\tilde\Sigma_p,\mathrm{exp}(\mathfrak g))_{\alpha,\beta}
\bullet \mathrm{Flat}(\tilde\Sigma_p)_\beta\subset \mathrm{Flat}(\tilde\Sigma_p)_\alpha.  
\end{equation}
For any $\alpha\in\mathrm{Hom}({\boldsymbol \pi},\mathrm{exp}(\mathfrak g))$, a monodromy map 
\bea
\mu_\alpha : \mathrm{Flat}(\tilde\Sigma_p)_\alpha\to 
\mathrm{Hom}({\boldsymbol \pi},\mathrm{exp}(\mathfrak g))/\mathrm{exp}(\mathfrak g)
\eea
is defined by the assignment $\mathcal J\mapsto [\gamma\mapsto(\gamma^*{\boldsymbol \Gamma}_{\mathcal J})^{-1}\alpha(\gamma){\boldsymbol \Gamma}_{\mathcal J}]$, where 
${\boldsymbol \Gamma}_{\mathcal J}$ is a solution of $d{\boldsymbol \Gamma}_{\mathcal J}={\mathcal J}{\boldsymbol \Gamma}_{\mathcal J}$, where $[-]$ denotes the class modulo the action of 
$\mathrm{exp}(\mathfrak g)$. It is such that for any $\alpha,\beta\in
\mathrm{Hom}({\boldsymbol \pi},\mathrm{exp}(\mathfrak g))$, any ${\mathcal J}\in \mathrm{Flat}(\tilde\Sigma_p)_\beta$ and
any $u\in
E^0(\tilde\Sigma_p,\mathrm{exp}(\mathfrak g))_{\alpha,\beta}$, one has
\bea
\mu_\alpha(u\bullet {\mathcal J})=\mu_\beta({\mathcal J}). 
\eea

\begin{rmk}
\label{nwrmk}
For any $\alpha$ in ${\rm Hom}(\pi_1,\exp(\mathfrak{g}))$, an element of ${\rm Flat}(\tilde \Sigma_p)_\alpha$ is the same as a flat section of the principal $\exp(\mathfrak{g})$-bundle over $\Sigma_p$ defined by $\alpha$.
\end{rmk}

\begin{rmk}
Equation \eqref{pdt:auto} says that 
$\sqcup_{\alpha,\beta}E^0(\tilde\Sigma_p,\mathrm{exp}(\mathfrak g))_{\alpha,\beta}$
forms a groupoid, and \eqref{act:gpoid} defines its groupoid action on 
$\sqcup_{\beta}\mathrm{Flat}(\tilde\Sigma_p)_\beta$. 
\end{rmk}

\subsection{Stabilizers of nondegenerate connections}

Denote by $\mathbf 1 : {\boldsymbol \pi}\to \mathrm{exp}(\mathfrak g)$ the trivial representation. 
One has: 

\begin{thm}\label{sec04}
(a) If  ${\mathcal J}\in \mathrm{Flat}(\tilde\Sigma_p)_{\mathbf 1}$ is such that 
$\mu_{\mathbf 1}({\mathcal J})$ is nondegenerate and
$u\in E^0(\tilde\Sigma_p,\mathrm{exp}(\mathfrak g))_{\mathbf 1,\mathbf 1}$ is such that ${\mathcal J}=u\bullet {\mathcal J}$, then $u=1$.

(b) if ${\mathcal J}\in \mathrm{Flat}(\tilde\Sigma_p)_{\mathbf 1}$ is such that 
$\mu_{\mathbf 1}({\mathcal J})$ is nondegenerate and if 
$(\theta,u)\in\mathrm{Aut}(\mathfrak g)\times E^0(\tilde\Sigma_p,\mathrm{exp}(\mathfrak g))_{\mathbf 1,\mathbf 1}$ 
is such that 
${\mathcal J}=u\bullet \theta({\mathcal J})$, then $u$ is constant, i.e.\ $u\in\mathrm{exp}(\mathfrak g)$, 
and $\theta$ is the inner automorphism $x\mapsto u^{-1}xu$.

(c) if ${\mathcal J},{\mathcal J}'\in \mathrm{Flat}(\tilde\Sigma_p)_{\mathbf 1}$ are such that 
$\mu_{\mathbf 1}({\mathcal J})$ and $\mu_{\mathbf 1}({\mathcal J}')$ are equal and nondegenerate, 
then there exists a unique $u\in E^0(\tilde\Sigma_p,\mathrm{exp}(\mathfrak g))_{\mathbf 1,\mathbf 1}$ such that ${\mathcal J}'=u\bullet {\mathcal J}$.

(d) if ${\mathcal J}\in \mathrm{Flat}(\tilde\Sigma_p)_{\mathbf 1}$, 
$\alpha\in\mathrm{Hom}({\boldsymbol{\pi}},\mathrm{exp}(\mathfrak g))$ and 
${\mathcal K}\in \mathrm{Flat}(\tilde\Sigma_p)_{\alpha}$
are such that 
$\mu_{\mathbf 1}({\mathcal J})$ and $\mu_{\mathbf \alpha}({\mathcal K})$ are nondegenerate, 
then there exists $\theta\in\mathrm{Aut}(\mathfrak g)$ and 
$u\in E^0(\tilde\Sigma_p,\mathrm{exp}(\mathfrak g))_{\alpha,\mathbf 1}$ such that 
${\mathcal K}=u\bullet \theta({\mathcal J})$.
\end{thm}

\begin{proof}
(a) Set $v:=\mathrm{log}u$ and let $v=v_1+v_2+\cdots$ be the expansion of $v$. Let us show by induction on $n\geq1$ the statement $S(n)$: ``$v_i=0$ for $i<n$ and $v_n$ is constant''. 

   The degree $1$ part of the equality  ${\mathcal J}=u\bullet {\mathcal J}$ implies $dv_1=0$, therefore $v_1$ is constant; this
   implies $S(1)$. 

   Assume $S(n)$, so $v=v_n+\cdots$. The degree $n+1$ part of ${\mathcal J}=u\bullet {\mathcal J}$ implies $dv_{n+1}+[v_n,{\mathcal J}_1]=0$. 
   It follows that for any loop $\gamma$ in $\Sigma_p$, one has $\int_\gamma[v_n,{\mathcal J}_1]=0$. Since $v_n$ is constant, 
   this implies $[v_n,\int_\gamma {\mathcal J}_1]=0$. Since the linear span of the values $\int_\gamma {\mathcal J}_1$ where $\gamma$
   runs over all loops generates $\mathfrak g$ and since $\mathfrak g$ has trivial center, this implies $v_n=0$. 
   The equality $dv_{n+1}+[v_n,{\mathcal J}_1]=0$ then implies that $v_{n+1}$ is constant. This proves $S(n+1)$. 

   The conjunction of all the statements $S(n)$ for $n\geq1$ implies $u=1$.  This proves (a). 

(b) Applying $\mu_{\mathbf 1}$ to 
${\mathcal J}=u\bullet \theta({\mathcal J})$ and using the invariance of
$\mu_{\mathbf 1}$ under the action of 
$E^0(\tilde\Sigma_p,\mathrm{exp}(\mathfrak g))_{\mathbf 1,\mathbf 1}$, one obtains 
$\mu_{\mathbf 1}({\mathcal J})=\mu_{\mathbf 1}(\theta({\mathcal J}))$, which implies that for any representative 
$\rho\in \mathrm{Hom}({\boldsymbol \pi},\mathrm{exp}(\mathfrak g))$ of $\mu_{\mathbf 1}({\mathcal J})$ there exists $a\in \mathrm{exp}(\mathfrak g)$
such that $\rho=\mathrm{Ad}_a\circ\theta\circ\rho$. Since~$\rho$ is nondegenerate, this implies 
$\mathrm{Ad}_a\circ\theta=id$, therefore $\theta$ is inner. One then uses~(a). 
   
(c) By assumption, there exist solutions ${\boldsymbol \Gamma}_{\mathcal J}$ and ${\boldsymbol \Gamma}_{{\mathcal J}'}$ of $d{\boldsymbol \Gamma}_{\mathcal J}={\mathcal J}{\boldsymbol \Gamma}_{\mathcal J}$ and 
$d{\boldsymbol \Gamma}_{{\mathcal J}'}={\mathcal J}'{\boldsymbol \Gamma}_{{\mathcal J}'}$, such that for any $\gamma\in{\boldsymbol \pi}$, one has 
$(\gamma^*{\boldsymbol \Gamma}_{\mathcal J})^{-1}{\boldsymbol \Gamma}_{\mathcal J}=(\gamma^*{\boldsymbol \Gamma}_{{\mathcal J}'})^{-1}{\boldsymbol \Gamma}_{{\mathcal J}'}$. It follows that $u:={\boldsymbol \Gamma}_{{\mathcal J}'}{\boldsymbol \Gamma}_{\mathcal J}^{-1}$ is 
${\boldsymbol \pi}$-invariant, therefore $u\in E^0(\tilde\Sigma_p,\mathrm{exp}(\mathfrak g))_{\mathbf 1,\mathbf 1}$. 
Then $d-{\mathcal J}={\boldsymbol \Gamma}_{\mathcal J}\circ d\circ {\boldsymbol \Gamma}_{\mathcal J}^{-1}$ and $d-{\mathcal J}'={\boldsymbol \Gamma}_{{\mathcal J}'}\circ d\circ {\boldsymbol \Gamma}_{{\mathcal J}'}^{-1}$ imply
$d-{\mathcal J}'={\boldsymbol \Gamma}_{{\mathcal J}'}\circ d\circ {\boldsymbol \Gamma}_{{\mathcal J}'}^{-1}=u{\boldsymbol \Gamma}_{{\mathcal J}}\circ d\circ (u{\boldsymbol \Gamma}_{{\mathcal J}})^{-1}=u\circ (d-{\mathcal J})\circ u^{-1}$, therefore 
${\mathcal J}'=u\bullet {\mathcal J}$. This proves the existence of $u$. Its uniqueness follows from (a). 


Let us prove (d). Let ${\boldsymbol \Gamma}_{\mathcal J},{\boldsymbol \Gamma}_{\mathcal K}$ be solutions of $d{\boldsymbol \Gamma}_{\mathcal J}={\mathcal J}{\boldsymbol \Gamma}_{\mathcal J}$ and 
$d{\boldsymbol \Gamma}_{\mathcal K}={\mathcal K}{\boldsymbol \Gamma}_{\mathcal K}$. The corresponding representations ${\boldsymbol \pi}\to\mathrm{exp}(\mathfrak g)$
are respectively $\gamma\mapsto (\gamma^*{\boldsymbol \Gamma}_{\mathcal J})^{-1}{\boldsymbol \Gamma}_{\mathcal J}$ (denoted~$\mu_{\mathcal J}$) and 
$\gamma\mapsto  (\gamma^*{\boldsymbol \Gamma}_{\mathcal K})^{-1}\alpha(\gamma){\boldsymbol \Gamma}_{\mathcal K}$.
Since both representations are nondegenerate, there exists $\theta\in\mathrm{Aut}(\mathfrak g)$
such that 
\begin{equation}\label{eq1}
    (\gamma^*{\boldsymbol \Gamma}_{\mathcal K})^{-1}\alpha(\gamma){\boldsymbol \Gamma}_{\mathcal K}=\theta((\gamma^*{\boldsymbol \Gamma}_{\mathcal J})^{-1}{\boldsymbol \Gamma}_{\mathcal J})
    \text{ for any }\gamma\in\boldsymbol{\pi}. 
\end{equation}
On the other hand, since $\mu_{\mathcal J}$ is nondegenerate, there exists
a Lie algebra endomorphism~$\varphi$ of $\mathfrak g$ such that $\alpha=\varphi\circ \mu_{\mathcal J}$
(equality of group morphisms $\boldsymbol{\pi}\to\mathrm{exp}(\mathfrak g)$). 
One has ${\boldsymbol \Gamma}_{\mathcal J}\in E^0(\tilde\Sigma_p,\mathrm{exp}(\mathfrak g))_{\mathbf 1,\mu_{\mathcal J}}$, hence 
$v:=\varphi({\boldsymbol \Gamma}_{\mathcal J})={\boldsymbol \Gamma}_{\varphi({\mathcal J})}\in E^0(\tilde\Sigma_p,
\mathrm{exp}(\mathfrak g))_{\mathbf 1,\alpha}$. 
Since ${\mathcal K}\in\mathrm{Flat}(\tilde\Sigma_p)_\alpha$, it follows that 
$v\bullet {\mathcal K}\in \mathrm{Flat}(\tilde\Sigma_p)_{\mathbf 1}$. 
The representations $\boldsymbol{\pi}\to\mathrm{exp}(\mathfrak g)$
attached to $v\bullet {\mathcal K}$ and~${\mathcal K}$ are the same, therefore 
\begin{equation}\label{eq2}
(\gamma^*{\boldsymbol \Gamma}_{v\bullet {\mathcal K}})^{-1}{\boldsymbol \Gamma}_{v\bullet {\mathcal K}}=(\gamma^*{\boldsymbol \Gamma}_{\mathcal K})^{-1}\alpha(\gamma){\boldsymbol \Gamma}_{\mathcal K}
    \text{ for any }\gamma\in\boldsymbol{\pi}. 
\end{equation}
The conjunction of \eqref{eq1} and \eqref{eq2} implies 
$(\gamma^*{\boldsymbol \Gamma}_{v\bullet {\mathcal K}})^{-1}{\boldsymbol \Gamma}_{v\bullet {\mathcal K}}=
(\gamma^*{\boldsymbol \Gamma}_{\theta({\mathcal J})})^{-1}{\boldsymbol \Gamma}_{\theta({\mathcal J})})$, therefore 
$\mu_{\mathbf 1}(v\bullet {\mathcal K})=\mu_{\mathbf 1}(\theta({\mathcal J}))$. 
Then (c) implies the existence of $w\in 
E^0(\tilde\Sigma_p,\mathrm{exp}(\mathfrak g))_{\mathbf 1,\mathbf 1}$
such that $\theta({\mathcal J})=w\bullet v\bullet {\mathcal K}$. This implies the claim, 
with $u:=(w\cdot v)^{-1}$. 
\end{proof}

\begin{rmk}\label{rem:pitfalls}
The assumption of single-valuedness of $u$ is essential in Theorem \ref{sec04}(a). Indeed, if 
${\mathcal J}\in \mathrm{Flat}(\tilde\Sigma_p)_{\mathbf 1}$ then nontrivial elements 
$u\in E^0(\tilde\Sigma_p,\mathrm{exp}(\mathfrak g))$ with ${\mathcal J}=u\bullet {\mathcal J}$ 
can be constructed as 
$u_c:={\boldsymbol \Gamma}_{\mathcal J}c{\boldsymbol \Gamma}_{\mathcal J}^{-1}$, where ${\boldsymbol \Gamma}_{\mathcal J}\in E^0(\tilde\Sigma_p,\mathrm{exp}(\mathfrak g))$ is a 
solution of $d{\boldsymbol \Gamma}_{\mathcal J}={\mathcal J}{\boldsymbol \Gamma}_{\mathcal J}$ and $c\in \mathrm{exp}(\mathfrak g)$. 
\end{rmk}

\begin{rmk}
\label{whrprf}
When $\mathcal J=\mathcal J_{\mathrm{DHS}}$, a direct proof of Theorem \ref{sec04}(b), not appealing
to (a), is as follows. This proof relies on the injectivity of the map 
\begin{equation}\label{crucial:map}
\mathbb C^g\times\mathbb C^g\times E^0(\tilde\Sigma_p)_q\to E^1(\tilde\Sigma_p), \quad 
((u_I)_I,(v^I)_I,w)\mapsto u_I \bar \omega^I+v^I\omega_I+dw, 
\end{equation}
where $q$ is any point of $\tilde\Sigma_p$ and where 
$E^0(\tilde\Sigma_p)_q:=\{w\in E^0(\tilde\Sigma_p)|w(q)=0\}$ (which is established by 
pairing with $H_1(\tilde\Sigma_p,\mathbb Z)$). 

Assume $u\bullet \theta(\mathcal J_{\mathrm{DHS}})
=\mathcal J_{\mathrm{DHS}}$. Let $q\in\tilde\Sigma_p$. Let $\tilde u:=u\cdot u(q)^{-1}$ and 
$\tilde\theta:=\mathrm{Ad}_{\tilde u(q)}\circ \theta$, then one has  $\tilde u(q)=1$ together with 
$\tilde u\bullet \tilde\theta(\mathcal J_{\mathrm{DHS}})=\mathcal J_{\mathrm{DHS}}$, i.e. 
\begin{equation}\label{id:dhs}
\tilde u\cdot \tilde\theta(\mathcal J_{\mathrm{DHS}})\cdot \tilde u^{-1}-\tilde ud(\tilde u^{-1})=\mathcal J_{\mathrm{DHS}} 
\end{equation}
(equality in $E^1(\tilde\Sigma_p,\mathfrak g)$). Projecting \eqref{id:dhs} in  
$E^1(\tilde\Sigma_p,\mathfrak g_1)$, one obtains
$\tilde\theta^{ab}(\mathcal J_{1,\mathrm{DHS}})+d(v_1)=0$, where $v=v_1+v_2+\cdots$
is the expansion of $v:=\mathrm{log}\tilde u$, where $\mathcal J_{1,\mathrm{DHS}}+\cdots$
is the expansion of $\mathcal J_{\mathrm{DHS}}$, and $\tilde\theta^{ab}$ is the abelianization of
$\tilde\theta$, which is a linear endomorphism of $\mathfrak g_1$. Given the 
expansion $\mathcal J_{1,\mathrm{DHS}}=-\pi b_I\overline \omega^I+a^I\omega_I$, 
the latter equation gives 
\bea
\tilde\theta^{ab}(-\pi b_I)\overline \omega^I+\tilde\theta^{ab}(a^I)\omega_I+d(v_1)=0, 
\eea
which in view of the injectivity of \eqref{crucial:map} and $v_1(q)=0$ 
implies $v_1=0$ and $\tilde\theta^{ab}(-\pi b_I)=\tilde\theta^{ab}(a^I)=0$, so that $\tilde\theta^{ab}=0$. 
This equality implies $\tilde\theta=\mathrm{exp}(D)$, where $D$ is a derivation of $\mathfrak g$
of the form $D=D_1+D_2+\cdots$ and $D_k$ is a derivation of degree $k$. 

Let us now show by induction on $k\geq1$ the statement $S(k)$: 
$D_i=0$ for $i<k$ and $v_i=0$ for $i\leq k$. The statement $S(1)$
has already been established. Assume that $S(k)$ holds, then \eqref{id:dhs}
yields
$e^{v_{k+1}+\cdots}\cdot e^{D_k+\cdots}(\mathcal J_{\mathrm{DHS}})\cdot e^{-v_{k+1}-\cdots}
-e^{v_{k+1}+\cdots}d(e^{-v_{k+1}-\cdots})=\mathcal J_{\mathrm{DHS}}$, 
whose degree $k+1$ component is 
\bea
D_k(\mathcal J^1_{\mathrm{DHS}})+d(v_{k+1})=0
\eea
therefore
\bea
D_k(-\pi b_I)\overline \omega^I+D_k(a^I)\omega_I+d(v_{k+1})=0. 
\eea
As before, this together with the injectivity of \eqref{crucial:map} and $v_{k+1}(q)=0$ 
implies $v_{k+1}=0$ and $D_k(-\pi b_I)=D_k(a^I)=0$, so that $D_k=0$; which implies $S(k+1)$. 

It follows that $\tilde u=1$ and $\tilde\theta=id$, therefore that $u=u(q)$ and 
$\theta=\mathrm{Ad}^{-1}_{u(q)}$; this proves Theorem \ref{sec04}(b). 
\end{rmk}

\subsection{Relation of the gauge transformations of sections \ref{sec:3} and \ref{sec:4}} 

Recall that ${\boldsymbol \pi}$ is freely generated by the elements 
$\mathfrak A^K,\mathfrak B_K$, where $K=1,\ldots,h$. 

\begin{deff}
    $\rho\in\mathrm{Hom}({\boldsymbol \pi},\mathrm{exp}(\mathfrak g))$ is the group morphism defined by 
$\mathfrak A^K\mapsto 1$, $\mathfrak B_K\mapsto e^{-2\pi i b_K}$.
\end{deff}

\begin{lem}
    (a) The element $\mathcal K_{\mathrm E}$ defined in section \ref{sec:2.1} is such that 
    $\mathcal K_{\mathrm E}
    \in \mathrm{Flat}(\tilde\Sigma_p)_{\rho}$.

    (b) The element $\mathcal J_{\mathrm{DHS}}$ defined in section \ref{sec:2.2} is such that 
    $\mathcal J_{\mathrm{DHS}}\in \mathrm{Flat}(\tilde\Sigma_p)_{\mathbf 1}$. 
\end{lem}

\begin{proof}
    (a) follows from (\ref{2.E1}) and (b) follows from Theorem \ref{2.thm:2}.
\end{proof}

\begin{deff}
    (a) $\phi\in\mathrm{Aut}(\mathfrak g)$ is the automorphism induced by $a,b\mapsto \hat a,\hat b$ (see (\ref{3.KJ})). 
    
    (b) $u\in E^0(\tilde\Sigma_p,\mathrm{exp}(\mathfrak g))$ is the map 
    $x\mapsto \mathcal U_{\mathrm{DHS}}(x,p;\hat\xi,\hat\eta)^{-1}$ (see (\ref{3.cUb}), (\ref{3.mon1})). 

    (c) $\psi\in\mathrm{Aut}(\mathfrak g)$ is the automorphism induced by $a,b\mapsto \check a,\check b$ 
    (see Theorem \ref{4.thm:1}). 
    
    (d) $v\in E^0(\tilde\Sigma_p,\mathrm{exp}(\mathfrak g))$ is the map 
    $x\mapsto \mathcal U_{\mathrm E}(x,p;\check\xi,\check\eta)^{-1}$ (see (\ref{4:defgauge})). 
\end{deff}

The equalities
\begin{equation}\label{gauge:rels}
u\bullet \phi(\mathcal J_{\mathrm{DHS}})=\mathcal K_{\mathrm E}\text{ and }v\bullet 
\psi(\mathcal K_{\mathrm E})=\mathcal J_{\mathrm{DHS}}
\end{equation}
follow from (\ref{3.eqmon}) and (\ref{4.maineq}), respectively.
 
\begin{lem}\label{lem5dec}
One has 
\begin{equation}\label{toprove}
\text{$u\in E^0(\tilde\Sigma_p,\mathrm{exp}(\mathfrak g))_{\rho,\mathbf 1}$ and 
$v\in E^0(\tilde\Sigma_p,\mathrm{exp}(\mathfrak g))_{\mathbf 1,\psi\rho}$.}     
\end{equation}
\end{lem}

\begin{proof}
The first statement follows from Corollary \ref{dhsrmk}. 
Let us show the second statement. One has 
$v=wz$, where $w:=\boldsymbol{\Gamma}_-(x,p;b)$ and 
$z:=\boldsymbol{\Gamma}_{\mathrm E}(x,p,\cdot ; \check\xi,\check\eta)^{-1}$. Then 
$w\in E^0(\tilde\Sigma_p,\mathrm{exp}(\mathfrak g))_{\mathbf 1,\mu_-(-,p;b)}$ by (\ref{eqmongminus}), where
$\mu_-(-,p;b)\in\mathrm{Hom}({\boldsymbol \pi},\mathrm{exp}(\mathfrak g))$ is the group 
morphism defined by $\gamma\mapsto \mu_-(\gamma,p;b)$ (see (\ref{5:eqmumin})), and 
$z\in E^0(\tilde\Sigma_p,\mathrm{exp}(\mathfrak g))_{\mu_-(-,p;b),\psi\rho}$ by (\ref{4.monUE}) and 
(\ref{4.monp}), which together with \eqref{pdt:auto} implies the statement. 
\end{proof}

\begin{thm}\label{thm:inverses}
    There exists $c\in \mathrm{exp}(\mathfrak g)$ such that 
$\psi\phi=(x\mapsto c x c ^{-1})$ and $v\cdot \psi(u)=c^{-1}$. 
\end{thm}

\begin{proof}
  Lemma \ref{lem5dec} implies 
$v\cdot \psi(u)\in E^0(\tilde\Sigma_p,\mathrm{exp}(\mathfrak g))_{\mathbf 1,\mathbf 1}$. 
It follows from \eqref{gauge:rels} that 
$v\bullet (\psi(u)\bullet \psi\phi(\mathcal J_{\mathrm{DHS}}))=\mathcal J_{\mathrm{DHS}}$ i.e.
\bea
(v\cdot \psi(u))\bullet \psi\phi(\mathcal J_{\mathrm{DHS}})=\mathcal J_{\mathrm{DHS}}.
\eea
The result then follows from Theorem \ref{sec04}(b).    
\end{proof}

Theorem \ref{thm:inverses} shows that the pairs $(u,\phi)$ and $(v,\psi)$, which are the 
elements of $E^0(\tilde\Sigma_p,\mathrm{exp}(\mathfrak g))\rtimes   
\mathrm{Aut}(\mathfrak g)$ relating $\mathcal J_{\mathrm{DHS}}$ and $\mathcal K_{\mathrm{E}}$, 
are such that their classes in the quotient group $(E^0(\tilde\Sigma_p,
    \mathrm{exp}(\mathfrak g))\rtimes    \mathrm{Aut}(\mathfrak g))/\mathrm{exp}(\mathfrak g)$
(see Lemma \ref{lem14}) are inverse to one another. 

\begin{rmk}
\label{finalrem}
  In view of Remark \ref{rem:pitfalls}, Lemma \ref{lem5dec} appears to be a necessary step 
  in the proof of Theorem \ref{thm:inverses}.  
\end{rmk}

\newpage

\section{Comparing spaces of higher-genus polylogarithms}
\label{sec:5}

To compare the higher-genus polylogarithms that arise from the connections $\cJ_{\rm DHS}$ and~$\cK_{\rm E}$, we begin by comparing their generating series $\bGam_\text{E}$ and $\bGam _\text{DHS}$ whose expansions are displayed in (\ref{epoly}) and (\ref{dhspoly}) and given in terms of the path-ordered exponentials of $\cJ_{\rm DHS}$ and $\cK_{\rm E}$, respectively. 
\begin{lem} The generating series of higher-genus polylogarithms
$\bGam_{\rm E}$ and $\bGam _{\rm DHS}$ are related by
\bea
\bGam_{\rm E} (x,y,p;a,b) = \cU_{\rm DHS}(x,p)^{-1} \, \bGam _{\rm DHS}(x,y,p;\ta, \tb) \, \cU_{\rm DHS}(y,p)
\label{5.eq1}
\eea
as well as by
\bea
\bGam_{\rm DHS} (x,y,p;a,b) = \cU_{\rm E}(x,p)^{-1} \, \bGam _{\rm E}(x,y,p;\check{a}, \check{b}) \,  \, \cU_{\rm E}(y,p),
\label{5.eq2}
\eea
with $\cU_{\rm DHS}(x,p), \ta, \tb$ as in section \ref{sec:3} and $\cU_{\rm E}(x,p), \check{a}, \check{b}$ as in section \ref{sec:4}.
\end{lem}
\begin{proof}
It follows from Theorem \ref{3.thm:1} (resp.\ Theorem \ref{4.thm:1}) that both sides of \eqref{5.eq1} (resp.\ \eqref{5.eq2}) satisfy the same differential equation with the same initial conditions.
\end{proof}

Recall from section \ref{sec:1.2} that, given a flat connection $d-\cJ$, one can define a space $\cH(\cJ)$ of polylogarithms associated with $\cJ$, which is the algebra\footnote{The algebra structure is the usual algebra structure of a space of functions.} of smooth multiple-valued functions generated over $\mathbb C$ by all the polylogarithms $\Gamma(\mw;x,y)$, namely the coefficients of the path-ordered exponential $\bGam(x,y;c)=\text{P} \exp \int _y ^x \cJ(\ti;c)$, as specified by (\ref{1.words}). The polylogarithms $\Gamma(\mw;x,y)$ are considered here only as functions of~$x$, namely we fix the Riemann surface and we fix a choice of integration base-point~$y\in\tilde\Sigma_p$, hence the (complex) coefficients of the algebra $\cH(\cJ)$ can be functions of these parameters. As a consequence of the two constructions presented in sections~\ref{sec:3} and~\ref{sec:4}, respectively, we are able to deduce the following precise relation between the algebras of polylogarithms associated with $\cJ_{\rm DHS}$ and~$\cK_{\rm E}$.

\begin{thm}\label{5.thmpolspaces}
One has 
\bea
\label{5.eq3}
\cH(\cJ_{\rm DHS})=\cH(\cK_{\rm E})\cdot\cH(\cJ^{(0,1)}_{\rm DHS}),
\eea
where $\cH(\cJ^{(0,1)}_{\rm DHS})$ denotes the algebra of polynomials in the anti-holomorphic iterated integrals
$\int^x_p \bar \omega_{I_1}(t_1) \int^{t_1}_p\bar \omega_{I_2}(t_2) \cdots \int^{t_{r-1}}_p \bar \omega_{I_r}(t_r)$, $r\geq 1$, with complex coefficients which may depend on the other fixed parameters, such as~$y$ or the moduli of~$\Sigma_p$.
\end{thm}
\begin{proof} The generators of $\cH(\cJ^{(0,1)}_{\rm DHS})$ are immediately deduced from the expansion of the path-ordered 
exponential of $\cJ^{(0,1)}_{\rm DHS}(x;b) = - \pi b_I \bar \omega^I(x)$, see section \ref{J01sec}.

\sm

The inclusion $\cH(\cJ_{\rm DHS})\subset\cH(\cK_{\rm E})\cdot\cH(\cJ^{(0,1)}_{\rm DHS})$ follows by combining the identity \eqref{5.eq2} with the expansions displayed in (\ref{epoly}) and (\ref{dhspoly}), because we recall that $\cU_{\rm E}(x,p)=\gE(x,p,p;\check{\xi},\check{\eta})\,\gminus(x,p;b)^{-1}$, with $\gminus(x,p;b)=\text{P} \exp \int _p ^x \cJ^{(0,1)}_{\rm DHS}(t;b)$ and $\check{\xi},\check{\eta}$ as in 
Definition \ref{4deffff}.

\sm

To conclude the proof we need to show the opposite inclusion $\cH(\cK_{\rm E})\cdot\cH(\cJ^{(0,1)}_{\rm DHS})\subset \cH(\cJ_{\rm DHS})$. This follows by combining the inclusion $\cH(\cK_{\rm E})\subset \cH(\cJ_{\rm DHS})$, which is a consequence of \eqref{5.eq1} and of the fact that $\cU_{\rm DHS}(x,p)=\gDHS(x,p,p;\xih,\etah)$  with $\hat{\xi},\hat{\eta}$ as in 
Definition \ref{3deffff}, with the obvious inclusion $\cH(\cJ^{(0,1)}_{\rm DHS})\subset \cH(\cJ_{\rm DHS})$.
\end{proof}

If we denote by $\mathrm{Hol}(\tilde\Sigma_p)$ the space of all holomorphic multiple-valued functions on~$\Sigma_p$, then we have the following consequence of Theorem \ref{5.thmpolspaces}.
\begin{cor} \label{5.corol}
One has
\bea 
\cH(\cK_{\rm E})=\cH(\cJ_{\rm DHS})\cap \mathrm{Hol}(\tilde\Sigma_p).
\eea
\end{cor}
\begin{proof}
The statement follows by combining \eqref{5.eq3} with the fact that $\cJ^{(0,1)}_{\rm DHS}(x;b)$ is purely anti-holomorphic in $x$.
\end{proof}

\newpage

\appendix

\section{Proof of Lemma \ref{3.lem:4}}
\setcounter{equation}{0}
\label{sec:C}

This appendix is dedicated to establishing properties of the coefficients $ \XX^{I J_1 \cdots J_r} $ in the expansion (\ref{A.Lie}) of the solution $\xih$ to the monodromy conditions $\cU_{\rm DHS} (\mA^K\cdot p, p; \xih, \etah)=1$ 
of (\ref{3.mon1}). Their independence on $p$ and symmetry properties (\ref{shprop}), (\ref{appbb.06}) under permutations of the indices $I,J_1,\cdots, J_r$ will be proven in sections \ref{sec:nopprf} and \ref{sec:shprf}, respectively.

\subsection{Proving independence of $\XX$ on $p$}
\label{sec:nopprf}

\begin{proof} The proof of item 1.\  of Lemma \ref{3.lem:4}, namely the  independence of the coefficients $\XX$ on the point $p$, proceeds via a recursive formula for the derivatives  $\p_p \xih^J$ and $\bar \p _p \xih^J$. These derivatives are obtained by differentiating the $\mA$-monodromy conditions with respect to~$p$, while keeping $\etah$ fixed (the derivative with respect to $\bar p$ is obtained  analogously), 
\bea
\label{A.Jder}
\p_p \, \cU_\text{DHS} (\mA^K \cdot p , p; \xih, \etah) & = &
\int _p ^{\mA^K \cdot \, p }  \cU_\text{DHS} (\mA^K \cdot p, t; \xih, \etah) \cJ_\text{DHS}^{(1,0)}  (t, \cdot \, ; \p_p \xih, \etah)\,
\cU_\text{DHS} (t, p; \xih, \etah) 
\no \\ &&
+ \Big [ \cJ_\text{DHS}^{(1,0)}  (p, \cdot \, ; \xih, \etah) , \cU_\text{DHS} ( \mA^K \cdot p , p; \xih, \etah) \Big ].
 \eea
 Here we have used the linearity of $\cJ_\text{DHS}^{(1,0)}  (t, \cdot \, ;  \xih, \etah) $  in $\xih$  to carry out the derivative of $\cJ_\text{DHS}  (t, \cdot \, ;  \xih, \etah) $ and to regroup the result in the form $\cJ_\text{DHS}^{(1,0)}  (t, \cdot \, ; \p_p \xih, \etah) $.  In view of the monodromy conditions $\cU_\text{DHS} (\mA^K \cdot p, p; \xih, \etah)=1$, the left side of (\ref{A.Jder}) vanishes;  the commutator on the second line of the right side vanishes for the same reason, and the remaining equation may be simplified as follows,
\bea
\label{A.intA}
\int _p ^{ \mA^K \cdot \, p} \cU_\text{DHS}  (t,p; \xih, \etah)^{-1}  \cJ_\text{DHS}^{(1,0)} (t, \cdot \, ; \p_p \xih, \etah) \, 
\cU_\text{DHS} (t, p; \xih, \etah) 
=0.
 \eea
 Both $\cU_\text{DHS} (t,p;\xih,\etah)$ and $\cJ_\text{DHS}^{(1,0)}(t, \cdot \, ; \p_p \xih, \etah)$ admit expansions in powers of $\etah$,
\bea
\label{A.expand}
\cU_\text{DHS} (t,p; \xih, \etah)^{-1} X \cU_\text{DHS} (t, p; \xih, \etah) 
& = & X + \sum_{r=1}^\infty \cT^{I_1 \cdots I_r} (t,p) \hat H_{I_1} \cdots \hat H_{I_r} X,
 \\
\cJ_\text{DHS}^{(1,0)}(t, \cdot \, ; \p_p \xih, \etah) & = & \om_J(t) \p_p \xih ^J 
+ \sum_{r=1}^\infty \p_t \Phi ^{I_1 \cdots I_r}{}_J (t) \hat H_{I_1 } \cdots \hat H_{I_r} \p_p \xih ^J.
\no
\eea
Combining the two expansions, we may rearrange  the result in terms of a single expansion,
\begin{align}
\label{A.UJ}
\cU_\text{DHS} & (t,p; \xih, \etah)^{-1}  \cJ_\text{DHS}^{(1,0)} (t, \cdot \, ; \p_p \xih, \etah) \,  
\cU_\text{DHS} (t, p; \xih, \etah) 
\no \\ & =
\om_J(t) \p_p \xih ^J + \sum_{r=1}^\infty \cS^{I_1 \cdots I_r} {}_J (t,p) \hat H_{I_1} \cdots \hat H_{I_r} \p_p \xih ^J,
\end{align}
where the coefficients $\cS$ are functions of $\cT$, $\p_t \Phi$ and $\om$. Substituting the expansion (\ref{A.UJ})  into (\ref{A.intA}) and carrying out the $\mA^K$ integral of the first term in the expansion  gives
\bea
\label{A.rec}
\p_p \xih^K + \sum_{r=1}^\infty \PPer^{K I_1 \cdots I_r}{}_J(p)  \hat H_{I_1} \cdots \hat H_{I_r} \p_p \xih ^J =0,
\eea
where $\PPer$ has been defined by,
\bea
\PPer^{K I_1 \cdots I_r}{}_J(p)= \int _p ^{\mA^K \cdot \, p} \cS^{I_1 \cdots I_r} {}_J (t,p).
\eea
Finally, we substitute the derivative $\p_p \xih^J$ obtained from differentiating the Lie series (\ref{A.Lie}) into (\ref{A.rec}), and identify terms of degree $r\geq 1$ in $\etah$,
\bea
 \p_p \XX^{K I_1 \cdots I_r} (p) \etah _{I_1} \cdots \etah_{I_r} 
=-  \sum_{{ s,t \geq 1 \atop s+t=r}}^\infty \PPer^{K I_1 \cdots I_s} {}_M (p) \p_p \XX^{M J_1 \cdots J_t} (p) \hat H_{I_1} \cdots \hat H_{I_s}   \etah _{J_1} \cdots \etah_{J_t}.
\qquad
 \eea
The key observation is that the rank of the tensor $\p_p \XX^{M J_1 \cdots J_t} (p)$ under the sum on the right side is strictly smaller than the rank of the tensor $\p_p \XX^{K I_1 \cdots I_r} (p)$ on the left side since $t<r$.  Since we have already shown that $\XX^{MJ}=\pi Y^{MJ}$ is independent of $p$, it follows that $\p_p \XX^{K I_1 \cdots I_r} (p)=0$ for all $r \geq 1$ by induction on $r$. From implementing the corresponding argument on the derivative $\p _{\bar p} \xih$ we conclude analogously that $\p_{\bar p} \XX^{K I_1 \cdots I_r} (p)=0$ for all $r \geq 1$, so that $\XX^{K I_1 \cdots I_r} (p)$ is independent of $p$ for all $r \geq 1$.  This result can be readily confirmed at rank three by differentiating 
the expressions (\ref{A.AB}) or (\ref{appbb.27}) for $\XX^{LIJ}$ below
with respect to $p$ and $\bar p$.
\end{proof}

\subsection{Proving the shuffle property and cyclic invariance}
\label{sec:shprf}

The shuffle product is an associative and  commutative binary operation on words for which  the empty set $\emptyset$ is the neutral element (see footnote \ref{ft:shuf}).  For multi-index words $J_1 \cdots J_r$ and $ K_1 \cdots K_s$ of length $r,s\geq 1$,  the shuffle product is defined recursively by
 \bea
J_1 \cdots J_r \shuffle K_1 \cdots K_s =  J_1 (J_2 \cdots J_r \shuffle K_1 \cdots K_s  ) 
+ K_1 (J_1 \cdots J_r \shuffle K_2 \cdots K_s ) ,
 \label{appbb.02}
 \eea
along with the neutrality of the empty set
 \bea
J_1 \cdots J_r  \shuffle \emptyset = J_1 \cdots J_r   , \ \ \ \ \ \ r \geq 0 .
 \label{appbb.in}
 \eea
On functions of multi-index words, such as $\XX^{J_1 \cdots J_r}$, the shuffle product acts linearly
\bea
\XX^{J_1 \cdots J_r \, \shuffle \, K_1 \cdots K_s} & = & 
\XX^{J_1 (J_2 \cdots J_r \, \shuffle \, K_1 \cdots K_s  ) } + \XX^{K_1 (J_1 \cdots J_r \, \shuffle \, K_2 \cdots K_s )} ,
\no \\
\XX^{J_1 \cdots J_r \, \shuffle \, \emptyset} & = & \XX^{J_1 \cdots J_r} .
\eea
\begin{proof}
The proof of item 2.\ of Lemma \ref{3.lem:4}, namely the relation (\ref{shprop}) claiming the vanishing
of all shuffles $\XX^{I (J_1 \cdots J_r \shuffle  K_1 \cdots K_s)} $ with $ r,s \geq 1$, proceeds as follows. 
The monodromy condition $\cU_{\rm DHS} (\mA^K \cdot p, p; \xih, \etah){=1}$ is a group-like element for $\xih,\etah \in \mg^h_b$ so that each order $\XX^{I J_1 \cdots J_r}\etah_{J_1}\cdots \etah_{J_r} $ in the expansion of $\xih$ in (\ref{A.Lie})  must be a Lie polynomial in $\etah_J$. By Ree's theorem \cite{Ree:shuffle}, this implies the shuffle properties in~(\ref{shprop}).

\sm

The proof of item 3.\ of Lemma \ref{3.lem:4}, namely cyclic permutation relation of (\ref{appbb.06}), proceeds by exploiting  the vanishing of the commutator $[ \etah_I ,\xih^I]$ established in Lemma \ref{3.lem:2}.
Given that all $r^{\rm th}$-order contributions to $\xih^I$ in $\etah_J$ are bound to separately result
in vanishing commutators, we have
\bea
0 = 
\XX^{I J_1 J_2 \cdots J_r} [ \etah_I , \etah_{J_1} \etah_{J_2}\cdots \etah_{J_r} ]
= 
 \etah_I  \etah_{J_1} \etah_{J_2}\cdots \etah_{J_r}
 ( \XX^{I J_1 J_2 \cdots J_r} - \XX^{ J_1 J_2 \cdots J_r I }).
 \label{appbb.05}
\eea
Linear independence of different words in $ \etah_{K}$ then completes the proof.
\end{proof}

\subsection{Several remarks}

We conclude this appendix with several remarks upon Lemma \ref{3.lem:4}.
\begin{itemize}
\item By combining the symmetry properties (\ref{shprop}) and (\ref{appbb.06}),
the number of independent permutations of $\XX^{I_1 I_2 \cdots I_s}$ in the $s$
indices is $(s-2)!$: the cyclic symmetry (\ref{appbb.06}) can be used to move
any given index $I_k$ (with fixed $k=1,2,\cdots ,s$) to the first entry, and the vanishing shuffles
(\ref{shprop}) imply that only $(r-1)!$ out of the $r!$ permutations
of $\XX^{I J_1 J_2 \cdots J_r} $ in $J_1,\cdots ,J_r$ are independent.
\item As  a result of item 2.\ of Lemma \ref{3.lem:4}, we obtain an explicit formula for the Lie series associated with the coefficients $\XX$. Upon inserting  the $r^{\rm th}$-order contribution to $\xih^I$ in $\etah_J$ into (\ref{A.Lie}), the series may be recast in a Lie series form
\bea
\XX^{I J_1 J_2 \cdots J_r}\etah_{J_1} \etah_{J_2}\cdots \etah_{J_r} 
= \frac{1}{r} \XX^{I J_1 J_2 \cdots J_r} [ \etah_{J_1}, [ \etah_{J_2}, \cdots [ \etah_{J_{r-1}}, \etah_{J_r} ] \cdots ]].
 \label{appbb.04}
\eea
\item Finally, at low ranks, the shuffle and cyclicity properties, established in Lemma \ref{3.lem:4}, imply the following relations. For rank three, the shuffle relation $\XX^{I (J\shuffle K)}=0$ implies anti-symmetry in the last two indices $\XX^{I JK } = - \XX^{I KJ}$ which, combined with the cyclic property, makes $\XX^{IJK}$ totally anti-symmetric,
\bea
\XX^{IJK} = \XX^{ [ IJK ] } .
 \label{appbb.07}
\eea
For rank four,  $\XX^{I_i I_j I_k I_\ell}$ with any permutation $i,j,k,l$ of $1,2,3,4$ can be expressed in a basis generated by the functions $\XX^{I_1 I_2 I_3 I_4}$ and $\XX^{I_1 I_2 I_4 I_3}$ with coefficients $0$ or $\pm1 $. For example, the shuffle properties imply 
\bea
\XX^{I J_1 J_2 J_3} -   \XX^{I J_3 J_2 J_1} & = & 0,
\no \\
\XX^{I J_1 J_2 J_3} + {\rm cycl}(J_1,J_2,J_3) & = & 0 .
\eea
\end{itemize}

\newpage

\section{Proof of Corollary \ref{cor:mshuf}}
\label{sec:cor}

In this appendix, we shall prove that the coefficients $\cM_\shuffle^{I_1 \cdots I_r}{}_J$ in the expansion (\ref{3.bsh})
of $\etah_J$ in $b_{I}$ are explicitly given in terms of the coefficients $\cM^{I_1 \cdots I_r}{}_J$ of $\hat a^I-\xih^I$ by (\ref{mshuf}) and obey shuffle relations (\ref{mshufrel}).

\begin{proof} 
Starting with an Ansatz 
\bea
\label{3.other}
\hat H_J = B_J - \sum_{r=2}^\infty \cM_\shuffle^{I_1 \cdots I_r}{}_J(p) B_{I_1} \cdots B_{I_r}
\eea
for the derivation that implements the adjoint action of $\etah_J$,  the $r^{\rm th}$ order in the expansion of (\ref{3.HB}) in $B_K$ implies the recursion relation
\bea
\cM_\shuffle^{I_1 \cdots I_r}{}_J = \cM^{I_1 \cdots I_r}{}_J
- \sum_{\ell=2}^{r-1}
\cM_\shuffle^{I_1 \cdots I_\ell}{}_K
\cM^{ K I_{\ell+1} \cdots I_r}{}_J
\label{mind.1}
\eea
among the coefficients. We shall now prove by induction that this recursion is solved by (\ref{mshuf}). For $r=2$, one readily shows that $\cM_\shuffle^{I_1 I_2}{}_J= \cM^{I_1 I_2}{}_J$ obeys both (\ref{mind.1}) and (\ref{mshuf}).
Assuming that the claim (\ref{mshuf}) holds for $r=2,3,\cdots,n$, and substituting these relations  into the case of (\ref{mind.1}) at $r=n{+}1$ yields
\begin{align}
\cM_\shuffle^{I_1 \cdots I_{n+1}}{}_J &= \cM^{I_1 \cdots I_{n+1}}{}_J
- \sum_{\ell=2}^{n}
\sum_{s=0}^{\ell-2} (-1)^s \sum_{2\leq j_1< \cdots < j_s}^{\ell-1}
\cM^{I_1 I_2 \cdots I_{j_1}}{}_{K_1}  \cM^{K_1 I_{j_1+1}  \cdots I_{j_2}  }{}_{K_2}   \times  \cdots \notag \\
&\quad\quad\quad\quad\quad\quad\quad
\times \cdots
\cM^{K_{s-1} I_{j_{s-1}+1}  \cdots I_{j_s}  }{}_{K_s} 
\cM^{K_s I_{j_s+1}  \cdots I_{\ell}  }{}_{K} 
\cM^{ K I_{\ell+1} \cdots I_{n+1}}{}_J.
\label{mind.2}
\end{align}
This needs to be lined up with the $r=n{+}1$ instance of the claim (\ref{mshuf}),
\begin{align}
\cM_\shuffle^{I_1 \cdots I_{n+1}}{}_J  &= \cM^{I_1 \cdots I_{n+1}}{}_J + \sum_{u=1}^{n-1} (-1)^u \sum_{2\leq j_1 < j_2 < \cdots < j_u}^{n}
\cM^{I_1 I_2 \cdots I_{j_1}}{}_{K_1}  \cM^{K_1 I_{j_1+1}  \cdots I_{j_2}  }{}_{K_2}   \times  \cdots \notag \\
&\quad\quad\quad\quad\quad\quad\quad
\times \cdots
\cM^{K_{u-1} I_{j_{u-1}+1}  \cdots I_{j_u}  }{}_{K_u} 
\cM^{K_u I_{j_u+1}  \cdots I_{n+1}  }{}_{J} .
\label{mind.3}
\end{align}
We have renamed the summation variable of (\ref{mshuf}) to $\ell \rightarrow u$ and
note that the first term $ \cM^{I_1 \cdots I_{n+1}}{}_J $ on the right side readily matches that
of (\ref{mind.2}).
The remaining contributions in the nested sums of (\ref{mind.2}) can be shown to
match those of (\ref{mind.3}) by interchanging the sums over $\ell$ and $s$ and  then renaming $s=u{-}1$,
\begin{align}
\cM_\shuffle^{I_1 \cdots I_{n+1}}{}_J &= \cM^{I_1 \cdots I_{n+1}}{}_J
+ \sum_{u=1}^{n-1} (-1)^u
\sum_{\ell=u+1}^{n}  \sum_{2\leq j_1< \cdots < j_{u-1}}^{\ell-1}
\cM^{I_1 I_2 \cdots I_{j_1}}{}_{K_1}  \cM^{K_1 I_{j_1+1}  \cdots I_{j_2}  }{}_{K_2} \times  \cdots 
\notag \\
&\quad\quad \times \cdots
\cM^{K_{u-2} I_{j_{u-2}+1}  \cdots I_{j_{u-1}}  }{}_{K_{u-1}} 
\cM^{K_{u-1} I_{j_{u-1}+1}  \cdots I_{\ell}  }{}_{K_u} 
\cM^{ K_u I_{\ell+1} \cdots I_{n+1}}{}_J , 
\label{mind.4}
\end{align}
where the last summation index has been renamed from $K$ to $K_u$.
As a last step, we rename $\ell$ in (\ref{mind.4}) to $j_u$ and rearrange the nested
sum $\sum_{\ell=u+1}^n \sum_{2\leq j_1< \cdots < j_{u-1}}^{\ell-1}$ into 
$ \sum_{2\leq j_1< \cdots < j_{u}}^n$ (the lower bound $\ell\geq u{+}1$ is
consistent with $j_i \geq i{+}1)$. This recovers the expression (\ref{mind.3}) for the 
claim (\ref{mshuf}) at $r=n{+}1$ and completes both the inductive step and the inductive proof
of (\ref{mshuf}).

\sm

Finally, the shuffle relations (\ref{mshufrel}) are necessary by
Ree's theorem to ensure that the expression  (\ref{3.bsh}) for $\etah_J$ is in $\mg$ as
exposed by (\ref{exphatb}). \end{proof}

\newpage

\section{Explicit evaluation to low orders}
\setcounter{equation}{0}
\label{sec:A}

The purpose of this appendix is to illustrate the relation between the connections $d - \cK_\text{E}$ and $d - \cJ_\text{DHS}$, stated in (\ref{gaugetransf1}) of the introduction and proven in Theorem \ref{3.thm:1} of section~\ref{sec:3}, by providing a detailed derivation of explicit formulas for the gauge transformation $\cU_\text{DHS}$ and the automorphism $a \, \cup \, b \to \hat a \, \cup \, \hat b$ to lowest and next-to-lowest orders in the generators~$a$ and~$b$.  Along the way, we shall produce a number of useful properties of the connections and their interrelation.  We begin by summarizing the approach suitably readied for practical calculations.  
 
\sm

Throughout this appendix, we shall denote the solutions $\hat \xi$ and $\hat \eta$ of the monodromy equations (\ref{3.mon1}) simply by $\xi$ and $\eta$ in order to avoid cluttering.

\subsection{Practical summary of the approach}

A convenient starting point for the construction of the gauge transformation is obtained by combining the results of Lemmas \ref{3.lem:1}, \ref{3.lem:2} and Corollaries \ref{3.cor:1}, \ref{3.coroll2:BIS} which guarantee that  $\cU_\text{DHS}(x,p)$ is given by the path-ordered exponential, 
\bea
\label{A.U}
\cU_\text{DHS} (x,p;\xi, \eta) =  \text{P} \exp \int ^x _p  \cJ_\text{DHS} (\ti, \cdot\, ; \xi, \eta).
\eea
Here $\xi$ and $\eta$ obey $[\eta_I, \xi^I]=0$ so that the connection $\cJ_\text{DHS} (t, p; \xi, \eta) $ is independent of $p$, smooth in $x$,  and given by the simplified formula, 
\bea
\label{A.conn}
\cJ_\text{DHS} (t, \cdot \, ; \xi, \eta) = \om_J(t) \xi^J - \pi \bar \om ^I(t)\eta_I + \sum _{r=1}^\infty
 \p_t \Phi ^{I_1 \cdots I_r}{}_J(t) H_{I_1} \cdots H_{I_r} \xi^J,
\eea
where  $H_I X = [\eta_I, X]$ for any $X \in \mg$, the smooth functions $\Phi^{I_1 \cdots I_r}{}_J(t)$ 
are defined in (\ref{eq:251114n1}), (\ref{eq:251114n2}), and $\xi$ and $\eta$ are the unique solutions to  the system of monodromy conditions,
\bea
\label{A.mon}
\cU_\text{DHS} (\mA^K \cdot p , p; \xi, \eta) & = & 1,
\no \\
\cU_\text{DHS} (\mB_K \cdot p , p; \xi, \eta) & = & e^{2 \pi i b_K}.
\eea
Solving the above system of equations gives $\xi,  \eta$ and $\cU_\text{DHS}$  in terms of $b$. 

\subsection{Expansions in words of the alphabet $\xi \cup \eta$}

To solve the monodromy relations (\ref{A.mon}) for $\xi$ and $\eta$ as a function of $b$ (and the moduli of the surface $\Sigma _p$), it will be convenient to obtain first the expansion of $\cU_\text{DHS}(x,p;\xi,\eta)$ in a power series in $\xi$ and $\eta$, subject only to the constraint $[\eta_I, \xi^I]=0$. 
This expansion may be organized in terms of the integer $r$ which is the sum of the number of letters $\xi$ and the number of letters $\eta$ (or equivalently the total word length),   
\bea
\label{A.exp1}
\cU_\text{DHS}(x,p;\xi, \eta) = 1 + \sum_{r=1}^\infty \cU_r(x,p;\xi,\eta).
\eea
We refrain from imposing any relation between $\xi$ and $\eta$ due to the 
monodromy relations at this stage which in general do not preserve the total word length.
While $\cU_\text{DHS}(x,p;\xi, \eta)$ takes values in $\exp(\mg_b)$, each term $\cU_r(x,p;\xi,\eta)$ takes values in $\CC\< \! \< b \> \! \>$.  The expansion of (\ref{A.exp1}) has the advantage of leading to simple shuffle relations obeyed by the  polylogarithms in their coefficients. It will be convenient to also consider the complementary expansion of the Lie series $\muu_\text{DHS}(x,p;\xi, \eta) \in \mg_b$ defined by
\bea
\label{A.exp}
\cU_\text{DHS} (x,p;\xi, \eta) = \exp \big \{ \muu_\text{DHS}(x,p;\xi, \eta) \big \}
\eea 
 in a power series in $\xi$ and $\eta$, whose terms $\muu_r(x,p;\xi, \eta)$ have word length $r$,  
 \bea
\muu_\text{DHS}(x,p;\xi, \eta) =  \sum_{r=1}^\infty \muu_r(x,p;\xi, \eta). 
\eea
Being directly in terms of Lie algebra elements, this expansion\footnote{The expansion is often referred to as the \textit{Magnus expansion} of the path-ordered exponential, see for example  \cite{Blanes:2008xlr} and references therein.} provides a convenient way to relate the connections of (\ref{3.KJ}) which are valued in $\mg$, but it has the disadvantage of obscuring the role of the shuffle relations and lacking a canonical presentation due to the Jacobi identity. The terms in the two expansions may, of course, be 
simply related to one another using (\ref{A.exp}), and the lowest three orders are given by 
\bea
\label{A.logs}
\muu_1 (x,p;\xi, \eta) & = & \cU_1(x,p;\xi, \eta),
\no \\
\muu_2 (x,p;\xi, \eta) & = & \cU_2(x,p;\xi, \eta)  -\thalf \cU_1(x,p;\xi, \eta) ^2,
\no \\
\muu_3 (x,p;\xi, \eta) & = & \cU_3(x,p;\xi, \eta) - \thalf \cU_1(x,p;\xi, \eta) \cU_2(x,p;\xi, \eta)
\no \\ &&
-\thalf \cU_2(x,p;\xi, \eta) \cU_1(x,p;\xi, \eta) +\tfrac{1}{3} \cU_1(x,p;\xi, \eta)^3.
\eea
These relations allow us to pursue the two expansions in parallel.  Since $\cJ_\text{DHS}(x, \cdot\, ;\xi,\eta)$ is a flat connection, the gauge transformation $\cU_\text{DHS}(x,p;\xi,\eta)$, and its expansion components $\cU_r(x,p;\xi, \eta)$ and $\muu_r(x,p;\xi, \eta)$ are all  homotopy invariant. 

\sm

In the remainder of this appendix, we shall evaluate $\cU_1$ and $\cU_2$, or equivalently $\muu_1$ and~$\muu_2$, enforce the monodromy relations (\ref{A.mon}) to second order in $b$, and then use the result to obtain $\xi, \eta$ and $\cU_\text{DHS}$ to second order in $b$.

\subsection{Calculating $\cU_1, \cU_2$ and $\muu_1, \muu_2$}

The contributions $\cU_1, \cU_2$ are obtained by substituting the connection (\ref{A.conn}) into the path-ordered exponential of (\ref{A.U}) and using the expansion (\ref{A.exp1}), 
\bea
\cU_1(x,p;\xi, \eta) & = &  \int_p^x  \big ( \om_J \xi^J - \pi \bar \om ^I\eta_I \big ),
\\
\cU_2(x,p;\xi,\eta) & = & 
\int _p^x \! \big ( \om_I(\ti ) \xi^I - \pi \bar \om _I(\ti )\eta ^I \big ) 
\int _p^{\ti} \! \big ( \om_J \xi^J - \pi \bar \om _J \eta ^J \big )
+ [\eta_I, \xi^J] \int_p^x  \! \p_\ti  \Phi^I{}_J (\ti).
\no
\eea
Using the first two lines of (\ref{A.logs}), we obtain the second order contributions to $\muu_1, \muu_2$, 
\bea
\label{A.muu}
\muu_1(x,p;\xi, \eta) & = &  \int_p^x  \big ( \om_I \xi^I - \pi \bar \om ^J\eta_J \big ),
\no \\
\muu_2(x,p;\xi,\eta) & = & 
 [\eta_I, \xi^J] \int _p ^x \left (  \p_\ti \Phi ^I{}_J (\ti) - {\pi \over 2}    \bar \om^I (t) \int^t _p  \om_J 
+ {\pi \over 2}   \om_J (t) \int^t _p  \, \bar \om^I \right )
\no \\ &&
+\half [\xi^I, \xi^J] \int ^x _p  \om_I (t) \int^t _p  \om_J
+ {\pi^2 \over 2}  [\eta_I, \eta_J] \int ^x _p  \bar \om^I (t) \int^t _p  \bar \om^J,
\eea
where we have used the familiar rearrangement formula underlying shuffle relations
\bea
\int ^x _p d\ti_1 \int ^x_p d\ti_2 = \int _p ^x d\ti_1 \int ^{\ti_1} _p d\ti_2 + \int _p ^x d\ti_2 \int ^{\ti_2}  _p d\ti_1
\eea
to present $\muu_2$ in a form that manifestly belongs to $\mg$. The contributions $\muu_1$ and the second line of $\muu_2$ are manifestly homotopy invariant for arbitrary $\xi, \eta \in \mg$. Using the relation,
\bea
\label{A.delphi}
\bar \p_t \p_t \Phi ^I{}_J (t) = \pi \kappa (t) \, \delta ^I_J - \pi \bar \om^I (t) \wedge \om_J(t),
\eea
where $\kappa(x)$ was defined in (\ref{2.kappa}), we see that the integrand of the first line of $\muu_2$ is a closed 1-form in view of the fact that the term in $\kappa$ cancels thanks to the relation $[\eta_I, \xi^I]=0$. Thus,  $\muu_2$ is  homotopy invariant to this order in $\xi, \eta$, as expected on general grounds.

\subsection{Solving the monodromy relations to order $b^2$}

The monodromies of $\muu_1$ are $p$-independent,  and we recover the result of (\ref{3.xietaA}), 
\bea
\label{A.muu1}
\muu_1(\mA^K \cdot p ,p;\xi,\eta) & = & \xi^K - \pi \eta^K,
\no \\
\muu_1(\mB_K \cdot p ,p;\xi,\eta) & = & \Omega_{KI} \xi^I - \pi \bar \Omega _{KI} \eta ^I.
\eea
Setting these monodromies equal to 0 and $2\pi i b_K$, respectively, we recover the result of (\ref{3.xieta}) to first order in $b$, namely $  \eta _I =   b_I + \cO(b^2)$ and $  \xi_I = \pi b_I + \cO(b^2)$. 

\sm

The monodromies of $\muu_2$ may be simplified by using the fact that we need their evaluation only  to second order in $\xi$ and $\eta$ so that we are allowed to substitute the first order relation $\xi_I = \pi \eta_I$ into $\muu_2$  in (\ref{A.muu}).  The result is as follows, 
\bea
\label{A.u2}
\muu_2(x,p;\pi \eta,\eta) & = &   \pi \, \eta_I \eta_J \left (   \int _p ^x  \, 2\, \p_\ti \Phi ^{[IJ]} (\ti ) 
- 4 \pi \, \Im \int^x _p  \om^{ [I}(\ti ) \, \Im \int^\ti _p  \om^{J  ]} \right ),
\eea
where the square brackets stand for anti-symmetrization of the indices $I,J$, such as for example $\eta_{[I} \eta _{J]} = \half ( \eta _I \eta _J - \eta _J \eta _I)$. We have used this property  to recast the commutators in (\ref{A.muu}) in terms of a product $\eta_I \eta_J$ in (\ref{A.u2}).   
The integral of $\p_\ti \Phi ^{[IJ]} $ may be simplified by using the following rearrangement, 
\bea
2 \, \p_\ti \Phi ^{[IJ]} & = & d_\ti \Phi ^{[IJ]} + \p_t \Phi^{[IJ]} + \overline{\p_t \Phi ^{[IJ]}},
\eea
which is obtained with the help of the  complex-conjugation property $\overline{\Phi^{[IJ]}(t)} = - \Phi^{[IJ]}(t)$. Integrating the total differential $d_\ti \Phi^{[IJ]} $ and recasting the remainder as a sum of an integral and its complex conjugate, we obtain the following  formula for $\muu_2$ to this order,
\bea
\muu_2(x,p;\pi \eta ,\eta) & = &  
\pi  \, \eta_I \eta_J \left (  \Phi ^{[IJ]} (x) -  \Phi ^{[IJ]} (p) 
+   \int _p ^x  \lambda ^{IJ} (\ti,p)  +  \int _p ^x  \overline{\lambda ^{IJ} (\ti,p) } \right ).
\qquad
\eea
The  $(1,0)$-form $\lambda^{IJ} (t,p)$ is given by,
\bea
\label{A.lambda}
\lambda ^{IJ} (t,p) & = & \p_t \Phi ^{[IJ]} (t) + 2 \pi i \,  \om^{[I}(t) \, \Im \int^t _p  \om^{J]}.
\eea
One verifies that $\lambda^{IJ} (t,p)$ is anti-symmetric in its indices $I,J$ as well as holomorphic in $t$ in view of the relation (\ref{A.delphi}). As a result, the function $\muu_2(x,p;\pi \eta,\eta) $ is  homotopy invariant, as expected. Its monodromies are given by,
\bea
\label{A.muu2}
\muu_2(\mA^K \cdot p , p; \pi \eta , \eta) & = &  - \XX^{KIJ} (p) \, \eta_I \eta_J,
\no \\
\muu_2(\mB_K \cdot p , p; \pi \eta , \eta) & = &  - \YY_{K}{}^{IJ} (p) \, \eta _I \eta _J,
\eea
where $\XX$ and $\YY$ are given by
\bea
\label{A.AB}
\XX^{KIJ} (p) & = &  - \pi  \int _p ^{\mA^K \cdot \, p } \lambda ^{IJ} (\ti ,p) + \hbox{ c.c.}
\no \\
\YY_{K}{}^{IJ} (p) & = &  - \pi  \int _p ^{\mB_K \cdot \, p}  \lambda ^{IJ} (\ti ,p) + \hbox{ c.c.}
\eea
By construction, $\XX^K{}_{IJ}$ and $\YY_{KIJ}$ are anti-symmetric in $I,J$  and real multiple-valued harmonic functions of $p$. Using the properties of $\lambda ^{IJ}$ one readily shows that
\bea
\p _p \, \XX^{KIJ} (p)= \p_{\bar p} \, \XX^{KIJ} (p) =0,
\eea
i.e.\ $\XX^{KIJ}$ is independent of $p$ as expected on general grounds
by the proof in appendix~\ref{sec:nopprf}. 

\sm

Combining the conditions of (\ref{A.mon}) with (\ref{A.exp}) and the monodromies of $\muu_1$ in (\ref{A.muu1}) and $\muu_2$ in (\ref{A.muu2}),  we obtain the full monodromy relations to second order,
\bea
\label{A.xieta}
\xi^K - \pi \eta ^K -  \XX^{KIJ} \, \eta _I \eta _J + \cO(\eta^3)  & = & 0, 
\no \\
\Omega_{KJ} \xi^J - \pi \bar \Omega _{KJ}\eta^J -  \YY_{K}{}^{IJ} (p) \, \eta_I \eta _J + \cO(\eta^3)  & = &
2\pi i b_K.
\eea
Solving the first equation of (\ref{A.xieta}) for $\xi$ in terms of $\eta$, substituting the result into the second equation of (\ref{A.xieta}), and then solving for $\eta$ and $\xi$ up to second order in $b$ gives the final result which may be summarized in terms of the following lemma.

{\lem
\label{A.lem:1}
The solution to the monodromy equations  (\ref{A.mon}) is given as follows,
\bea
\label{7.eta-xi}
 \xi ^K & = & \pi b^K  +  \Big ( \XX^{K IJ} - \pi \cM^{IJK} (p) \Big ) \, b_I b_J + \cO(b^3),
\no \\
\eta _K & = &  b_K -  \cM^{IJ}{}_K (p) \, b_I b_J + \cO(b^3),
\eea
where $\XX^{KIJ} $ is independent of $p$ and $\cM^{IJ}{}_K (p)$ is given by 
\bea
\label{A.M}
\cM^{IJ}{}_K (p)=  { i \over 2 \pi } \Big (  \YY_{K}{}^{IJ}(p) - \Omega _{KL} \XX^{LIJ}  \Big ) 
\eea
in terms of $\XX$ and $\YY$ given in (\ref{A.AB}). }

\begin{rmk}
An alternative way to express the first line in (\ref{7.eta-xi}) is as follows,
\bea
\xi_K = \pi b_K - { i \over 2} \Big ( \YY_K{}^{IJ} (p) - \bar \Omega _{KL} \XX^{LIJ} \Big ) b _I b _J + \cO(b^3).
\eea
\end{rmk}

\subsection{Matching residues and relating the connections}

We are now ready to state and prove the main proposition of this appendix. 

{\prop 
\label{A.prop:1}
The relation between the connections $d - \cK_\text{E}$ and $d - \cJ_\text{DHS}$ which, to the lowest and next-to-lowest orders in $b$ are given by
\bea
\label{A.KJ}
\cK_\text{E}( x,p;a,b) & = & \om_J(x) a^J +  g^I{}_J(x,p) [b_I , a^J] + \cO(a b^2),
\no \\
\cJ_\text{DHS}^{(1,0)}  (x,p; \hat a, \hat b) & = & 
\om_J(x) \hat a^J  + f^I{}_J(x,p) [\hat b_I, \hat a^J] + \cO(\hat a \hat b^2),
\eea
is determined by the automorphism $a \cup b \to \hat a \cup \hat b$ of $\mg$,
\bea
\label{A.hats}
\hat a ^K & = & a^K + \xi^K  + \cM^{KI}{}_J (p) [b_I, a^J]  + \cO(ab^2, b^3),
\no \\
\hat b_K & = & b_K - \cM^{IJ} {}_K(p) b_I b_J  + \cO(b^3),
\eea
and  the following relation between the forms $f^{I}{}_J $ and $g^I{}_J$, 
 \bea
\label{7.omf}
g^I{}_J(x,p)  =   f^I{}_J(x,p) - 2 \pi i \,  \om_J(x) \,  \Im \int_p^x   \om^I + \om_K(x) \cM^{KI}{}_J(p).
\eea 
The coefficients $\XX$ and $\cM$ are given by the first equation in (\ref{A.AB}) as well as (\ref{A.M}) and~$\xi$ is given in terms of $b$ to this order in (\ref{7.eta-xi}). }

\smallskip	

{\proof 
To prove the proposition, we begin by using the first equation in (\ref{3.hats}) to establish $\hat b_I = \eta_I$, which by the second line in (\ref{7.eta-xi}) readily proves the second equation in (\ref{A.hats}). To prove the first equation in (\ref{A.hats}), we substitute the expression for $\eta$, obtained in the second line of (\ref{7.eta-xi})  to second order in $b$, into the residue matching condition (\ref{3.res-m}),
\bea
{} [b_K, a^K] & = & \left [ b_K - \cM^{IJ}{} _{K}(p) \, b_I b_J , \ta^K - \xi ^K \right ] + \cO(b^3).
\eea
Since the Lie algebra $\mg$ is freely generated, the unique solution to first order in $a$ and zeroth order in $b$ is given by  $\ta^K=a^K + \cO(b)$. To first order in $b$ we obtain uniquely, 
\bea
\label{A.aD}
\ta ^K - \xi^K  = a^K + \cM^{KI}{}_J  (p) [b_I, a^J] + \cO(b^2),
\eea
which establishes the first line in (\ref{A.hats}). Finally, using  the relation (\ref{3.JKa}) to first order in $b$ by  expanding $\cU_\text{DHS}$ accordingly, we obtain the condition 
\bea
\cK_\text{E}( x,p;a,b) & = & \cJ_\text{DHS}^{(1,0)}  (x,p; \ta - \xi, \eta)
- \Big [ \muu_1 (x,p;\xi,\eta), \cJ_\text{DHS}^{(1,0)}  (x,p; \ta - \xi, \eta) \Big ] + \cO(b^2),\qquad
\eea
where $\cK_\text{E}( x,p;a,b) $ is given in the first line of (\ref{A.KJ}), $\muu_1(x,p;\xi,\eta)$ can be found in the first line of (\ref{A.muu}),  and $\cJ_\text{DHS}^{(1,0)}  (x,p; \ta - \xi, \eta) $ is given by, 
\bea
\cJ_\text{DHS}^{(1,0)}  (x,p; \ta - \xi, \eta) & = & 
\om_K(x) (\ta ^K - \xi^K) + f^I{}_J(x,p) [\eta_I, \ta^J - \xi ^J] + \cO(b^2).
\eea
Using the relation $\hat b _I=\eta_I$ and the solution of (\ref{A.aD}) for $\hat a^I$, we readily obtain the relation (\ref{7.omf}), thereby completing the proof of Proposition \ref{A.prop:1}. }


\subsection{Evaluation of $\XX^{KIJ}$ and $\cM^{IJ}{}_K$}
\label{sec:7.AL}

To evaluate the coefficients $\XX^{KIJ}$ and $\YY_K{}^{IJ}$ and thus $\cM^{IJ}{}_K$ in Lemma \ref{A.lem:1} we begin by simplifying the integral of the differential form $\lambda_{IJ}$ defined in (\ref{A.lambda}). To do so, we recall the expression for $\p_t \Phi^I {}_J(t)$, raise the index $J$ and anti-symmetrize in the indices $I,J$, 
\bea
\label{A.Phi}
\p_t \Phi ^{[IJ]} (t) = - {i \over 4} \int_\Sigma  \p_t \cG(t,\tp) 
\Big ( \bar \om^I(\tp) \wedge \om^J(\tp) - \bar \om^J(\tp) \wedge \om^I(\tp) \Big ).
\eea
The derivative of $\cG(t,\tp)$ may be expressed in terms of the prime form $E(x,y)$,
\bea
\p_t \cG(t,\tp) = - \p_t \ln E(t,\tp) - \p_t \gamma (t) - 2 \pi i \, \om_K(t) \,\Im \int ^t _{\tp} \om^K,
\eea
where the expression for $\gamma(t)$ may be found in \cite{DHoker:2017pvk}, but will not be needed here as its contribution is $y$-independent and integrates to zero against   $\bar \om^I(\tp) \wedge \om^J(\tp) - \bar \om^J(\tp) \wedge \om^I(\tp) $. To proceed, we denote the imaginary part of the Abelian integral on $\Sigma _p$ by
\bea
\phi^I(x) = \Im \int ^{x} _p \om^I.
\eea
One  verifies that $\phi^I$ has vanishing $\mA$-monodromies while its $\mB$-monodromies are given by $
\phi^I ( \mB_K \cdot x) = \phi^I (x) + \delta^I_K$. We also have the following relation, 
\bea
\bar \om^I(\tp) \wedge \om^J(\tp) - \bar \om^J(\tp) \wedge \om^I(\tp) = 4 \, d \phi^I (\tp) \wedge d \phi^J(\tp).
\eea
Expressing (\ref{A.Phi}) in terms of $\phi_I$ we have
\bea
\p_t \Phi ^{[IJ]} (t) & = &   \int_\Sigma  
\Big (   i \p_t \ln { E(t,\tp) \over E(t,p)} + 2 \pi  \, \om_K(t)  \phi ^K(\tp)   \Big ) d \phi^I(\tp) \wedge d \phi^J(\tp),
\eea
where we have used the fact that the contributions inside the parentheses proportional to $\phi^K(t)$ and $\p_t \ln E(t,p)$  integrate to zero against $d \phi^I \wedge d \phi^J$. The term in $\p_t \ln E(t,p)$ has been included to render the integrand monodromy free in $t$, while its monodromy in $\tp$ integrates to zero against $d \phi^I \wedge d \phi^J$. Using the fact that the first term inside the parentheses has vanishing $\mA$-periods in $t$ and that the $\mB$-periods are given by  
\bea
\oint _{ \mB_K} \p_t \ln { E(t,\tp) \over E(t,p)} =  2 \pi i \int ^{\tp}_p \om_K(t),
\eea
we evaluate the periods of $\p_t \Phi_{[IJ]} $ as follows, 
\bea
\oint _{\mA^K} \p_t \Phi^{[IJ]} & =  & 2 \pi \int _\Sigma \phi ^K \, d \phi^I \wedge d \phi^J,
\no \\
\oint _{\mB_K} \p_t \Phi^{[IJ]}  & = & 2 \pi  \int_\Sigma  
\left ( - \int^{\tp} _p  \om_K  + \Omega_{KL}  \phi^L(\tp)   \right ) d \phi^I(\tp) \wedge d \phi^J(\tp).
\eea
As  a result, we find 
\bea
\label{7.AIJK}
\XX^{KIJ} & = &  - 4 \pi^2   \int_\Sigma  \phi ^ K \, d\phi^I \wedge d \phi^J 
+ 2 \pi^2   \oint _{\mA^K} \Big ( \phi^J \, d \phi^I - \phi^I \, d \phi^J \Big ),
\no \\
\YY_{K}{}^{ IJ} (p) & = &  - 4 \pi^2   \int_\Sigma  \left ( - \Re \int^{\tp} _p \om_K 
+ \Re(\Omega) _{KL} \phi^L(\tp) \right ) d\phi^I(\tp) \wedge d \phi^J (\tp),
\no \\ &&
+ 2 \pi^2   \int _p^{\mB_K \cdot \, p}  \Big ( \phi^J \, d \phi^I - \phi^I \, d \phi^J \Big ).
\eea
Both expressions are manifestly real-valued and antisymmetric in $I$ and $J$ as expected on general grounds. The combination $\cM^{IJ}{}_K(p)$ then follows from (\ref{A.M}) and is given by
\bea
\cM^{IJ}{}_K (p)  & = &
-2 \pi i   \int_\Sigma  d\phi^I(\tp) \wedge d \phi^J (\tp) \,   \int^{\tp}_p \om_K
 \\ &&
+  \pi i   \int _p^{\mB_K \cdot \, p}  \Big ( \phi^J \, d \phi^I - \phi^I \, d \phi^J \Big )
-  \pi i \, \Omega_{KL} \oint _{\mA^L} \Big ( \phi^J \, d \phi^I - \phi^I \, d \phi^J \Big ).
\no
\eea
For later use, we record the following simplification
\bea
\label{B.phiphi}
\oint _{\mA^K} \Big ( \phi^J \, d \phi^I - \phi^I \, d \phi^J \Big ) = 
2 \oint _{\mA^K} \phi^J \, d \phi^I 
\eea
in view of the fact that the integrand has vanishing $\mA$-monodromies, so that $d(\phi^I \phi^J)$ has vanishing $\mA$-period. Note, however, that no such simplification applies to the corresponding integral over the $\mB_K$ cycles since $\phi^I$ does have non-vanishing $\mB$-monodromy.  For the same reason the integrals over the $\mB$ cycles depend on the base-point of the cycle, which is why we have left the base point $p$ exposed in the above formulas for $\mB$ periods.

\subsection{Verifying the cyclic property  of $\XX_{IJK}$}

The relation $[\eta _K, \xi^K]=0$ requires the following identity on $\XX$,
\bea
\XX^{KIJ} \big [b_K, [b_I, b_J] \big ]=0.
\eea 
Since the algebra $\mg_b$ is freely generated, the only relation the commutators can satisfy is the Jacobi identity. As a result, the above relation implies the following condition on $\XX$,
\bea
\label{7.cycl}
 \XX^{KIJ} = \XX^{JKI} = \XX^{IJK}
\eea
which is proven on general grounds in appendix \ref{sec:shprf}. Here, we shall verify that this condition holds by explicit calculation. We regroup the difference as follows,
\bea
\label{A.Aconj}
\XX^{JKI} - \XX^{KIJ} & = & 
- 4 \pi^2  \bigg\{ \int_\Sigma d \Big (  \phi ^ J  \phi^K d \phi^I \Big ) 
+    \oint _{\mA^J}   \phi^K \, d \phi^I 
+   \oint _{\mA^K}   \phi^J \, d \phi^I  \bigg\}
 \eea
 by using the simplifications of (\ref{B.phiphi}) to convert $ \phi^I \, d \phi^K  \rightarrow - \phi^K \, d \phi^I$
 and $ \phi^I \, d \phi^J \rightarrow - \phi^J \, d \phi^I $ in the integrand of $\mA^J$ and $\mA^K$, respectively.
 The first term may be recast as a line integral using Stokes's theorem, 
  \bea
 \int_D  d \Big ( \phi^J \phi^K  d \phi^I \Big ) = \oint _{\p D} \phi^J \phi^K  d \phi^I,
 \eea
 which in turn may be evaluated using the canonical decomposition of the boundary $\p D$ of the fundamental domain depicted in figure \ref{fig:2}, 
  \bea
 \oint _{\p D} \phi^J \phi^K  d \phi^I &  = &
 \sum_{L=1}^h  \int _{y} ^{\mA^L \cdot \, y} 
 \Big ( \phi^J (t) \phi^K (t)   - \phi^J ( \mB_L \cdot t ) \phi^K  (\mB_L \cdot t) \Big ) d \phi^I (t)
  \no \\ & = &
-  \sum_{L=1}^h   \int _{y} ^{\mA^L \cdot \, y} 
\Big ( \delta^J_L  \phi^K + \delta^K_L \phi^J + \delta ^J_L \delta ^K_L  \Big ) d \phi^I 
\no \\ & = & 
- \oint _{\mA^J} \phi^K d \phi ^I - \oint _{\mA^K} \phi^J d\phi^I.
\label{dphicancel}
 \eea
The last term in the integrand on the second line  above cancels by itself and the remaining integrals cancel the second and third integrals on the right side of (\ref{A.Aconj}), thus confirming the identity (\ref{7.cycl}).

\newpage

\section{Explicit construction of Enriquez kernels}
\setcounter{equation}{0}
\label{sec:B}

This appendix describes the explicit construction of the
Enriquez kernels $g^{I_1 \cdots I_r}{}_J(x,p)$ in (\ref{2.Kexp}) 
in terms of the $f$-tensors in Lemma \ref{2.lem:1}, Abelian differentials
as well as their iterated integrals.
In particular, we spell out detailed examples and intermediate results in the
procedure of section \ref{sec:3.6} including the
coefficients of the contributing series expansions.

\sm

Throughout this appendix, we shall denote the solutions $\hat \xi$ and $\hat \eta$ of the monodromy equations (\ref{3.mon1}) simply by $\xi$ and $\eta$ in order to avoid cluttering.

\subsection{Expanding the gauge transformation}
\label{sec:B.2}

Before determining the explicit form of the $\XX^{I J_1 \cdots J_r}$ coefficients  in the expansion (\ref{A.Lie}) of $\xi^I$, it is convenient to first compute the expansion of the path-ordered exponential $ {\cal U}_{\rm DHS}(x,p;\xi,\eta )^{-1} $ in (\ref{reinstated}) for generic $x$. By the vanishing of $[\eta_I, \xi^I]$ and Corollary \ref{3.cor:1}, the DHS polylogarithms in this expansion boil down to iterated integrals involving the $\delta^J_{I_r}$-traceless and $p$-independent parts $\partial_x \Phi^{I_{1} \cdots I_r}{}_J(x) $  of the $f^{I_1 \cdots I_r}{}_J(x,p)$-tensors in \eqref{phiastrless}. 
Homotopy invariance at the subleading order in $\eta_J$ relies on the combination of terms in
 \bea
  \label{gammaang}
 \Gamma^{\langle IJ \rangle}(x,p) = \int^x_p  \bigg ( \partial_t  \Phi^{IJ}(t) - \p_t \Phi^{JI}(t)
 + \pi  \omega^J(t_1) \! \int^{t_1}_p \! \!  \bar \omega^I(t_2) - \pi \omega^I(t_1) \! \int^{t_1}_p \! \! \bar \omega^J(t_2) \bigg )
\quad
 \eea
 with $ \Phi^{IJ}(t) =  \Phi^{I}{}_K(t) Y^{KJ}$ and manifest antisymmetry $ \Gamma^{\langle IJ \rangle}(x,p) 
= -  \Gamma^{\langle JI \rangle}(x,p) $. The expression (\ref{gammaang})
enters $ {\cal U}_{\rm DHS}(x,p;\xi,\eta )^{-1} $ in combination with the following iterated Abelian integrals 
 realized as DHS polylogarithms associated with $\mw = a^{I_1} \cdots a^{I_r}$ in (\ref{dhspoly})
 \bea
  \Gamma_{I_1 I_2 \cdots I_r}(x,p) = \int^x_p \om_{I_1}(t_1)\int^{t_1}_p \om_{I_2}(t_2) 
  \cdots \int^{t_{r-1}}_p \om_{I_r}(t_r) .
  \label{appbb.12}
\eea
In slight abuse of notation, we shall write $  \Gamma^{J_1 J_2 \cdots J_r}(x,p) $
for their contraction with $Y^{I_1 J_1}\cdots$ $Y^{I_r J_r}$ (which does not refer to
the DHS polylogarithms associated with $\mw = b_{J_1} \cdots b_{J_r}$) 
and $\overline{ \Gamma^{J_1 J_2 \cdots J_r}(x,p) }$ for the complex conjugates
of these contractions. 
 
\sm

With these prerequisites in place, we can write the coefficients $\TT^{I_1  \cdots I_r}(x,p)$ at the order
$r\leq 2$ of the expansion (\ref{reinstated}) of ${\cal U}_{\rm DHS}(x,p;\xi,\eta)^{-1}$
in the following form:
\begin{align}
\TT^{I}(x,p) &= \pi \big(   \overline{ \Gamma^{I}(x,p) } -   \Gamma^{I}(x,p) \big) ,
 \label{appbb.19} \\
\TT^{IJ}(x,p) &= - \XX_M{}^{IJ}  \Gamma^M(x,p) - \pi \Gamma^{ \langle IJ \rangle}(x,p) 
\notag \\
&\quad + \pi^2  \big( \Gamma^{JI}(x,p)  -  \overline{ \Gamma^{J}(x,p)  } \Gamma^{I}(x,p)   + \overline{ \Gamma^{JI}(x,p)   } \big) .\notag
\end{align}
Together with the results for the coefficients $\XX_M{}^{IJ}$ in 
section \ref{sec:B.3} below, the expressions (\ref{appbb.19}) determine all
instances of the key ingredients (\ref{dhse.u}) at $r\leq 2$, e.g.
\begin{align}
 h^{I }{}_J(x,p) &=   f^{I }{}_J(x,p)  
 + \TT^{I }(x,p)  \omega_J(x),  \label{appbb.20} \\
 h^{I_1 I_2}{}_J(x,p) &=   f^{I_1 I_2  }{}_J(x,p)  
 + \TT^{I_1}(x,p)   f^{ I_{2} }{}_J(x,p) 
 + \TT^{I_1 I_2 }(x,p) \omega_J(x).  \notag
\end{align}
This is an important step
towards expressing the Enriquez kernels $ g^{I_1 \cdots I_r}{}_J(x,p) $ 
at $r\leq 2$ in terms of $f$-tensors, Abelian
differentials and DHS polylogarithms.

\subsection{Determining the explicit form of the $\XX$-coefficients}
\label{sec:B.3}

We now proceed to determining the $\XX$-coefficients from the
$\mA$-monodromy condition (\ref{3.mon1}) which translates into the vanishing
\begin{align}
\TT^{I_1 \cdots I_r}(\mA^L \cdot p ,p) = 0 \, , \ \ \ \ r\geq 1 
\label{appbb.21} 
\end{align}
of the coefficients in the expansion (\ref{dhse.G}) at the special value
$x=p{+}\mA^L$. By their explicit form (\ref{appbb.19})
at $r\leq 2$ and for generic endpoints, the solution to (\ref{appbb.21}) will express
$\XX_M{}^{IJ}$ in terms of $\mA$-periods of iterated integrals of
Abelian differentials (\ref{appbb.12}) and antisymmetric kernels in (\ref{gammaang}).
In view of the similar $\mB$ periods to be encountered in section \ref{sec:B.5} below, 
we introduce the shorthand notation
\bea
 \alpha^{ L \langle IJ \rangle}(p)  & = &  \Gamma^{\langle IJ \rangle}( \mA^L \cdot p ,p) 
 \no \\ 
  \beta_L{}^{ \langle IJ \rangle}(p)  & = &  \Gamma^{\langle IJ \rangle}(\mB_L \cdot p ,p) 
 \label{appbb.22} 
\eea
and more generally write the special values $x=\mA^L \cdot p $ or $x=\mB_L \cdot p $ of the iterated Abelian integrals
in (\ref{appbb.12}) as follows
\bea
\alpha^{L}{}_{ I_1 \cdots I_r}(p) & = & \Gamma_{I_1 \cdots I_r}(\mA^L \cdot p ,p)  
\no \\ 
\beta_{L | I_1 \cdots I_r}(p) & =  & \Gamma_{I_1 \cdots I_r}(\mB_L \cdot p ,p) 
 \label{appbb.23} 
\eea
with the usual raising of indices via $Y^{JK}$ to $\alpha^{L  | I_1 \cdots I_r}(p) = \Gamma^{I_1 \cdots I_r}(\mA^L \cdot p ,p) $. While the periods (\ref{permat}) lead to simple, $p$-independent expressions at $r=1$,
\bea
\alpha^{L}{}_{ I}  = \delta^L_I  \, , \ \ \ \  \beta_{L | I } = \Omega_{LI}
\label{appbb.24} 
\eea
equivalent to $\alpha^{L | I} = Y^{LI}  \, , \  \beta_{L }{}^I = Y^{IK} \Omega_{KL}$,
generic instances of (\ref{appbb.23}) depend non-trivially on $p$ as exemplified by
$\partial_p \alpha^{L}{}_{ IJ }(p) = \omega_I(p) \delta^L_J -  \omega_J(p) \delta^L_I $.
The shuffle relations (\ref{1.shuffle}) straightforwardly propagate to (see (\ref{appbb.02}) 
and (\ref{appbb.in}) for $\shuffle$)
\bea
\alpha^{L}{}_{I_1 \cdots I_r \shuffle    J_1 \cdots J_s }(p) & = & 
\alpha^{L}{}_{I_1 \cdots I_r  }(p) \, \alpha^{L}{}_{ J_1 \cdots J_s }(p)  
\no \\
\beta_{L | ( I_1 \cdots I_r\shuffle J_1 \cdots J_s  )}(p) & = &  
\beta_{L |  I_1 \cdots I_r}(p) \, \beta_{L |  J_1 \cdots J_s   }(p)  
 \label{appbb.25} 
\eea
and reduce certain permutation sums of (\ref{appbb.23}) in the $I_j$ to
combinations of (\ref{appbb.24}). For instance, the fact that both of 
$\alpha^{L}{}_{ IJ}(p)+ \alpha^{L}{}_{J I}(p)$ and its complex conjugate simplify 
to $\delta^L_I \delta^L_J$ implies the antisymmetry of
 $\Im \alpha^{L}{}_{ IJ}(p) = - \Im \alpha^{L}{}_{ JI }(p)$. 
 
 \sm 
 
In this setting, we can express the contribution of $\TT^{IJ}(x,p)$ in (\ref{appbb.19}) to the $\mA$-monodromy of 
${\cal U}_{\rm DHS}(x,p;\xi,\eta)^{-1}$ as
\begin{align}
\TT^{IJ}(\mA^L \cdot p ,p) &= {-}\XX^{LIJ}
- \pi  \alpha^{ L \langle IJ \rangle}(p)  
 + \pi^2 \big( \alpha^{L | JI}(p) - Y^{LJ} Y^{LI} + \overline{ \alpha^{L | JI} (p) } \big).
 \label{appbb.26} 
\end{align}
Hence, the monodromy condition $\TT^{IJ}(\mA^L \cdot p ,p) =0$ determines
\bea
\XX^{LIJ} =
- \pi  \alpha^{ L \langle IJ \rangle}(p)  
 +\frac{\pi^2}{2} \big( \alpha^{L | JI}(p) + \overline{ \alpha^{L | JI}  }(p)  - (I \leftrightarrow J)\big),
 \label{appbb.27} 
\eea
where we have rewritten $\alpha^{L | JI}(p)  = \frac{1}{2}( \alpha^{L | JI}(p) - \alpha^{L | IJ}(p) + Y^{LJ} Y^{LI} )$
and similarly for its complex conjugate to expose the antisymmetry $\XX^{LIJ} =  -\XX^{LJI}$.
Even though individual terms on the right side of (\ref{appbb.27}) depend on $p$, one can verify
that $\partial_p \XX^{LIJ} = \partial_{\bar p} \XX^{LIJ}  = 0$
by combining the derivatives in (\ref{dhse.18a}) below.
Note that both of  $ \alpha^{ L \langle IJ \rangle}(p) $ and
$ \beta_L{}^{ \langle IJ \rangle}(p)$ in (\ref{appbb.22}) are antisymmetric
in $I,J$ since $\Gamma^{ \langle IJ \rangle}(x,p) $ is.

\sm

Inserting the result (\ref{appbb.27}) for $\XX^{LIJ}$ into the expressions (\ref{appbb.19}) 
makes the intermediate objects
$ h^{I_1 \cdots I_r}{}_J(x,p)$ in (\ref{dhse.u}) for the relations between Enriquez kernels
$ g^{I_1 \cdots I_r}{}_J(x,p)$ and $f$-tensors fully explicit for $r\leq 2$.
 
\subsection{Implementing the automorphism}
\label{sec:B.4} 

The next step is to express the Enriquez kernels $ g^{I_1 \cdots I_r}{}_J(x,p)$
in terms of the above $ h^{I_1 \cdots I_r}{}_J(x,p)$ and the expansion coefficients
$ \LL_I{}^{ J_1 \cdots J_r K}(p)$ and $\MM^{I I_1 \cdots I_r} {}_J (p) $
 in (\ref{exphatb}) and (\ref{exphatb1}) that implement the automorphism $a \cup b \to \ta \cup \tb$.
 This is most conveniently done after eliminating one of the infinite families
 of coefficients $\MM$ or $\LL$ in terms of the other by exploiting the matching
 (\ref{3.res-m}) of residues in the connections (\ref{3.JKa}).
 
 \sm
 
The order-by-order computations amount to inserting the expansions (\ref{exphatb}) and (\ref{exphatb1})
into $[b_I, a^I] =   [\eta_I, \ta^I - \xi^I]  $ and bringing all the nested brackets into
the standard form $B_{I_1}\cdots B_{I_r} a^J$ by exhaustive use of Jacobi identities, 
for instance $[ [b_I, b_J], a^K ] = [ b_I,  [ b_J, a^K ] ] $ $- [ b_J,  [ b_I, a^K ] ]$.
 Matching the coefficients of $B_{I_1}\cdots B_{I_r} a^J$ in (\ref{3.res-m}) then 
 implies identities such as 
 \begin{align}
\LL_I{}^{JK} - \LL_I{}^{KJ} &= -\MM^{JK}{}_I  , 
\label{defshu}\\
\LL_I{}^{JKP} - \LL_I{}^{JPK}-  \LL_I{}^{PJK}+  \LL_I{}^{PKJ} &= \MM^{JK}{}_Q \MM^{QP}{}_I -\MM^{JKP}{}_I .
\notag
\end{align}
Their generalization to higher order is most conveniently obtained by iterative use of (\ref{3.HB})
and leads to the closed formula (\ref{mshuf}) for the coefficients $\MM_\shuffle^{I I_1 \cdots I_r} {}_J (p) $ of the
reorganized expansion (\ref{3.bsh}) of $\eta_I$. By the shuffle symmetries
\bea
\LL_I{}^{  J_1 \cdots J_r \shuffle  K_1 \cdots K_s } = 0 
\, , \ \ \ \
r,s \geq 1
 \label{appbb.ree}
\eea
following from Ree's theorem and $\eta_I \in \mg$, we have
\bea
r \LL_I{}^{  J_1 \cdots J_r  } = - 
\MM_{\shuffle}^{J_1 \cdots J_r}{}_I \, , \ \ \ \ r\geq 2,
 \label{lvsmshuff}
\eea
where (\ref{defshu}) for instance identifies $\MM_{\shuffle}^{I_1 I_2 }{}_K
= \MM^{I_1 I_2 }{}_K $ and
 \bea
\MM_{\shuffle}^{I_1 I_2 I_3}{}_K
= \MM^{I_1 I_2 I_3}{}_K -  
   \MM^{I_1 I_2}{}_L  \MM^{L I_3}{}_K .
   \label{defshu.a}
 \eea
With the relations in (\ref{defshu}), one can now express the
Enriquez kernels $g^{I_1\cdots I_r}{}_J(x,p)$ solely in terms of the
expansion coefficients $h^{I_1\cdots I_r}{}_J(x,p)$ in (\ref{dhse.u}) and the
quantities $\MM^{I I_1 \cdots I_r} {}_J (p) $: Inserting all of (\ref{exphatb}),
(\ref{exphatb1}) and (\ref{defshu}) into (\ref{3.JKa}) and comparing the coefficients of
$B_{I_1}\cdots B_{I_r} a^J$ determines
 \bea
g^{I}{}_J(x,p) &=& h^{I}{}_J(x,p) + \omega_K(x) \MM^{KI}{}_J(p),
\label{gtoh.1} 
\no  \\
g^{I_1 I_2}{}_J(x,p) &= & h^{I_1 I_2}{}_J(x,p) 
  +  h^{I_1 }{}_K(x,p) \MM^{KI_2}{}_J(p)
 \no \\ &&
 -   h^{K }{}_J(x,p) \MM^{I_1 I_2}{}_K(p)
 +  \omega_K(x) \MM^{KI_1 I_2}{}_J(p).
  \eea
 Now the only missing piece of information in the above relations between
$g^{I_1\cdots I_r}{}_J$ and $h^{I_1\cdots I_r}{}_J$ is the
 explicit form of the $\MM^{J K_1 \cdots K_r}{}_I$ which will be determined next.

\subsection{Explicit form of the automorphism and Enriquez kernels}
\label{sec:B.5} 
 
The last step in our explicit construction of the Enriquez kernels is the
extraction of the coefficients $\MM^{J K_1 \cdots K_r}{}_I$ in the
automorphism (\ref{exphatb1}) from the $\mB$-monodromy condition 
in (\ref{3.mon1}). The latter is particularly suitable for order-by-order computations when rewritten
in the form
\begin{align}
e^{-2\pi i b_K} &=  {\cal U}_{\rm DHS}\big(\mB_K \cdot p ,p;\xi,\eta \big)^{-1}
= 1 + \sum_{r=1}^{\infty} \eta_{I_1} \cdots \eta_{I_r} \TT^{I_1 \cdots I_r}(\mB_K \cdot p ,p)
 \label{appbb.31}
\end{align}
that incorporates the expansion in (\ref{dhse.G}) with explicit
results for its coefficients $ \TT^{I_1 \cdots I_r}$ in (\ref{appbb.19}).
The desired $\MM^{J K_1 \cdots K_r}{}_I$ are determined by lining up
the Lie-algebra valued expansion variables $\eta_I$ and $b_K$ on the two sides
of (\ref{appbb.31}) which can be done in a variety of ways. We found it convenient
to invert the expansion of $\eta_I$ in (\ref{3.bsh}) while eliminating the coefficients
$\MM_{\shuffle}$ in favor of $\MM$ using (\ref{mshuf}), e.g.\footnote{The coefficient of 
$ \eta_{I_1} \eta_{I_2}  \eta_{I_3} \eta_{I_4}$ at the next order of the expansion (\ref{appbb.32}) of $b_J$ is given by 
$ \MM^{I_1 I_2 I_3 I_4}{}_J 
+ \MM^{I_1 I_2 K}{}_J   \MM^{I_3 I_4}{}_K  
+  \MM^{I_1 K}{}_J    \MM^{I_2 I_3 I_4}{}_K  
+   \MM^{I_1 K I_4}{}_J    \MM^{I_2 I_3}{}_K  
 +   \MM^{I_1K}{}_J   \MM^{ I_2 L}{}_K   \MM^{I_3 I_4}{}_L$ and does not line up with a
closed formula where the number of terms doubles at each order as in the coefficients (\ref{mshuf})
of the inverse expansion of $\eta_J$ in terms of $b_{I_1}\cdots b_{I_r}$.}
\bea
  \label{appbb.32}
 b_J & = & \eta_J + \eta_{I_1} \eta_{I_2} \MM^{I_1 I_2}{}_J(p) 
 \no \\ &&
 +  \eta_{I_1} \eta_{I_2} \eta_{I_3}
 \big( \MM^{I_1 I_2 I_3}{}_J(p)  + \MM^{I_3 I_2}{}_K(p)  \MM^{KI_1}{}_J(p) 
 \big)+ {\cal O}(\eta^4),
\eea
where the coefficient of $  \eta_{I_1} \eta_{I_2} \eta_{I_3}$ is another shuffle-symmetric combination
different from $ \MM_{\shuffle}^{I_1 I_2 I_3}{}_I(p) $ in (\ref{defshu.a}). After inserting the 
expansion (\ref{appbb.32}) into the exponentials of (\ref{appbb.31}),
\begin{align}
 e^{-2\pi i b_K} &= 1 - 2\pi i \eta_{I} \delta^I_K +   \eta_{I_1}  \eta_{I_2}
 \big( \tfrac{1}{2} (2\pi i)^2 \delta^{I_1 I_2}_K - 2\pi i  \MM^{I_1 I_2}{}_K(p) \big)   \label{appbb.33} 
 \\
 &\quad  +   \eta_{I_1}  \eta_{I_2} \eta_{I_3}
 \big( {-}\tfrac{1}{6} (2\pi i)^3 \delta^{I_1 I_2 I_3}_K
  + \tfrac{1}{2} (2\pi i)^2  \MM^{I_1 I_2}{}_K(p) \delta^{I_3}_K
  + \tfrac{1}{2} (2\pi i)^2  \delta^{I_1}_K  \MM^{I_2 I_3}{}_K(p) \notag \\
  &\quad\quad\quad\quad\quad\quad
  - 2\pi i \big[ \MM^{I_1 I_2 I_3}{}_K(p) - \MM^{I_2 I_3}{}_L(p)  \MM^{L I_1 }{}_K(p)  \big] \big)  + {\cal O}(\eta^4),
\notag
\end{align}
we can equate the unknown $\MM$-dependent coefficients of $ \eta_{I_1} \cdots \eta_{I_r}$ with
the computable quantities $\TT^{I_1 \cdots I_r}(p+\mB_K,p)$. Among the resulting conditions at
rank $r \leq 3$,
\begin{align}
\TT^{I}(\mB_K \cdot p ,p) &= -2\pi i \delta^I_K,
\label{appbb.34}  \\
\TT^{I_1 I_2}(\mB_K \cdot p ,p) &= \tfrac{1}{2} (2\pi i)^2 \delta^{I_1 I_2}_K - 2\pi i  \MM^{I_1 I_2}{}_K(p),
\notag
\\
\TT^{I_1 I_2 I_3}(\mB_K \cdot p ,p) &= {-}\tfrac{1}{6} (2\pi i)^3 \delta^{I_1 I_2 I_3}_K
  + \tfrac{1}{2} (2\pi i)^2  \MM^{I_1 I_2}{}_K(p) \delta^{I_3}_K
  + \tfrac{1}{2} (2\pi i)^2  \delta^{I_1}_K  \MM^{I_2 I_3}{}_K(p) \notag \\
  &\quad
  + 2\pi i   \MM^{I_2 I_3}{}_N(p)  \MM^{N I_1 }{}_K(p) 
    - 2\pi i  \MM^{I_1 I_2 I_3}{}_K(p) ,\notag
\end{align}
the first line is trivially satisfied since the leading order $ b_I = \eta_I +  {\cal O}(\eta^2)$ in
the expansions (\ref{exphatb}) and (\ref{appbb.32}) already takes our computations
in (\ref{3.xietaA}) into account. The second line of (\ref{appbb.34}) in turn can be solved
for $ \MM^{IJ}{}_K(p)$ in terms of the following $\mB$-periods
\begin{align}
\TT^{I J}(\mB_K \cdot p ,p) &= 
{-} \Omega_{KR} \XX^{RIJ}  
 - \pi   \beta_K{}^{ \langle IJ \rangle}(p) \label{appbb.35}  \\
 &\quad
+ \pi^2 \big[ \beta_K{}^{JI}(p)
- Y^{JR} \bar \Omega_{RK} Y^{IS} \Omega_{SK}
+ \overline{\beta_K{}^{JI}(p)}
\big]  \notag 
\end{align}
obtained from (\ref{appbb.19}) at $x= \mB_K \cdot p $,
with $\XX$ given by the $\mA$-periods (\ref{appbb.27}),
and the $\alpha,\beta$ notation introduced in (\ref{appbb.22}) as well
as (\ref{appbb.23}). One can view the second 
line of (\ref{appbb.34}) as providing both an expression for the
unknown coefficient $ \MM^{I J}{}_K(p)$ and as a crosscheck of the $\mB$-periods
(\ref{appbb.35}): the antisymmetry of the solution
\begin{align}
 \MM^{I J}{}_K(p) = i\pi  \delta^{IJ}_K + \frac{i}{2\pi }   \TT^{I J}(\mB_K \cdot p ,p) 
\label{appbb.36} 
\end{align}
in $I \leftrightarrow J$ derived from (\ref{defshu}) is not manifest term by term and relies on finding the symmetric part
$\TT^{I J}(\mB_K \cdot p ,p)+ \TT^{JI}(\mB_K \cdot p ,p)= (2\pi i)^2 \delta^{I J}_K$ in (\ref{appbb.35}). This can be verified using the consequence $ \beta_{K|IJ}(p) +  \beta_{K| JI}(p)=  \Omega_{KI}\Omega_{KJ}$ of the shuffle relations (\ref{appbb.25}), so the desired coefficient can be alternatively expressed as the antisymmetric part
 \begin{align}
 \MM^{I J}{}_K(p)&= \frac{i}{4\pi } \big[   \TT^{I J}( \mB_K \cdot p ,p) - \TT^{JI}(\mB_K \cdot p ,p) \big]
\label{appbb.37}  \\
&=  \frac{i}{2} \Big(  
\Omega_{KR} \alpha^{R \langle IJ \rangle}(p) -   \beta_K{}^{ \langle IJ \rangle}(p) \Big) \notag \\
&\quad +
 \frac{i \pi}{4} \Big(
\beta_K{}^{JI}(p)
- \Omega_{KR} \alpha^{R| JI }(p)
+ Y^{IR} \bar \Omega_{RK}   Y^{J S} \Omega_{SK} \notag \\
&\quad \quad \quad
+ \overline{\beta_K{}^{JI}(p)}
- \Omega_{KR} \overline{ \alpha^{R| JI }(p) }
 - (I \leftrightarrow J)
\Big), \notag
\end{align}
where the antisymmetrization prescription $- (I \leftrightarrow J)$ applies to the last two lines.

\sm

One can similarly solve the third equation of (\ref{appbb.34}) for $ \MM^{I_1 I_2 I_3}{}_K(p)$ in terms of $\mB$-monodromies $\TT^{I_1 I_2 I_3}(\mB_K \cdot p ,p)$ and lower-order expressions.
Again, the shuffle symmetry of $\MM^{I_1 I_2 I_3}{}_K(p) - \MM^{I_2 I_3}{}_L(p)  \MM^{L I_1 }{}_K(p)$ 
will not be obvious from the resulting expression but can be verified using a sequence of
shuffle relations (\ref{appbb.25}) of the $\mA$- and $\mB$ periods and their special cases
(\ref{appbb.24}). An alternative approach towards $\MM^{I_1 I_2 I_3}{}_K(p)$ is to project the
third equation of (\ref{appbb.34}) to those symmetry components with respect to the permutation
group of $I_j$ that eliminate all the admixtures of Kronecker deltas as done in (\ref{appbb.37}).
 
\subsection{Summary and further simplifications}
\label{sec:B.6}  
 
The explicit form of low rank Enriquez kernels may be obtained by combining (\ref{appbb.20}) with (\ref{gtoh.1}), resulting in equation (\ref{intro.41}) of the Introduction.  The simplest components $ \TT^{I}(x,p)$ and 
$\TT^{I_1 I_2}(x,p) $ of the gauge transformation are explicitly available by combining (\ref{appbb.19}) with the expression (\ref{appbb.27}) for $\XX^{KIJ}$. The coefficients $\MM^{KI_1 \ldots I_r}{}_J(p)$ of the automorphism are explicitly given in (\ref{appbb.37}) for $r=1$ and implicitly given in (\ref{appbb.34}) for $r=2$, resulting from the second- and third-order expansion of  $ {\cal U}_{\rm DHS}(\mB_K \cdot p ,p;\xi,\eta)^{-1}$, respectively.

\sm

The analogous expressions for $g^{I_1 \cdots I_r}{}_J(x,p)$ at higher order $r\geq 3$
require the contributions to the automorphism up to and including the
rank of $\MM^{KI_1 \ldots I_r}{}_J(p)$ which is computed from expansions of
$ {\cal U}_{\rm DHS}\big(\mB_K \cdot p ,p;\xi,\eta \big)^{-1}$ to the $(r+1)^{\rm th}$
order in $\eta_I$ or $b_I$. With the expression for $ \MM^{KI}{}_J(p)$ in (\ref{appbb.37}), we
have the fully explicit form of the Enriquez kernel $g^I{}_J(x,p)$,
and the solution $\MM^{KI_1 I_2}{}_J(p)$ of the last equation in 
(\ref{appbb.34}) completely fixes $g^{I_1 I_2}{}_J(x,p)$ in (\ref{intro.41}).

\sm

However, we have not yet attempted an exhaustive simplification of the $\mA$- and $\mB$-periods appearing in the expressions for  $\MM^{KI_1 \cdots I_r}{}_J(p)$ due to the procedure in section \ref{sec:3.6}. The expression (\ref{appbb.37}) for $ \MM^{KI}{}_J(p)$ admits a more minimal form that can be anticipated from the differential equations
\begin{align}
\partial_p \MM^{IJ}{}_K(p) &= \pi \big(  \delta^J_K \omega^I(p) - \delta^I_K \omega^J(p)  \big),
\label{dhse.20}  \\
\partial_{\bar p} \MM^{IJ}{}_K(p) &= \pi \big(  \delta^I_K  \bar \omega^J(p)  -  \delta^J_K  \bar \omega^I(p)  \big),
\notag
\end{align}
which follow from
\begin{align}
\partial_p \alpha^{K \langle IJ  \rangle}(p) &= \pi \big( \omega^J(p) Y^{KI}  -  \omega^I(p) Y^{KJ} \big),
\label{dhse.18a}  \\
\partial_{\bar p}  \alpha^{K \langle IJ  \rangle}(p)&=  \pi \big( \bar \omega^J(p) Y^{KI}  -  \bar \omega^I(p) Y^{KJ} \big),
\notag \\
\partial_p \alpha^{K|  IJ  }(p) &=  \omega^I(p) Y^{JK} {-}\omega^J(p) Y^{IK}, \notag
\end{align}
as well as
\begin{align}
\partial_p \beta_K{}^{\langle IJ \rangle}(p) &= \pi  \big( \omega^J(p) Y^{IR}  - \omega^I(p) Y^{JR}  \big) \bar \Omega_{RK} ,
\label{dhse.18b}  \\
\partial_{\bar p} \beta_K{}^{\langle IJ \rangle}(p)&=  \pi \big(  \bar \omega^J(p) Y^{IR} - \bar \omega^I(p) Y^{JR}  \big) \Omega_{RK} ,
\notag  \\
\partial_p \beta_K{}^{  IJ  }(p) &=  \big(\omega^I(p) Y^{JR} {-}\omega^J(p) Y^{IR} \big) \Omega_{RK}.  \notag
\end{align}
Matching the $p$- and $\bar p$-derivatives of (\ref{dhse.20}) with those of $2\pi i \, \Im \alpha_K{}^{IJ}(p) $, we establish %
\bea
\MM^{IJ}{}_K(p) = 2\pi i \, \Im \alpha_K{}^{IJ}(p) + \notc^{IJ}{}_K,
\label{dhse.21} 
\eea
where $\notc^{IJ}{}_K = - \notc^{JI}{}_K$ is independent on $p$ but may depend non-meromorphically
on the moduli of $\Sigma$. Hence, the final form of the first non-trivial Enriquez
kernel to be provided in this work is given by
\begin{align}
 g^I{}_J(x,p)= f^I{}_J(x,p) - 2\pi i \, \omega_J(x)  \Im \! \int^x_p \! \omega^I  
 + \omega_K(x) \Big (  2\pi i \,  \Im \alpha_J{}^{KI}(p) +  \notc^{KI}{}_J \Big ) ,
\label{gijfinal} 
\end{align}
which follows from (\ref{intro.41}) and (\ref{dhse.21}). While meromorphicity
in $x$ is a consequence of the absence of $(0,1)$-form components in (\ref{3.KJ}),
meromorphicity of $ g^I{}_J(x,p)$ in $p$ relies on the interplay of the first three terms in
(\ref{gijfinal}) and is guaranteed by the uniqueness of the Enriquez kernels based on the
defining properties of ${\cal K}_{\rm E}$.
We leave obtaining a direct derivation of (\ref{dhse.21}), an
explicit formula for $\notc^{IJ}{}_K$, and the generalizations to simplify higher-rank
contributions $\MM^{KI_1 \ldots I_r}{}_J(p)$ to the automorphism to future work.

\end{document}

\end{document}